\documentclass[a4paper,12pt]{article}
\usepackage[T1]{fontenc}
\pdfoutput=1

\usepackage{yfonts}
\usepackage{comment}
\usepackage{jheppub}
\usepackage{longtable}
\usepackage{array}
\usepackage{macro}
\usepackage{appendix}
\usepackage{amsmath,amssymb}
\usepackage{tikz}
\usepackage{pgfkeys}
\usepackage{pgfopts}
\usepackage{xcolor}
\usepackage{pict2e}
\usepackage{young}
\usepackage{youngtab}
\usepackage{ytableau}
\usetikzlibrary{decorations.pathreplacing}

\usepackage{tikz}

\newcommand{\vol}{\mathrm{vol}}
\newcommand{\inner}[2]{\left\langle #1,#2\right\rangle}
\newcommand{\Ltwo}[2]{\langle #1,#2\rangle}

\newtheorem{Theorem}{Theorem}[section]

\newtheorem{Proposition}[Theorem]{Proposition}

\makeatletter
\DeclareRobustCommand{\loplus}{\mathbin{\mathpalette\dog@lsemi{+}}}
\DeclareRobustCommand{\lotimes}{\mathbin{\mathpalette\dog@lsemi{\times}}}
\DeclareRobustCommand{\roplus}{\mathbin{\mathpalette\dog@rsemi{+}}}
\DeclareRobustCommand{\rotimes}{\mathbin{\mathpalette\dog@rsemi{\times}}}

\newcommand{\dog@rsemi}[2]{\dog@semi{#1}{#2}{-90,90}}
\newcommand{\dog@lsemi}[2]{\dog@semi{#1}{#2}{270,90}}
\newcommand{\dog@semi}[3]{%
  \begingroup
  \sbox\z@{$\m@th#1#2$}%
  \setlength{\unitlength}{\dimexpr\ht\z@+\dp\z@\relax}%
  \makebox[\wd\z@]{\raisebox{-\dp\z@}{%
    \begin{picture}(1,1)
    \linethickness{\variable@rule{#1}}
    \roundcap
    \put(0.5,0.5){\makebox(0,0){\raisebox{\dp\z@}{$\m@th#1#2$}}}
    \put(0.5,0.5){\arc[#3]{0.5}}
    \end{picture}%
  }}%
  \endgroup
}
\newcommand{\variable@rule}[1]{%
  \fontdimen8  
  \ifx#1\displaystyle\textfont3\else
    \ifx#1\textstyle\textfont3\else
      \ifx#1\scriptstyle\scriptfont3\else
        \scriptscriptfont3\relax
  \fi\fi\fi
}
\makeatother

\title{\textbf The asymptotic charges of Curtright dual graviton and Curtright extensions of $\mathfrak{BMS}$ algebra}

\author[a,b]{Federico Manzoni}

\affiliation[a]{Mathematics and Physics department, Roma Tre, Via della Vasca Navale 84, Rome, Italy}
\affiliation[b]{INFN Roma Tre Section, Physics department, Via della Vasca Navale 84, Rome, Italy}
\emailAdd{federico.manzoni@uniroma3.it, federico13.manzoni97@gmail.com, ORCID ID: 0000-0002-9979-6154}

\abstract{This paper studies the asymptotic gauge charges of the Curtright mixed-symmetry rank-3 field \( \phi_{[\rho\sigma]\nu} \) in Minkowski spacetime, interpreted in \( D = 5 \) as the dual graviton. In Bondi coordinates at future null infinity, we impose radiation fall-offs and fix a de Donder-like gauge together with an on-shell traceless condition, similarly to what happens in linearized gravity. Surface charges associated with the residual gauge transformations are constructed as boundary integrals via N\"other's 2-form. In \( D = 5 \), exploiting Hodge/Hodge-like decompositions on \( S^{3} \), the charge splits into a scalar sector \( Q_{\Phi} \), a vector sector \( Q_{V} \) and a TT sector $Q_{y^{\text{TT}}}$. \( Q_{\Phi} \) is parametrized by a single arbitrary scalar function \( \Phi \) (interpreted as the supertranslation-like parameter), \( Q_{V} \) is parametrized by a vector field \( V^{i} \in \mathfrak{Diff}(S^{3}) \) and the TT sector $Q_{y^{\text{TT}}}$ is parametrized by a transverse-traceless rank-2 tensor $y_{ij}^{\text{TT}} \in \mathfrak{TT}(S^3)$. The corresponding charge algebra closes only if $V_i \in \mathfrak{o}(4)$ as semidirect sum \( \mathfrak{o}(4) \loplus (C^{\infty}(S^3) \oplus \mathfrak{TT}(S^3)) \), i.e. an abelian extension of a $\mathfrak{BMS}$-like algebra featuring a higher-spin-like supertranslation sector.}

\begin{document} 
\maketitle
%\flushbottom
%\newpage

\section{Introduction}
The study of asymptotic symmetries has long provided a bridge between gauge redundancy and genuine physical information encoded at infinity \cite{Ciambelli:2022vot, strominger2018lectures, Manzoni:2025gxw, Heissenberg:2019fbn, Compere:2018aar}. In gravity, the analysis of asymptotically flat space-times led to the discovery that the symmetry group at null infinity is larger than the Poincar\'e group, giving rise to the Bondi-Metzner-Sachs-van der Burg (BMS) group and, in particular, to an infinite-dimensional family of angle-dependent translations \cite{Bondi:1962px,PhysRev.128.2851,Sachs:1962wk}. In the last decade this subject has acquired a renewed centrality, largely due to the realization that asymptotic symmetries, soft theorems and memory effects are different facets of the same infrared structure of gauge theories \cite{Pasterski:2015zua, Jokela:2019apz, Afshar_2019, Pate_2017, Strominger_2014, Francia:2018jtb, Hamada_2017, deAguiarAlves:2025vfu, Esmaeili:2020eua}. These developments have stimulated a systematic re-examination of the precise role played by boundary conditions, gauge choices and the definition of conserved charges, not only in four-dimensional gravity \cite{Barnich:2011ChargeAlgebra, Barnich:2010BMSsuperrotations, Campiglia:2014SubleadingSoft, Campiglia:2015NewSymmetriesSmatrix, Kapec:2014VirasoroSmatrix, Freidel:2021BMSW}   but also in higher dimensional gravity \cite{TanabeTanahashiShiromizu2010NullInfinity5D,TanabeKinoshitaShiromizu2011NullInfinityAnyD, KapecLysovPasterskiStrominger2017HigherDSupertranslations,HollandsIshibashi2005BondiEnergyHigherD, Fuentealba:2022BMS5PRL, Fuentealba:2022HamiltonianBMS5}, different backgrounds \cite{Kehagias_2016, Tolish:2016ggo,Bonga:2020fhx,Manzoni:2024agc, Ferreira:2016hee}, higher-spin gauge theories \cite{Campoleoni:2019ptc, Campoleoni:2020ejn, Campoleoni:2017mbt, Campiglia:2015qka, Campoleoni:2017qot, Strominger_2014} and for more exotic gauge systems \cite{Afshar:2018apx, Romoli:2024hlc, Ferrero:2024eva,Francia:2018jtb, Romoli:2025map, manz2, Manzoni:2025zmi, Manzoni:2024tow}.

A key lesson emerging from this modern perspective is that the asymptotic symmetry group is not a purely kinematical datum but depends delicately on which fall-off conditions one imposes on the fields. In $D>4$ gravity, for instance, there is a well-known tension between boundary conditions tailored to accommodate radiative degrees of freedom and those that allow for an enlargement of the asymptotic symmetry group beyond Poincar\'e \cite{TanabeTanahashiShiromizu2010NullInfinity5D, TanabeKinoshitaShiromizu2011NullInfinityAnyD, Fuentealba:2022BMS5PRL, Fuentealba:2022HamiltonianBMS5}. Different choices can lead to different residual gauge algebras and different sets of finite, non-trivial charges, while still capturing physically meaningful radiation. Accordingly, one natural strategy is to fix a gauge that reduces the equations of motion to a simple wave equation, impose radiation-compatible fall-offs motivated by flux finiteness, solve explicitly for the residual gauge parameters compatible with these conditions and then construct the associated surface charges \cite{Ruzziconi:2019pzd, Ciambelli:2022vot}.

The same circle of ideas extends well beyond the metric formulation of gravity. In fact, already at the free-field level, massless gauge fields can be formulated in terms of Lorentz tensors \cite{Bekaert:2006py} of various symmetry types, including representations described by non-trivial Young diagrams \cite{Curtright:1980un, Curtright:1980yk, Curtright:1980yj, Curtright:2019yur, CURTRIGHT1985304, Labastida:1986gy, Labastida:1987kw}. Such mixed symmetry fields arise naturally in higher-spin theory \cite{vukovic2018, Vasiliev_2004, witfre, sagnotti2012notes}, in the spectrum of string theory \cite{Polchinski:1998rq, Polchinski:1998rr, green_schwarz_witten_2012, Green:2012pqa} and in duality-covariant formulations of gravity and gauge theories \cite{Hull:2000zn, Hull:2001iu, Medeiros_2003, Manzoni:2025zmi}. Their covariant description requires a refined notion of gauge symmetry, typically involving more than one gauge parameter and, crucially, a hierarchy of gauge-for-gauge redundancies. These features make the analysis of asymptotic symmetries particularly interesting and potentially subtle. From this standpoint, mixed-symmetry fields provide a controlled arena in which one can test how robust the standard infrared story is under changes of field variables and under duality transformations.

A paradigmatic example is the Curtright field, the simplest ``hook'' mixed-symmetry tensor, which in the index convention used in this paper is written as \( \phi_{[\rho\sigma]\nu} \) and corresponds to the Young tableau \( \lambda=(2,1) \). The Curtright model generalizes the notion of gauge field beyond totally symmetric and totally antisymmetric tensors and comes equipped with two gauge parameters, one symmetric and one antisymmetric, together with a gauge-for-gauge vector parameter \cite{CURTRIGHT1985304}. In five space-time dimensions, the Curtright field is of particular interest because it provides a dual description of the graviton: on-shell, the two formulations encode the same propagating degrees of freedom, while off-shell their gauge structures look quite different \cite{hamermesh1989group, NicolasBoulanger_2003}. Understanding how asymptotic charges behave under such dualities is therefore a natural question, both for clarifying the physical meaning of dual observables and for exploring whether the infrared symmetry algebra is invariant under changes of field variables and dual descriptions. For works concerning global magnetic dual charges for the graviton and its symmetry properties also related to generalized symmetry in linearized gravity, see \cite{Hull:2023iny,Hull:2024qpy,2024CQGra..41s5012H,Hull:2024xgo,Hull:2024bcl,Hull:2024ism}. 

More broadly, the present work may be viewed as part of a wider line of research in which non-standard gauge structures, dual formulations and asymptotic symmetry algebras are used as probes of gravitational physics beyond the conventional metric description. In particular, three-dimensional higher-spin Chern-Simons gravities provide a distinguished arena in which gauge symmetry is enlarged in a controlled way and where the asymptotic symmetry analysis leads to infinite-dimensional boundary algebras of $W$-type or to their flat-space higher-spin analogues \cite{Campoleoni:2010zq,Henneaux:2010xg,Campoleoni:2011hg,Gonzalez:2013oaa}. These examples have made especially clear that the physical content of the asymptotic symmetry group depends delicately on the choice of boundary conditions, gauge fixing and phase-space variables. Related questions also arise in colored gravity models, where gravitational degrees of freedom are supplemented by non-abelian internal structures, leading to a richer interplay between bulk gauge symmetries and boundary symmetry algebras \cite{Joung:2017hsi,Monjo:2025u13}. In parallel, the study of extensions of the Bondi-Metzner-Sachs-van der Burg group in three and four dimensions has shown that even for ordinary gravity the asymptotic symmetry problem is considerably subtler than the original BMS analysis might suggest, with supertranslations, superrotations and their associated charge algebras depending sensitively on the adopted asymptotic framework \cite{Barnich:2010ojg,Barnich:2011ChargeAlgebra, Barnich:2010BMSsuperrotations,Barnich:2014kra,Freidel_2021weyl,Campiglia_2020}. From a different but closely related perspective, dual-graviton constructions and duality-symmetric formulations of linearized gravity have long suggested that the same propagating degrees of freedom may admit inequivalent off-shell gauge descriptions, with potentially different natural boundary observables \cite{Bunster:2013oaa,Hull:2001iu,Hull:2023iny,Hull:2024qpy}. Likewise, exceptional field theory approaches replace ordinary spacetime covariance by generalized diffeomorphisms and enlarged tensor hierarchies, thereby offering a broader setting in which duality covariance and generalized gauge symmetry become structural principles \cite{Hohm:2013pua,Hohm:2013vpa,Hohm:2013uia,Hohm:2014fxa,West:2014xma}. Although these programs are conceptually distinct, they share the common lesson that enlarged gauge structures and boundary symmetry algebras may encode information that is not manifest in the standard metric formulation and may therefore provide useful insight into the infrared and, more speculatively, the quantum structure of gravity. In this perspective, the Curtright field offers a particularly interesting testing ground. In five spacetime dimensions it furnishes the dual description of the graviton, while at the same time exhibiting a mixed-symmetry gauge structure that is substantially richer than that of the ordinary metric field. Studying its asymptotic charges is therefore a natural way to ask to what extent asymptotic symmetry data are preserved, modified or reorganized under dual formulations of gravity.

The aim of this paper is to carry out a detailed analysis of asymptotic charges for the Curtright field at null infinity, in a framework designed to make contact with the gravitational BMS story in \( D = 5 \). The starting point is considering Bondi patch of Minkowski space-time with coordinates adapted to future null infinity, together with radiation fall-offs motivated by finiteness and non-vanishing of the energy flux. We then fix a de Donder-like gauge condition, supplemented by an on-shell traceless condition in order to isolate the irreducible on shell representation, and thereby reduces the equations of motion to a free wave equation for \( \phi_{[\rho\sigma]\nu} \). This is very similar to what happens in linearized gravity and this prescription could be extended to all mixed symmetry tensors. With these gauge conditions in place, the residual gauge parameters are constrained by differential equations and by the requirement that they preserve the chosen fall-offs. To solve this coupled system in a uniform way across dimensions, the paper assumes polyhomogeneous expansions for the gauge parameters, allowing both half-integer powers and logarithmic terms.

Having identified the relevant residual transformations, the next step is the construction of the corresponding asymptotic charges. The resulting N\"oether charges are surface integrals over \( S^{D-2} \) at fixed retarded time and are expressed in terms of the leading boundary data of the (partial) field strength components with \( ru \)-indices, giving us an electric-like charge. In this respect, the mechanism is closely analogous to what happens in \( p \)-form gauge theories in Lorenz gauge \cite{Manzoni:2024tow}: a gauge condition removes potentially divergent components, the charge becomes electric-like in that it involves radial-null components of the field strength and only angular components of the gauge parameters survive in the final expression. This parallel is conceptually useful because it highlights how mixed-symmetry gauge fields fit into the broader pattern of infrared charges for massless gauge systems while at the same time stressing the new ingredients brought in by Young-projected tensors and by complex gauge-for-gauge redundancy hierarchy.

In \( D = 5 \) we can directly compare to the graviton because of duality. Exploiting Hodge and Hodge-like decompositions on \( S^{3} \) for the symmetric and antisymmetric gauge parameters, the Curtright charge naturally splits into a scalar sector \( Q_{\Phi} \), a vector sector \( Q_{V} \) and a TT sector \( Q_{y^{\text{TT}}} \). The scalar sector is parametrized by a single arbitrary function \( \Phi \) on \( S^{3} \), which is interpreted as the analogue of the supertranslation parameter while the vector sector is parametrized by a vector field \( V^{i} \) that could play the role of superrotations in this setting. However, the presence of the TT sector, interpretable as a higher spin supertranslation imposes $V^i \in \mathfrak{o}(4)$, i.e. a Killing vector, to ensure closure of the algebra. The resulting charge algebra closes as the semidirect sum \( \mathfrak{o}(4) \loplus (C^{\infty}(S^3) \oplus \mathfrak{TT}(S^3)) \) forming an abelian extension of a $\mathfrak{BMS}$-like algebra: the Curtright extension $\mathfrak{CBMS}(S^3)$. Notably, within the boundary conditions and asymptotic gauge-fixing adopted here, one finds only one independent scalar supertranslation-like parameter, rather than two scalar parameters as suggested by certain Hamiltonian treatments \cite{Fuentealba:2022HamiltonianBMS5}. Beyond the specific results for \( D = 5 \) , the analysis presented here fits into a broader programme of clarifying asymptotic symmetries for general massless fields in arbitrary dimensions. 

The paper is structured as follows. Section \ref{sec1}  reviews the Curtright \( (2,1) \) field, its gauge symmetries and gauge-for-gauge redundancy, and sets up the conventions used throughout. Section \ref{sec3} imposes radiation fall-offs in Bondi coordinates, implements the de Donder-like and traceless gauge fixing and solves the resulting residual-gauge system via polyhomogeneous expansions, isolating the gauge-parameter components relevant for finite charges. Section \ref{sec3} constructs the asymptotic charges from N\"oether's 2-form and derives a general expression valid in arbitrary dimension. Section \ref{sec4} specializes to \( D = 5 \) , performs the Hodge/Hodge-like decompositions of gauge parameters on \( S^{3} \) using the conditions $H^1_{\text{dR}}(S^3)=H^2_{\text{dR}}(S^3)=0$, derives the split \( Q = Q_{\Phi} + Q_{V} + Q_{y^{\text{TT}}}\), computes the charge algebra and discusses the comparison with graviton asymptotic charges. The appendices collect complementary material, including the explicit residual-gauge and fall-off equations, the Hodge-like decomposition technology for $(p,q)$-mixed symmetry tensors, technical commutation relations and asymptotic gauge-fixing steps.

\paragraph{Notation for traces, divergences and contractions.} We adopt the following shortcutting notations:
\begin{enumerate}
    \item for the gauge field $\phi^{[\alpha \beta]\nu}$
\begin{equation}
\begin{aligned}
    &\hat{\phi}^{\nu}:=\phi_{\alpha}^{\ \alpha \nu}=\eta_{\alpha \beta}\phi^{[\alpha \beta] \nu}\equiv0, \qquad \bar{\phi}^{\rho}:=\phi_{\alpha}^{\ \rho \alpha}=\eta_{\alpha \beta}\phi^{[\alpha \rho] \beta};\\
    &\hat{\partial} \cdot \phi^{\rho \nu}:=\partial_{\alpha}\phi^{\alpha \rho \nu}=\eta_{\alpha \beta}\partial^{\alpha} \phi^{[\beta \rho] \nu}, \qquad \bar{\partial} \cdot \phi^{\rho \sigma}:=\partial_{\alpha}\phi^{\rho \sigma \alpha}=\eta_{\alpha \beta}\partial^{\alpha}\phi^{[\rho \sigma] \beta},
    \label{notation}
\end{aligned}
\end{equation}
    \item for the symmetric gauge parameter $\lambda^{(\mu \nu)}$
\begin{equation}
    \lambda^\mu{}_\mu:=\eta_{\mu \nu} \lambda^{(\mu \nu)}, \qquad \partial \cdot \lambda^{\mu}:=\eta_{\mu \rho}\partial^{\rho}\lambda^{(\mu \nu)},
\end{equation}
    \item for the antisymmetric gauge parameter $\Lambda^{[\mu \nu]}$
\begin{equation}
    \Lambda^\mu{}_\mu:=\eta_{\mu \nu} \Lambda^{[\mu \nu]}\equiv 0, \qquad \partial \cdot \Lambda^{\mu}:=\eta_{\mu \rho}\partial^{\rho}\Lambda^{[\mu \nu]},
\end{equation}
    \item for the vector gauge-for-gauge parameter $\Theta^{\mu}$
    \begin{equation}
        \partial \cdot \Theta:=\eta_{\mu \rho}\partial^{\rho}\Theta^{\mu}.
    \end{equation}
\end{enumerate}
Similar notations hold for the lower indexes case.

\section{The Curtright three indexes field}\label{sec1}
In order to describe gauge fields the standard way is to use both totally symmetric and totally antisymmetric Lorentz tensors of arbitrary rank. Both types of gauge field tensors appear quite naturally in theoretical analyses of some interesting physical problems such as the description of massless particle fields of spin $s(s+\frac{1}{2})$ using symmetric gauge tensors (spinor-tensors) of rank $s$ or gauge fields naturally coupled to strings described by antisymmetric tensors. However, in 1980 Thomas Curtright generalize the concept of a gauge field to include higher rank Lorentz tensors which are neither totally symmetric nor totally antisymmetric under spacetime index permutations \cite{Curtright:1980un,Curtright:1980yj,Curtright:1980yk}. The basic simplest example is the hook field: a rank 3 tensor whose index permutation symmetry corresponds to the Young diagram $(2,1)$
\begin{equation}
\begin{ytableau}
  \none & & \\
  \none  & \\
\end{ytableau} \ \ .
\end{equation}
The fundamental step is to construct irreducible representations of the permutation group using symmetrizers and antisymmetrizers; in the following we assume to have applied the antisymmetrizer last, obtaining a field with manifest antisymmetry on the first two indexes $\phi_{[\rho \sigma]\nu}$. From the Young diagram of the Curtright three indexes field we can extract the gauge parameter removing one box so that we still have a Young diagram
\begin{equation}
    \begin{ytableau}
  \none & & \\
  \none  & \\
\end{ytableau} \ \ \ \Rightarrow  \begin{ytableau}
  \none & \\
  \none  & \\
\end{ytableau} \ \ \ \ ,  \begin{ytableau}
  \none & & \\
\end{ytableau} \ \ ;
\end{equation}
so we have one totally symmetric gauge parameter $\lambda_{(\mu \nu)}\equiv \lambda_{\mu \nu}$ and one totally antisymmetric gauge parameter $\Lambda_{[\mu \nu]} \equiv \Lambda_{\mu \nu}$. We can note that we can again eliminate a box and still get a Young diagram
\begin{equation}
\begin{aligned}
    &\begin{ytableau}
  \none & \\
  \none  & \\
\end{ytableau}  \ \ \ \Rightarrow  \begin{ytableau}
  \none & \\
\end{ytableau} \ \ ;\\
&\begin{ytableau}
  \none & & \\
\end{ytableau} \ \ \ \Rightarrow \begin{ytableau}
  \none & \\
\end{ytableau} \ \ ;
\end{aligned}
\end{equation}
this is called gauge-for-gauge redundancy: essentially the gauge parameters of a gauge field have their own gauge redundancy, here parametrized by a vector $\Theta_{\mu}$. The most obvious gauge transformation for the Curtright three indexes field is given by
\begin{equation}
    \phi'_{[\rho \sigma] \nu} = \phi_{[\rho \sigma ]\nu}+\partial_{\rho}\lambda_{(\sigma \nu)}-\partial_{\sigma}\lambda_{(\rho \nu)}-\partial_{\nu}\Lambda_{(\rho \sigma)},
\end{equation}
with
\begin{equation}
\begin{aligned}
&\Lambda'_{[\mu \nu]}=\Lambda_{[\mu \nu]}+\partial_{\mu}\Theta_{\nu}-\partial_{\nu}\Theta_{\mu}, \quad \lambda'_{(\mu \nu)}=\lambda_{(\mu \nu)}+3\partial_{\mu}\Theta_{\nu}+3\partial_{\nu}\Theta_{\mu},
\label{gfg}
\end{aligned}    
\end{equation}
but the irreducibility condition under $\mathrm{GL}(D-1,1)$, namely 
\begin{equation}
\phi_{[\rho \sigma]\nu}+\phi_{[\nu \rho]\sigma}+\phi_{[\sigma \nu ]\rho}=0,
\label{irregl}
\end{equation} 
requires adding the term $-2\partial_{\nu}\Lambda_{\sigma \rho}$ to the gauge redundancy of the field $\phi_{(\rho \sigma )\nu}$; therefore we have 
\begin{equation}
    \phi'_{[\rho \sigma] \nu} = \phi_{[\rho \sigma ]\nu}+\partial_{\rho}\lambda_{(\sigma \nu)}-\partial_{\sigma}\lambda_{(\rho \nu)}+\partial_{\rho}\Lambda_{[\sigma \nu]}-\partial_{\sigma}\Lambda_{[\rho \nu]}+2\partial_{\nu}\Lambda_{[\sigma \rho]}.
\end{equation}
We define a field strength
\begin{equation}\label{field}
    H_{[\alpha \beta \gamma]\mu}:=\partial_{\alpha}\phi_{[\beta \gamma]\mu}+\partial_{\beta}\phi_{[\gamma \alpha]\mu}+\partial_{\gamma}\phi_{[\alpha \beta]\mu},
\end{equation}
that is more precisely a "partial field strength" using the nomenclature of \cite{Manzoni:2025zmi} since it is not completely gauge invariant
\begin{equation}
    \delta_{\lambda, \Lambda} H_{[\alpha \beta \gamma]\mu}=-2\partial_{\nu}[\partial_{\alpha}\Lambda_{[\beta \gamma]}+\partial_{\beta}\Lambda_{[ \gamma \alpha]}+\partial_{\gamma}\Lambda_{[\alpha \beta ]}].
\end{equation}
However the contractions $H_{[\alpha \beta \gamma]\mu}H^{[\alpha \beta \gamma]\mu}$ and $H_{\alpha \beta }H^{\alpha \beta }$,
where ${\displaystyle H_{\alpha \beta }:=\eta ^{\gamma \mu }H_{\alpha \beta \gamma \mu }}$, can be combined into a gauge invariat lagrangian density
\begin{equation}
    {\displaystyle \mathcal{L}:=-{\frac {1}{6}}(H_{\alpha \beta \gamma \mu }H^{\alpha \beta \gamma \mu }-3H_{\alpha \beta }H^{\alpha \beta }).}
    \label{laden}
\end{equation}
The Euler-Lagrange equations obtained by varying the lagrangian density are
\begin{equation}
    E_{[\alpha \beta]\gamma}+\frac{1}{2}[\eta_{\alpha \gamma}E_{\beta}-\eta_{\beta \gamma}E_{\alpha}]=0,
\end{equation}
where 
\begin{equation}
    E_{[\alpha \beta] \gamma}:=\eta^{\mu \nu}\partial_{\mu}H_{[\nu \alpha \beta] \gamma}-\eta^{\mu \nu}\partial_{\gamma}H_{[\alpha \beta \mu]\nu}, \qquad E_{\alpha}:=\eta^{\mu \nu}E_{[\alpha \mu]\nu}=2\eta^{\mu \nu}\partial_{\nu}H_{\mu \alpha}.
    \label{eomfs}
\end{equation}
Here we prove some off-shell identities satisfied by the partial field strength \eqref{field}.

\begin{Proposition}[Off-shell identities of $H_{[\alpha\beta\gamma]\mu}$]
The partial field strength $H_{[\alpha\beta\gamma]\mu}$ defined in \eqref{field}
satisfies the off-shell identities 
    \begin{enumerate}[(a)]
        \item \begin{equation}
H_{[\alpha\beta\gamma]\mu}
- H_{[\mu\alpha\beta]\gamma}
+ H_{[\gamma\mu\alpha]\beta}
- H_{[\beta\gamma\mu]\alpha}
=0 ,
\label{lemma1}
\end{equation}
\item \begin{equation}
    \partial_{[\delta} H_{\alpha\beta\gamma]\mu} = 0,
\label{bianchi}
\end{equation}
\item \begin{equation}
    \partial^{\mu} H_{\alpha\beta\gamma\mu}
- \partial_{\alpha} H_{\beta\gamma}
- \partial_{\beta} H_{\gamma\alpha}
- \partial_{\gamma} H_{\alpha\beta}
= 0;
\label{bianchicon}
\end{equation}
\end{enumerate}
\end{Proposition}

\begin{proof}
For the identity $\textit{(a)}$, let us use the irreducibility condition for
$\phi_{[\beta\gamma]\mu},\phi_{[\gamma\alpha]\mu},\phi_{[\alpha\beta]\mu}$, namely
\begin{equation*}
\begin{aligned}
&\phi_{[\beta\gamma]\mu}+\phi_{[\mu\beta]\gamma}+\phi_{[\gamma\mu]\beta}=0,\\
&\phi_{[\gamma\alpha]\mu}+\phi_{[\mu\gamma]\alpha}+\phi_{[\alpha\mu]\gamma}=0,\\
&\phi_{[\alpha\beta]\mu}+\phi_{[\mu\alpha]\beta}+\phi_{[\beta\mu]\alpha}=0,
\end{aligned}
\end{equation*}
to rewrite \eqref{field} as
\begin{equation*}
\begin{aligned}
H_{[\alpha\beta\gamma]\mu}
=
\partial_{\alpha}(-\phi_{[\gamma\mu]\beta}-\phi_{[\mu\beta]\gamma})
+\partial_{\beta}(-\phi_{[\alpha\mu]\gamma}-\phi_{[\mu\gamma]\alpha})+\partial_{\gamma}(-\phi_{[\beta\mu]\alpha}-\phi_{[\mu\alpha]\beta}) .
\end{aligned}
\end{equation*}
Using the definition of $H_{[\gamma\mu\alpha]\beta}$ and
$H_{[\alpha\mu\beta]\gamma}$ we have
\begin{equation*}
\begin{aligned}
H_{[\alpha\beta\gamma]\mu}
=
- H_{[\gamma\mu\alpha]\beta}
- H_{[\alpha\mu\beta]\gamma}
+\partial_{\mu}(\phi_{[\beta\alpha]\gamma}+\phi_{[\alpha\gamma]\beta})-\partial_{\beta}\phi_{[\mu\gamma]\alpha}
-\partial_{\gamma}\phi_{[\beta\mu]\alpha}.
\end{aligned}
\end{equation*}
Thanks to the irreducibility condition of $\phi_{[\beta\alpha]\gamma}$ and the definition of
$H_{[\gamma\beta\mu]\alpha}$ we get
\begin{equation*}
H_{[\alpha\beta\gamma]\mu}
=
- H_{[\gamma\mu\alpha]\beta}
- H_{[\alpha\mu\beta]\gamma}
- H_{[\gamma\beta\mu]\alpha},
\end{equation*}
rearranging and using antisymmetry properties we get \eqref{lemma1}.
The identity $\textit{(b)}$ is a consequence of the antisymmetry property of $\phi_{[\alpha \beta]\mu}$ and the commutation of partial derivatives while the identity $\textit{(c)}$ follows from identity $\textit{(b)}$ by contracting on ($\delta, \mu$). 
\end{proof}

In order to discuss asymptotic symmetries at future null infinity of Bondi patch of Minkowski space-time, we introduce Bondi coordinates $(u,r,\{x^i\})$ where $u:=t-r$ and $\{x^i\}$ is a set o $D-2$ angular variables parameterizing the null infinity $(D-2)$-dimensional sphere $S^{D-2}$. Minkowski line element reads
\begin{equation}
    ds^2=-du^2-2dudr+r^2\gamma_{ij}dx^idx^j \ \ \ i,j=1,...,D-2;
\end{equation}
metric and non vanishing Christoffel symbols are given by
\begin{equation}
   g_{\mu \nu} =\begin{bmatrix}
-1 &-1 & 0\\
-1 & 0 & 0\\
0 & 0 & r^2\gamma_{ij}
\end{bmatrix}, \ \ \ 
 g^{\mu \nu} =\begin{bmatrix}
0 & -1 & 0\\
-1 & 1 & 0\\
0 & 0 & \frac{1}{r^2}\gamma_{ij}^{-1}
\end{bmatrix};
\end{equation}
\begin{equation}
    \Gamma^i_{jr}=\Gamma^i_{rj}=\frac{1}{r}\delta^i_j, \ \ \ \Gamma^u_{ij}=-\Gamma^r_{ij}=r\gamma_{ij}, \ \ \ \Gamma^k_{ij}=\frac{1}{2}\gamma^{kl}[-\partial_l\gamma_{ij}+\partial_j\gamma_{li}+\partial_i\gamma_{jl}].
    \label{cribon}
\end{equation}

\begin{comment}
    The equations of motion are the same as before but now with covariant derivative instead of ordinary ones. In Bondi coordinates the gauge transformation and gauge for gauge transformation are expressed in terms of covariant derivatives as
\begin{equation}
\begin{aligned}
   \phi'_{[\rho \sigma] \nu} &= \phi_{[\rho \sigma ]\nu}+\nabla_{\rho}\lambda_{(\sigma \nu)}-\nabla_{\sigma}\lambda_{(\rho \nu)}+\nabla_{\rho}\Lambda_{[\sigma \nu]}-\nabla_{\sigma}\Lambda_{[\rho \nu]}+2\nabla_{\nu}\Lambda_{[\sigma \rho]};\\
\Lambda'_{\mu \nu}&=\Lambda_{\mu \nu}+\nabla_{\mu}\Theta_{\nu}-\nabla_{\nu}\Theta_{\mu};\\
\lambda'_{\mu \nu}&=\lambda_{\mu \nu}+2\nabla_{\mu}\Theta_{\nu}+2\nabla_{\nu}\Theta_{\mu}.
\end{aligned}    
\end{equation}
\end{comment}

\section{The de Donder-like gauge fixing and residual gauge}\label{sec2}
Let us enumerate the independent field components. First of all, note that when the three indexes are all equal the field component is vanishing due to the irreducibility condition; moreover, since the antisymmetry property, field components such that the first two indexes are equal turn out to be vanishing. Furthermore some field components are related to others by antisymmetry or irreducubility property. Below the counting of the independent field components.

\paragraph{Components with $(u,r)$ in the antisymmetric pair.}

By antisymmetry, the non-vanishing components are
\begin{equation}
\phi_{[ur]u}, \qquad
\phi_{[ur]r}, \qquad
\phi_{[ur]i}.
\end{equation}
This gives
$1+1+(D-2)=D$
independent components in this sector.

\paragraph{Components with one angular index in the antisymmetric pair.}

For each $i=1,\dots ,D-2$, we have
\begin{equation}
\begin{aligned}
&\phi_{[ui]u}, \qquad
\phi_{[ui]r}, \qquad
\phi_{[ri]r}, \\
& \qquad \ \ 
\phi_{[ui]j}, \qquad \  
\phi_{[ri]j}.
\end{aligned}
\end{equation}
For fixed $i$, the first row gives $3$ components while the second row $2(D-2)$;  therefore
$3(D-2)+2(D-2)^2$ independent components in this sector.

\paragraph{Components with two angular indices in the antisymmetric pair.}

For $i<j$ the only independent field components are
\begin{equation}
\phi_{[ij]k},  
\end{equation}
which are subject to the restriction obtained from the irreducibility condition,
\begin{equation}
\phi_{[ij]k}
+\phi_{[ki]j}
+\phi_{[jk]i}
=0 ,
\end{equation}
which removes the totally antisymmetric part. Hence their number is \((D-2)\binom{D-2}{2}-\binom{D-2}{3}.
\)

\paragraph{Total number of components.}

Summing all contributions, one finds
\begin{equation}
D + 3(D-2)+2(D-2)^2
+ (D-2)\binom{D-2}{2}-\binom{D-2}{3}
= \frac{D(D^2-1)}{3},
\end{equation}
in agreement with the dimension of the irreducible $\mathrm{GL}(D-1,1)$ representation associated with the Young diagram $(2,1)$. For the specific case of $D=5$ we have
\begin{equation}
    5+9+18+9-1=40.
\end{equation}
Since we are interested in irreducible representation of  $\mathrm{SO}(D-1,1)$ we need to require the vanishing of the unique trace of the Curtright field. However, the traceless condition is not gauge invariant, therefore, as in the graviton case, we are going to require it on-shell ensuring a compatibility condition with other gauge fixing.\\

Let us discuss the gauge and gauge-for-gauge fixing. In order to reduce the equations of motion \eqref{eomfs} to a massless wave equation we have to require that
\begin{equation}
    \mathcal{D}_{\beta \alpha }:=\hat{\partial} \cdot \phi_{ \beta \alpha}-\frac{1}{2}\partial_{\alpha}\Bar{\phi}_{\beta}=0,
    \label{ddg}
\end{equation}
which is a de Donder-like fixing and where we are using notations \eqref{notation}. The equations of motion \eqref{eomfs} can be written as
\begin{equation}
    \Box \phi_{[\alpha \beta]\gamma}+\partial_{\gamma}(\mathcal{D}_{\alpha \beta}-\mathcal{D}_{\beta \alpha})+\partial_{\beta}\mathcal{D}_{\alpha \gamma}-\partial_{\alpha}\mathcal{D}_{\beta \gamma}=0,
    \label{eqde}
\end{equation}
hence, in de Donder-like gauge the equation of motion are 
\begin{equation}
    \Box \phi_{[\alpha \beta ]\gamma}=0;
    \label{eom}
\end{equation}
whose explicit writing for the independent field components is in Appendix \ref{AppB}, Section \ref{B3}. Indeed the non-Box term $E_{[\alpha \beta] \gamma}^{(\not \Box)}$  in the equations of motion is
\begin{equation}
\begin{aligned}
E_{[\alpha \beta] \gamma}^{(\not \Box)}&:=\partial_{\gamma}(\partial_{\alpha}\Bar{\phi}_{\beta}-\partial_{\beta}\Bar{\phi}_{\alpha})+\partial_{\beta}\hat{\partial} \cdot \phi_{\alpha \gamma}-\partial_{\alpha}\hat{\partial} \cdot \phi_{\beta \gamma}-\partial_{\gamma}\Bar{\partial} \cdot \phi_{[\alpha \beta]},
\end{aligned}
\end{equation}
while the non-Box term in equations \eqref{eqde} gives
\begin{equation}
\begin{aligned}
    &\partial_{\gamma}(\hat{\partial} \cdot \phi_{ \alpha \beta}-\frac{1}{2}\partial_{\beta}\Bar{\phi}_{\alpha}-\hat{\partial} \cdot \phi_{ \beta \alpha}+\frac{1}{2}\partial_{\alpha}\Bar{\phi}_{\beta})+\partial_{\beta}(\hat{\partial} \cdot \phi_{ \alpha \gamma}-\frac{1}{2}\partial_{\gamma}\Bar{\phi}_{\alpha})-\partial_{\alpha}(\hat{\partial} \cdot \phi_{ \beta \gamma}-\frac{1}{2}\partial_{\gamma}\Bar{\phi}_{\beta})
\end{aligned}
\end{equation}
which can be rewritten as
\begin{equation}
\begin{aligned}
    \partial_{\gamma}(\hat{\partial} \cdot \phi_{ \alpha \beta}-\hat{\partial} \cdot \phi_{  \beta \alpha})+\partial_{\beta}\hat{\partial} \cdot \phi_{\alpha \gamma}-\partial_{\alpha}\hat{\partial} \cdot \phi_{\beta \gamma}+\partial_{\gamma}(\partial_{\alpha}\Bar{\phi}_{\beta}-\partial_{\beta}\Bar{\phi}_{\alpha})\equiv E_{[\alpha \beta] \gamma}^{(\not \Box)}
\end{aligned}
\end{equation}
using the irreducibility condition to write 
\begin{equation}
    \hat{\partial} \cdot \phi_{ \alpha \beta}-\hat{\partial} \cdot \phi_{  \beta \alpha}=-\Bar{\partial} \cdot \phi_{[\alpha \beta]}.
    \label{barcap}
\end{equation}

The residual gauge after the de Donder gauge fixing is given by gauge parameters that satisfy
\begin{equation}
\begin{aligned}
\Box (\lambda_{\beta\alpha}+ \Lambda_{\beta\alpha})
-\partial_{\beta}\bigl(\partial\!\cdot\!(\lambda_{\alpha}+
\Lambda_{\alpha})\bigr)
-\frac{1}{2}\,\partial_{\alpha}\partial_{\beta}\lambda^{\mu}_{\mu}
+\frac{1}{2}\,\partial_{\alpha}\bigl(\partial\!\cdot\!\lambda_{\beta}\bigr)
-\frac{5}{2}\,\partial_{\alpha}\bigl(\partial\!\cdot\!\Lambda_{\beta}\bigr)=0 .
\end{aligned}
\end{equation} 
As for the graviton field, we now use the residual gauge to fix the traceless condition in order to recover the on-shell irreducible representation. The traceless condition can be imposed requiring
\begin{equation}
\partial_{\beta}\lambda^{\mu}{}_{\mu}
-\partial\!\cdot\!\lambda_{\beta}
+\partial\!\cdot\!\Lambda_{\beta}
= -\bar{\phi}_{\beta} .
\label{tracecon}
\end{equation}
The compatibility condition with the de Donder-like gauge is given by
\begin{equation}
\Box\bigl(\partial\!\cdot\!\Lambda_{\alpha}\bigr)
=\frac{1}{4}\,\partial_{\alpha}
\bigl(\partial_{\beta}\bar{\phi}^{\beta}\bigr);
\label{compatibility}
\end{equation}
Requiring \eqref{tracecon} and \eqref{compatibility} it is enough to reach de Donder-like plus traceless gauge.
In de Donder gauge, the equations of motion for the \((2,1)\) field take the simple form \eqref{eom}.
Contracting indices immediately yields
\begin{equation}
\Box \bar{\phi}^{\beta}=0,
\label{eom_trace}
\end{equation}
and taking an additional divergence gives
\begin{equation}
\Box\bigl(\partial_{\beta}\bar{\phi}^{\beta}\bigr)
=\partial_{\beta}\bigl(\Box\bar{\phi}^{\beta}\bigr)=0.
\label{eom_div_trace}
\end{equation}
Therefore, the source term on the right-hand side of
\eqref{compatibility}
is the gradient of a harmonic function, i.e. well defined.
Hence, the compatibility equation does not require
\(\partial_{\beta}\bar{\phi}^{\beta}\) to vanish identically.
It is sufficient that the field satisfies the equations of motion
\eqref{eom}, which ensure that the source term in
\eqref{compatibility}
is well defined.
In this sense, the compatibility condition is solvable
on-shell, although it does not collapse to an identity, in contrast to what happens in the graviton case. Moreover, from the irreducibility condition \eqref{irregl} we have that the $\bar{\bullet}$-divergence is related to the $\hat{\bullet}$-divergencies, see \eqref{barcap}. Since those ones are vanishing due to the de Donder-like gauge plus the traceless gauge we have also that the field is both $\bar{\bullet}$ and $\hat{\bullet}$ divergenceless. Therefore, also the trace is divergenceless. This is the analogue of what happens in the graviton case. After this gauge fixing the equations of motion \eqref{eomfs} can be written as
\begin{equation}
    \partial^{\nu}H_{[\nu \alpha \beta] \gamma}=0, \quad H_{\nu \alpha}=0.
    \label{eqfix}
\end{equation} 
The residual gauge is then parameterized by 
\begin{equation}
\begin{aligned}
     &\partial_{\beta}\lambda^{\mu}{}_{\mu}
-\partial\!\cdot\!\lambda_{\beta}
+\partial\!\cdot\!\Lambda_{\beta}=0,\\
&\Box\bigl(\partial\!\cdot\!\Lambda_{\alpha}\bigr)=0.
\label{resgauge}
\end{aligned}
\end{equation}

Regarding the gauge-for-gauge, we note that whatever gauge-for-gauge we want to fix, it must be such that also the new gauge parameters, as well as the original ones, must be in the residual gauge in order to not spoil de Dender-like plus traceless gauge and so to be compatible with our gauge fixing. Taking a divergence, from \eqref{gfg}, we have
\begin{equation}
    \partial \cdot \lambda'_{\beta}=\partial \cdot \lambda_{\beta}+3\Box\Theta_{\beta}+3\partial_{\beta}\partial \cdot\Theta, \qquad \partial \cdot \Lambda'_{\beta}=\partial \cdot \Lambda_{\beta}+\Box\Theta_{\beta}-\partial_{\beta}\partial \cdot\Theta;
    \label{divpar}
\end{equation}
moreover
\begin{equation}
    \lambda^{'\mu}{}_{\mu}=\lambda^{\mu}{}_{\mu}+6\partial \cdot \Theta.
    \label{traprar}
\end{equation}
We can ensure that a gauge-for-gauge transformation does not spoil our gauge fixing, i.e,
\begin{equation}
\begin{aligned}
     &\partial_{\beta}\lambda^{'\mu}{}_{\mu}
-\partial\!\cdot\!\lambda'_{\beta}
+\partial\!\cdot\!\Lambda'_{\beta}=0,\\
&\Box\bigl(\partial\!\cdot\!\Lambda'_{\alpha}\bigr)=0,
\label{eqgauge}
\end{aligned}
\end{equation}
by requiring conditions on $\Theta_{\mu}$ given by inserting \eqref{divpar} and \eqref{traprar} in \eqref{eqgauge}
\begin{equation}
\begin{aligned}
    &\partial_{\beta}\lambda^{'\mu}{}_{\mu}
-\partial\!\cdot\!\lambda_{\beta}
+\partial\!\cdot\!\Lambda_{\beta} + 6 \partial_{\beta} \partial \cdot \Theta -3\Box \Theta_{\beta}-3 \partial_{\beta} \partial \cdot \Theta+\Box \Theta_{\beta} -\partial_{\beta} \partial \cdot \Theta=0;\\
&\Box\bigl(\partial\!\cdot\!\Lambda_{\alpha}\bigr)+\Box(\Box \Theta_{\alpha} -\partial_{\alpha} \partial \cdot \Theta)=0;
\end{aligned}
\end{equation}
using that the original gauge parameters are in the residual gauge \eqref{resgauge}, $\Theta_{\mu}$ has to satisfy the following system in order to fix a possible gauge-for-gauge without producing new gauge parameter that are not in the residual gauge 
\begin{equation}
\begin{aligned}
&\Box\Theta_\beta-\partial_\beta(\partial\!\cdot\!\Theta)=0,\\    
&\Box\Bigl(\Box\Theta_\alpha-\partial_\alpha(\partial\!\cdot\!\Theta)\Bigr)=0 ,
\label{comp}
\end{aligned}
\end{equation}
which is redundant, so only the first one is necessary. 

The gauge-for-gauge fixing we require is
\begin{equation}
\lambda'_{ui}=\lambda'_{ri}=\Lambda'_{ui}=\Lambda'_{ri}=0;
\end{equation}
looking to the gauge-for-gauge transformations \eqref{gfg} this means
\begin{equation}
\begin{aligned}
    &\Lambda_{r i}+\partial_{r}\Theta_{i}-\partial_{i}\Theta_{r}=0,\\
&\lambda_{ri}+3\partial_{r}\Theta_{i}+3\partial_{i}\Theta_{r}=0,\\
      &\Lambda_{u i}+\partial_{u}\Theta_{i}-\partial_{i}\Theta_{u}=0,\\
&\lambda_{ui}+3\partial_{u}\Theta_{i}+3\partial_{i}\Theta_{u}=0.
\end{aligned}    
\end{equation}
Solving for the gradient components of $\Theta_{\mu}$ we get 
\begin{align}
\partial_r\Theta_i &= -\frac16\,\lambda_{ri}-\frac12\,\Lambda_{ri},
\label{gfgf1}\\
\partial_i\Theta_r &= -\frac16\,\lambda_{ri}+\frac12\,\Lambda_{ri}, \\
\partial_u\Theta_i &= -\frac16\,\lambda_{ui}-\frac12\,\Lambda_{ui}, \\
\partial_i\Theta_u &= -\frac16\,\lambda_{ui}+\frac12\,\Lambda_{ui} ,
\label{gfgf4}
\end{align}
that fix gradients of some components of $\Theta_{\mu}$.
Now we need to check compatibility with the residual gauge, i.e. with \eqref{comp}, that can be rewritten in terms of an auxiliary object $F_{\mu \nu}:= \partial_{\mu}\Theta_{\nu}-\partial_{\nu}\Theta_{\mu}$ as 
\begin{equation}
    \partial^{\mu}F_{\mu \nu}=0 \quad \Rightarrow \quad \begin{cases}

\partial^\mu F_{\mu r} &= \partial^u F_{u r}+\partial^i F_{i r}=0,\\
\partial^\mu F_{\mu u} &= \partial^r F_{r u}+\partial^i F_{i u}=0,\\
\partial^\mu F_{\mu i} &= \partial^u F_{u i}+\partial^r F_{r i}+\partial^j F_{j i}=0.
    \end{cases}
\end{equation}
Equations \eqref{gfgf1}--\eqref{gfgf4} fix some components of $F_{\mu \nu}$ but we have enough free functions to ensure compatibility.

To study the residual gauge we have also to require some fall-offs condition to preserve. To this goal we study the energy-momentum tensor obtained by varying the action with respect to the metric. Since the lagrangian depends on the metric only through index contractions, the result can be written as
\begin{equation}
T_{\mu\nu}
=
\frac{4}{3}\,
H_{\mu\alpha\beta\gamma} H_{\nu}^{\alpha\beta\gamma}
- 2\,
H_{\mu\alpha} H_{\nu}^{\alpha}
-\frac{1}{6}\,
\eta_{\mu\nu}
\left(
H_{\alpha\beta\gamma\delta} H^{\alpha\beta\gamma\delta}
- 3 H_{\alpha\beta} H^{\alpha\beta}
\right).
\end{equation}
All components of $T_{\mu\nu}$ are therefore quadratic in the partial field strength or its unique trace. Requiring the energy flux to be finite and non-vanishing implies the asymptotic conditions
\begin{equation}
\begin{aligned}
    H_{uriu} \sim \mathcal{O}(r^{-(\frac{D-4}{2})}), \qquad  H_{urir} \sim \mathcal{O}(r^{-(\frac{D-4}{2})}), \qquad  H_{urij} \sim \mathcal{O}(r^{-(\frac{D-6}{2})}).
    \label{leadingHcharge}
\end{aligned}    
\end{equation}
where we have taken into account that in general, for each angular index $i$ we have one $r$ more than the cartesian component since in Bondi coordinates the field carries additional powers of $r$ due to the jacobian of the coordinate transformation. In Bondi coordinates the gauge transformations has to preserve the radiation fall-offs on field components
\begin{equation}
\delta_{\lambda,\Lambda}(\phi_{[\alpha \beta]\gamma})\sim  \begin{cases}
        O\big(r^{-\frac{(D-2)}{2}}\big)& if \ \alpha,\beta, \gamma \notin \{x^i\};\\
        O\big(r^{-\frac{(D-4)}{2}}\big)& if \ \alpha \ or \ \beta \ or \ \gamma \in \{x^i\};\\
        O\big(r^{-\frac{(D-6)}{2}}\big)& if \ \alpha,\beta \ or \ \beta, \gamma \ or \ \alpha, \gamma \in \{x^i\};\\
        O\big(r^{-\frac{(D-8)}{2}}\big)& if \ \alpha,\beta, \gamma  \in \{x^i\}.
    \end{cases}
\end{equation}

We now assume a polyhomogeneous power law expansion for the gauge parameters of the form\footnote{We will be interested in $l$ of the form $\frac{D-n}{2}$ where $n \in \mathbb{Z}$; we have $\frac{D-n}{2}<\frac{D-m}{2}$ if and only if $m<n$ and $\frac{D-n}{2}>\frac{D-m}{2}$ if and only if $m>n$. Given $\frac{1}{r^{\frac{D-n}{2}}}$ the power $\frac{1}{r^{\frac{D-m}{2}}}$ is subleading if and only if $m<n$.}
\begin{equation}
\begin{aligned}
\lambda_{\mu \nu}&=\sum_{l\in \frac{1}{2}\mathbb{Z}}\frac{\lambda_{\mu \nu}^{(l)}(u,\{x^i\})}{r^l}+\frac{\bar{\lambda}_{\mu \nu}^{(l)}(u,\{x^i\})}{r^l}ln(r),\\
\Lambda_{\mu \nu}&=\sum_{l\in \frac{1}{2}\mathbb{Z}}\frac{\Lambda_{\mu \nu}^{(l)}(u,\{x^i\})}{r^l}+\frac{\bar{\Lambda}_{\mu \nu}^{(l)}(u,\{x^i\})}{r^l}ln(r).
\label{expgauge}
\end{aligned}
\end{equation}
Looking at the residual gauge condition equations in  \ref{B1}, preservation of fall-off condition equations in \ref{B2} and equation of motions in \ref{B3} we find the expansions of the gauge parameter components entering in the asymptotic charge
\begin{equation}
\begin{aligned}
\lambda_{ij}&=\frac{\lambda_{ij}^{(\frac{D-8}{2})}(\{x^i\})}{r^{\frac{D-8}{2}}}+\sum_{l>\frac{D-8}{2}}\frac{\lambda_{\mu \nu}^{(l)}(u,\{x^i\})}{r^l}+\frac{\bar{\lambda}_{\mu \nu}^{(l)}(u,\{x^i\})}{r^l}ln(r),\\
\Lambda_{ij}&=\frac{\Lambda_{ij}^{(\frac{D-8}{2})}(\{x^i\})}{r^{\frac{D-8}{2}}}+\sum_{l>\frac{D-8}{2}}\frac{\Lambda_{\mu \nu}^{(l)}(u,\{x^i\})}{r^l}+\frac{\bar{\Lambda}_{\mu \nu}^{(l)}(u,\{x^i\})}{r^l}ln(r).
\label{expgaugelead}
\end{aligned}
\end{equation}
These expansions depend both on the choice of fall-offs and on the gauge fixings considered here.
\section{Asymptotic charges}\label{sec3}
The asymptotic charge relevant for our discussion has its roots in the covariant phase space formalism à la Wald \cite{Lee:1990nz,Ciambelli:2022vot}; we consider the N\"other charge
\begin{equation}
    Q_{\Lambda, \lambda}[S^{D-2}_u]=\lim_{r \rightarrow +\infty}\oint_{S^{D-2}_u} k^{ur}_{\Lambda, \lambda}r^{D-2}d\Omega,
    \label{leewald}
\end{equation}
where we already specialize to the codimension-2 celestial sphere at Minkowski null infinity. In \eqref{leewald}, $k^{\mu \nu}_{\Lambda,\lambda}$ is the N\"other two-form, $S_u^{D-2}$ is the celestial sphere at Minkowski null infinity, and $r^{D-2}d\Omega$ is its integration measure. To compute the N\"other two-form $k^{\mu \nu}_{\Lambda,\lambda}$ we need to compute the variation of the lagrangian density \eqref{laden} 
\begin{equation}
    {\displaystyle \delta \mathcal{L}=-{\frac {1}{3}}[\delta(H_{\alpha \beta \gamma \mu })H^{\alpha \beta \gamma \mu }-3\delta(H_{\alpha \beta })H^{\alpha \beta }]}=:\delta \mathcal{L}_1+\delta\mathcal{L}_2,
    \label{ladenvar}
\end{equation}

\begin{comment}
where
\begin{equation}
\begin{aligned}
    \delta(H_{\alpha \beta \gamma \mu })=&\partial_{\alpha}\delta\phi_{[\beta \gamma]\mu}+\partial_{\beta}\delta\phi_{[\gamma \alpha]\mu}+\partial_{\gamma}\delta\phi_{[\alpha \beta]\mu}= \\ =&+\partial_{\alpha}[\partial_{\beta}\lambda_{(\gamma \mu)}-\partial_{\gamma}\lambda_{(\beta \mu)}-\partial_{\beta}\Lambda_{[\mu \gamma ]}-\partial_{\gamma}\Lambda_{[\beta \mu]}+2\partial_{\mu}\Lambda_{[\gamma \beta]}]+\\
    &+\partial_{\beta}[\partial_{\gamma}\lambda_{(\alpha \mu)}-\partial_{\alpha}\lambda_{(\gamma \mu)}-\partial_{\gamma}\Lambda_{[\mu \alpha ]}-\partial_{\alpha}\Lambda_{[\gamma \mu]}+2\partial_{\mu}\Lambda_{[\alpha \gamma]}]+\\
    &+\partial_{\gamma}[\partial_{\alpha}\lambda_{(\beta \mu)}-\partial_{\beta}\lambda_{(\alpha \mu)}-\partial_{\alpha}\Lambda_{[\mu \beta ]}-\partial_{\beta}\Lambda_{[\alpha \mu]}+2\partial_{\mu}\Lambda_{[\beta \alpha]}]=\\
    =&+2\partial_{\mu}[\partial_{\alpha}\Lambda_{[\beta \gamma]}+\partial_{\beta}\Lambda_{[\gamma \alpha]}+\partial_{\gamma}\Lambda_{[\alpha \beta]}];
\end{aligned}
\end{equation}
and 
\begin{equation}
\begin{aligned}
     \delta(H_{\alpha \beta})=&\eta^{\gamma \mu}\delta(H_{\alpha \beta \gamma \mu })=2\partial^{\gamma}[\partial_{\alpha}\Lambda_{[\beta \gamma]}+\partial_{\beta}\Lambda_{[\gamma \alpha]}+\partial_{\gamma}\Lambda_{[\alpha \beta]}].
\end{aligned}
\end{equation}

\end{comment}
\paragraph{Variation of the first term $\delta \mathcal{L}_1$.}
The variation of the first term reads
\begin{equation}
\delta\mathcal L_1
=-\frac{1}{3}H^{\alpha\beta\gamma\mu}\,
\delta H_{\alpha\beta\gamma\mu}.
\end{equation}
Using the definition \eqref{field} and the antisymmeetry properties of $H^{\alpha\beta\gamma\mu}$ we get
\begin{equation}
H^{\alpha\beta\gamma\mu}\delta H_{\alpha\beta\gamma\mu}
=3H^{\alpha\beta\gamma\mu}\partial_{\alpha}
\delta\phi_{[\beta\gamma]\mu},
\end{equation}
and hence
\begin{equation}
\delta\mathcal L_1
=-H^{\alpha\beta\gamma\mu}\partial_{\alpha}
\delta\phi_{[\beta\gamma]\mu}.
\end{equation}
Integrating by parts,
\begin{equation}
\delta\mathcal L_1
=
\bigl(\partial_{\alpha}H^{\alpha\beta\gamma\mu}\bigr)
\delta\phi_{[\beta\gamma]\mu}
-\partial_{\alpha}\Bigl(
H^{\alpha\beta\gamma\mu}\delta\phi_{[\beta\gamma]\mu}
\Bigr).
\end{equation}

\paragraph{Variation of the trace term $\delta \mathcal{L}_2$.}
The variation of the second term of the Lagrangian reads
\begin{equation}
\delta\mathcal L_2
=H^{\alpha\beta}\delta H_{\alpha\beta}
=H^{\alpha\beta}\partial^{\gamma}\delta\phi_{[\alpha\beta]\gamma}
+H^{\alpha\beta}
\bigl(\partial_{\alpha}\delta\bar\phi_{\beta}
-\partial_{\beta}\delta\bar\phi_{\alpha}\bigr).
\end{equation}
Since $H^{\alpha\beta}$ is antisymmetric, the second term reduces to
\begin{equation}
H^{\alpha\beta}
\bigl(\partial_{\alpha}\delta\bar\phi_{\beta}
-\partial_{\beta}\delta\bar\phi_{\alpha}\bigr)
=2H^{\alpha\beta}\partial_{\alpha}\delta\bar\phi_{\beta},
\end{equation}
and, integrating by parts, we get
\begin{align}
\delta\mathcal L_2
={}&
-\bigl(\partial^{\gamma}H^{\alpha\beta}\bigr)
\delta\phi_{[\alpha\beta]\gamma}
-2\bigl(\partial_{\alpha}H^{\alpha\beta}\bigr)
\delta\phi_{[\beta\gamma]}{}^{\gamma}
+\partial_{\mu}\big(H^{\alpha\beta}\delta\phi_{[\alpha\beta]}{}^{\mu}
+2H^{\mu\beta}\delta\phi_{[\beta\gamma]}{}^{\gamma}\big),
\end{align}
\paragraph{Total variation.} Combining $\delta\mathcal L_1$ and $\delta\mathcal L_2$, the total variation
takes the form
\begin{equation}
\delta\mathcal L
=
\Bigl(
\partial_{\alpha}H^{\alpha\beta\gamma\mu}
-\partial^{\mu}H^{\beta\gamma}
\Bigr)\delta\phi_{[\beta\gamma]\mu}
-2\bigl(\partial_{\alpha}H^{\alpha\beta}\bigr)
\delta\phi_{[\beta\gamma]}{}^{\gamma}
+\partial_{\mu}\\j^{\mu},
\end{equation}
where
\begin{equation}
j^{\mu}
=
-\,H^{\mu\beta\gamma\nu}\delta\phi_{[\beta\gamma]\nu}
+H^{\beta\gamma}\delta\phi_{[\beta\gamma]}{}^{\mu}
+2H^{\mu\beta}\delta\phi_{[\beta\gamma]}{}^{\gamma}.
\end{equation}
\paragraph{N\"oether two-form.}

The N\"oether two-form can be computed by inserting the gauge variation
\begin{equation}
\delta_{\lambda,\Lambda}\phi_{[\rho\sigma]\nu}
=
\underbrace{\partial_{\rho}\lambda_{(\sigma\nu)}-\partial_{\sigma}\lambda_{(\rho\nu)}}_{\delta_\lambda\phi_{[\rho\sigma]\nu}}
+
\underbrace{\partial_{\rho}\Lambda_{[\sigma\nu]}-\partial_{\sigma}\Lambda_{[\rho\nu]}+2\,\partial_{\nu}\Lambda_{[\sigma\rho]}}_{\delta_\Lambda\phi_{[\rho\sigma]\nu}},
\label{fullgaugephi}
\end{equation}
into
\begin{equation}
j^{\mu}
=
-\,H^{\mu\beta\gamma\nu}\,\delta\phi_{[\beta\gamma]\nu},
\label{ThetaDef}
\end{equation}
where we have assumed the de Donder-like plus traceless gauge fixing and so the equation of motion \eqref{eqfix}. We note that since $H^{\mu\beta\gamma\nu}$ is contracted with a (2,1) irreducible tensor, therefore all other representations different from (2,1) on the contracted indices do not enter. This forces a further identity on $H^{\mu\beta\gamma\nu}$ in the determination of $j^{\mu}$
\begin{equation}
H^{\mu\beta\gamma\nu}
+ H^{\mu\gamma\nu\beta}
+ H^{\mu\nu\beta\gamma}
= 0 \, .
\label{id}
\end{equation}
We get, using antisymmetry property on $(\beta,\gamma)$ and identity \eqref{id}
\begin{equation}
\begin{aligned}
j^{\mu}
=&
-2\,H^{\mu\beta\gamma\nu}\,\partial_{\beta}\lambda_{\gamma\nu}
-2\,H^{\mu\beta\gamma\nu}\,\partial_{\beta}\Lambda_{\gamma\nu}
+2\,H^{\mu\beta\gamma\nu}\,\partial_{\nu}\Lambda_{\beta\gamma}=\\
=&-2\,H^{\mu\beta\gamma\nu}\,\partial_{\beta}\lambda_{\gamma\nu}
-6\, H^{\mu\nu\beta\gamma}\,\partial_{\nu}\Lambda_{\beta\gamma}.
\end{aligned}
\label{j:traceless}
\end{equation}
Integrating by parts we get
\begin{equation}
j^{\mu}=\partial_{\alpha}k^{\mu\alpha}+R^{\mu} \, ,
\label{j:split}
\end{equation}
where
\begin{equation}
\begin{aligned}
k^{\mu\alpha}
&:=
-2\,H^{\mu\alpha\gamma\beta}\lambda_{\gamma\beta}
-6\,H^{\mu\alpha\beta\gamma}\Lambda_{\beta\gamma}
\end{aligned}
\label{def:U}
\end{equation}
while the residual term \(R^\mu\) is
\begin{equation}
\begin{aligned}
R^{\mu}
&:=
2\,(\partial_{\beta}H^{\mu\beta\gamma\nu})\,\lambda_{\gamma\nu}
+6\,(\partial_{\nu}H^{\mu\nu\beta\gamma})\,\Lambda_{\beta\gamma},
\end{aligned}
\label{def:R0}
\end{equation}
that is vanishing due to equations of motion \eqref{eqfix}. In particular, on-shell, we get
\begin{equation}
j^\mu\approx
\partial_{\alpha}k^{[\mu\alpha]}.
\label{ThetaFullOnShell}
\end{equation}

The charge is 
\begin{equation}
\begin{aligned}
Q_{\Lambda, \lambda}[S^{D-2}_u]&=-2\lim_{r \rightarrow +\infty}\oint_{S^{D-2}_u} k^{ur}_{\Lambda, \lambda}r^{D-2}d\Omega\\
&=-2\lim_{r \rightarrow +\infty}\oint_{S^{D-2}_u}H^{u r i j}\bigl(\lambda_{(i j)}+3\,\Lambda_{[i j]}\bigr)r^{D-2}d\Omega=\\
&=-2\lim_{r \rightarrow +\infty}\oint_{S^{D-2}_u}g^{u \alpha}g^{r\beta}g^{i\gamma}g^{j\delta}H_{\alpha \beta \gamma \delta}\bigl(\lambda_{(i j)}+3\,\Lambda_{[i j]}\bigr)r^{D-2}d\Omega=\\
&=-2\lim_{r \rightarrow +\infty}\oint_{S^{D-2}_u}g^{u r}g^{r\beta}g^{i\tilde{i}}g^{j\tilde{j}}H_{r \beta \tilde{i} \tilde{j}}\bigl(\lambda_{(i j)}+3\,\Lambda_{[i j]}\bigr)r^{D-2}d\Omega=\\
&=-2\lim_{r \rightarrow +\infty}\oint_{S^{D-2}_u}\gamma^{i\tilde{i}}\gamma^{j\tilde{j}}H_{r u \tilde{i} \tilde{j}}\bigl(\lambda_{(i j)}+3\,\Lambda_{[i j]}\bigr)r^{D-6}d\Omega,
\end{aligned}
\end{equation}
since $g^{ij}=\frac{\gamma^{ij}}{r^2}$. Let us note here that the gauge condition $\hat{\nabla} \cdot \phi_{ \beta \alpha}=0$, which is the result of $\mathcal{D}_{\alpha \beta}=0$ plus the traceless gauge, implies that $\partial_u\phi_{rij}^{(\frac{D-6}{2})}=0$; therefore the leading order of $H_{ruij}$ is $\mathcal{O}(r^{-(D-4)/2})$. Inserting the leading orders \eqref{expgaugelead} we get a finite $O(1)$, non-vanishing charge
\begin{equation}
    Q_{\Lambda, \lambda}[S^{D-2}_u]= -2\oint_{S^{D-2}_u}\gamma^{i\tilde{i}}\gamma^{j\tilde{j}}H^{\big(\frac{D-4}{2}\big)}_{r u \tilde{i} \tilde{j}}\bigg(\lambda^{\big(\frac{D-8}{2}\big)}_{(i j)}+3\,\Lambda^{\big(\frac{D-8}{2}\big)}_{[i j]}\bigg)d\Omega
    \label{chargefin}
\end{equation}
since $-\frac{D-4}{2}-\frac{D-8}{2}+D-6=0$.

This form of the charge resemble the one of $p$-form gauge fields and, in fact, the mechanisms leading to charge \eqref{chargefin} are very similar to those entering in the $p$-form charge \cite{manz2, Manzoni:2024tow}: gauge condition that shut off the field components that can make the charge divergent, the entering of the partial field strength\footnote{In the case of $p$-forms is the true field strength that enter in the charge.} components with $ru$-indices (hence an electric-like charge), only angular components of the gauge parameters entering in the charge. Moreover, let us note that the gauge condition fixed here is essentially the same fixed for $p$-forms once we take into account the indices structure of the field. Indeed, for a $p$-form, there is no non-trivial trace and the de Donder-like gauge fixing reduces to Lorenz-like gauge fixing since there is only one independent divergence for a $p$-form field. \\
There are so at least two possible guesses we can made with a relative degree of safety: the form of the charge for a hook field ($p,1$) with $p>1$ and the path to follows for derive finite $O(1)$, non-vanishing charge for a generic mixed symmetry tensor gauge field. 
\paragraph{Asymptotic charges for ($p,1$)-hook field.}
The guess follows from a straightforward merging of the technicalities in deriving the charge \eqref{chargefin} and in generalizing Maxwell field to the $p$-form case \cite{manz2,Esmaeili:2020eua, Afshar:2018apx, Manzoni:2024tow}. The ($p,1$) hook field $\phi_{[\mu_1,...\mu_{p}]\nu}$ has only one trace (by contracting the last index with one of the first $p$ antisymmetric ones) and two independent divergencies (the one where the derivative contacts with the last index and the one where contracts with one of the antisymmetric indices). Therefore, we are in a situation very similar to the case of the $(2,1)$ field with $p$-antisymmetric indices like the case of a $p$-form field. The guess is that the charge is of the form
\begin{equation}
    Q^{(p,1)}[S^{D-2}_u]= \#_1\oint_{S^{D-2}_u}\gamma^{i_1\tilde{i}_1}...\gamma^{i_{p}\tilde{i}_{p}}H^{\big(a(D,p)\big)}_{r u \tilde{i}_1 ... \tilde{i}_{p}}\bigg(\lambda^{\big(b(D,p)\big)}_{[{i}_1 ... {i}_{p-1}]i_p}+\#_2\,\Lambda^{\big(b(D,p)\big)}_{[{i}_1 ... {i}_{p}]}\bigg)d\Omega,
    \label{chargefingen}
\end{equation}
where $a(D,p):=\frac{D-2(p+1)+2}{2}$, $b(D,p):=\frac{D-(2(p+1)+2)}{2}$ and $\#_1,\#_2$ are two constants which depends critically form the coefficients in the lagrangian. 
\paragraph{Asymptotic charges for general mixed symmetry tensor fields.} The asymptotic charges of a general mixed symmetry tensor gauge field with $N$-families of indices that carry the irreducible representation $(p_1,...,p_N)$ could be computed following the same footsteps. A tensor field of this type will be denoted by
\begin{equation}
    \phi_{A^{(1)}\, A^{(2)} \cdots A^{(N)}} \, , 
\qquad 
A^{(i)} := [a^{(i)}_1 \cdots a^{(i)}_{p_i}] \, ,
\end{equation}
meaning
\begin{equation}
    \phi_{A^{(1)}\,\cdots\,A^{(N)}} 
= \phi_{[a^{(1)}_1 \cdots a^{(1)}_{p_1}] \, \cdots \, [a^{(N)}_1 \cdots a^{(N)}_{p_N}]} \, .
\end{equation}
This field coincides with our $\phi_{[\alpha \beta],\nu}$ when $N=2$ and $p_1=2$, $p_2=1$. Since in general we have more than 2 divergences and 1 trace here we use the notation $\mathrm{div}_i \phi$ for the divergence on the column $i$ and $\mathrm{tr}_{(i,j)} \phi$ for the trace between the columns $i$ and $j$.
Besides antisymmetry in each column, Young irreducibility has to be imposed: for each $i=1,\dots,N-1$ and for any choice of one index $j$ taken from the $(i+1)$-th column, the total antisymmetrization over the whole $i$-th column together with $j$ vanishes
\begin{equation}
\phi_{A^{(1)}...\,[A^{(i)} \, j] \, , \, A^{(i+1)} \setminus j \, , \, A^{(i+2)} \cdots A^{(N)}} = 0 \, .
\label{young}
\end{equation}
Here $A^{(i+1)} \setminus j$ denotes the $(i+1)$-th family with the index $j$ removed and the antisymmetrization bracket in runs over $p_i+1$ indices (the whole $i$-th family plus $j$).\\ 
The key point is that we can expresses divergences $\mathrm{div}_{i+1}\phi$ on column $(i+1)$ as linear combinations of divergences $\mathrm{div}_{i}\phi$ on column $i$. Iterating from the rightmost column to the left, one obtains a chain of reductions
\begin{equation}
    \mathrm{div}_{N}\phi \;\longrightarrow\; \mathrm{div}_{N-1}\phi \;\longrightarrow\; \cdots \;\longrightarrow\; \mathrm{div}_{1}\phi \, .
\end{equation} 
To see the mechanism, fix an adjacent pair $(i,i+1)$ and select an index $j$ appearing in $A^{(i+1)}$ . The irreducibility condition will contains a term of the form $\phi_{A^{(1)}...[a^{(i)}_1 \cdots a^{(i)}_{p_i}] [j \cdots a^{(i+1)}_{p_{i+1}}]...A^{(N)}} $ and terms of the form $\phi_{A^{(1)}...[a^{(i)}_1 \cdots j] [a^{(i)}_{p_i} \cdots a^{(i+1)}_{p_{i+1}}]...A^{(N)}}$. Taking the divergence with respect to $j$ will produce a relation between $\mathrm{div}_{i}\phi$ and $\mathrm{div}_{i+1}\phi$. In this precise sense, for a Young--irreducible tensor in column convention, the divergences taken on later families are not independent: they can be rewritten in terms of divergences taken on earlier families, ultimately in terms of $\mathrm{div}_1\phi$.\\ 
Similarly, for the traces
\begin{equation}
    \mathrm{tr}_{(i,j)} \phi \;\longrightarrow\; 
\mathrm{tr}_{(i,j-1)} \phi \;\longrightarrow\; \cdots \;\longrightarrow\;
\mathrm{tr}_{(i,i+1)} \phi \, ,
\end{equation}
namely any trace between non-adjacent columns can be rewritten as a linear combination of traces between closer columns, ultimately reducing to traces between adjacent columns (and, if desired, to traces involving a preferred column). \\
The crucial step now will be to reduce the equations of motion to an homogeneous wave equation rewriting all the non-$\Box$ terms in terms of a single divergence and a single trace, for example $\mathrm{div}_N\phi$ and $\mathrm{tr}_{(N-1,N)}\phi$, and fixing a de Donder-like gauge to cancel them. This produce residual gauge conditions and we use the residual gauge to eliminate the trace $\mathrm{tr}_{(N-1,N)}\phi$ so that all the traces will vanish and the field will be divergences-free. The next step is the gauge-for-gauge fixing such that only gauge parameters with angular indices enter in the charge. 

The two paragraphs above are mainly speculative and a deeper understanding of these technicalities and the effectiveness of these guesses is one of the next point in the agenda for future works.

\section{Connections with the graviton asymptotic symmetries in $D=5$}\label{sec4}
Since in $D=5$ the graviton and the Curtright field are dual, our scope is to connect the Curtright field asymptotic charge to supertranslation and superrotation charges of the graviton. The $D=5$ charge for the Curtright field reads
\begin{equation}
    Q_{\Lambda, \lambda}[S^{3}_u]= \oint_{S^{3}_u}\gamma^{i\tilde{i}}\gamma^{j\tilde{j}}H^{\big(\frac{1}{2}\big)}_{r u \tilde{i} \tilde{j}}\bigg(\lambda^{\big(-\frac{3}{2}\big)}_{(i j)}+3\,\Lambda^{\big(-\frac{3}{2}\big)}_{[i j]}\bigg)d\Omega;
    \label{chargefinD5}
\end{equation}
the two gauge parameters can be decomposed using Hodge and Hodge-like decompositions which are additional geometric/topological inputs derived from $H^1_{\text{dR}}(S^3)=H^2_{\text{dR}}(S^3)=0$:
\begin{itemize}
 \item Antisymmetric parameter. The antisymmetric parameter can be consider as a 2-form on $S^3$; since $H^2_{\text{dR}}(S^3)=0$ we have no harmonic part
     \begin{equation}
        \Lambda_{[ij]}^{\big(-\frac{3}{2}\big)} = D_{[i} W_{j]}+\epsilon_{ijk}D^k\Phi;
    \end{equation}
    \item Symmetric parameter. The symmetric parameter can be considered as a (1,1) mixed symmetry tensor obtained by the Young projector $\Pi_{(1,1)}$ acting of the bi-form space on $S^3$. Since $H^1_{\text{dR}}(S^3)=0$ we have
\begin{equation}\label{eq:decomp-sym2}
\begin{aligned}
  \lambda^{\big(-\frac{3}{2}\big)}_{(ij)} 
=
-D^{k} B_{k(ij)}
+D_{(i} C_{j)}
+D_{i}D_{j}\,\phi 
-D^{k}D_{(i}\,\Xi_{j)k}
+D^{k}D^{\ell}\,Y_{k i \ell j} .
\end{aligned}
\end{equation}
    where $Y_{lisj}$ carries the same irreducible representation of the Riemann tensor and $S_{ij}:=D^{k}D^{\ell}\,Y_{k i \ell j}$ is symmetric since its antisymmetric part is proportional to the antisymmetric part of the "Ricci" tensor associated to $Y_{kils}$ that is vanishing (see Appendix \ref{AppD}, Section \ref{s}). Moreover, $B$ has to be vanishing. Indeed the $(2) \otimes (1) = (2,1) \oplus (3)$ and both terms vanishes\footnote{For the 3-form is obvious, while for the $(2,1)$ we would have 
\begin{equation*}\label{eq:B-step1}
B_{kij}=B_{kji}=-B_{jki}.
\end{equation*}
Substituting this into the irreducibility condition yields
\begin{equation*}\label{eq:B-step2}
0=B_{kij}+B_{ijk}+B_{jki}
= B_{kij}+B_{ijk}-B_{kij}
= B_{ijk}.
\end{equation*}} requiring a simmetrization on $(i,j)$. 
\end{itemize}
Let us define $\mathcal{H}^{ij} := \gamma^{i\tilde{i}}\gamma^{j\tilde{j}}H^{\big(\frac{1}{2}\big)}_{ru\tilde{i}\tilde{j}}$ that represents the leading-order dynamical components of the Curtright field at the boundary; performing an asymptotic gauge fixing (see Appendix \ref{AppD}, Section \ref{agfa}) to get $D^{\tilde{i}}H^{\big(\frac{1}{2}\big)}_{ru\tilde{i}\tilde{j}}=0$, we have
\begin{equation}
D_i\mathcal{H}^{ij}=D_i\gamma^{i\tilde{i}}\gamma^{j\tilde{j}}H^{\big(\frac{1}{2}\big)}_{ru\tilde{i}\tilde{j}}=\gamma^{j\tilde{j}}D^{\tilde{i}}H^{\big(\frac{1}{2}\big)}_{ru\tilde{i}\tilde{j}}=0.
    \label{gaugefix0}
\end{equation} 
 Moreover by the equations of motion \eqref{eqfix} we also have that
\begin{equation}
    \mathcal{H}:=\gamma_{ij}\mathcal{H}^{ij} \approx 0.
    \label{trace=0}
\end{equation}
The charge reduces, upon integration by parts, to
\begin{equation}
    Q_{\Lambda, \lambda}[S^{3}_u] = \oint_{S^3} \bigg(\Phi\epsilon_{ijk}D^k\mathcal{H}^{ij} + Y_{ki\ell j}D^{k}D^{\ell}\mathcal{H}^{ij}\bigg)d\Omega=:Q_{\Phi}+Q_Y.
\end{equation}  
Let us better study the charge $Q_Y$. To this goal we define $v^{i} := D_{j}\mathcal{H}^{ij}$ and the "Ricci" tensor and scalar associated to $Y_{ki\ell j}$
\begin{equation}
y_{ij} := \gamma^{k\ell }Y_{ki\ell j} \ , \qquad y := \gamma^{ij}y_{ij} \ .
\end{equation}
A tensor with the same symmetries of the Riemann tensor is completely determined by its "Ricci" tensor and scalar in 3 dimensions, hence $Q_Y$ can be completely rewritten it terms of $y_{ij}$ and $y$ as
\begin{equation}
Q_{Y}
=
\int_{S^{3}}
\Big[
y_{ij}\,\Delta \mathcal{H}^{ij}
- y_{i\ell}\,D_{j}D^{\ell}\mathcal{H}^{ij}
+\frac{y}{2}\,D_{i}v^i
\Big]\,
d\Omega \ .
\end{equation}
Using the commutation relation in Appendix \ref{AppD}, \eqref{comm2} and \eqref{comm3} with $K=1$ and the fact that $y_{ij}$ is symmetric we get
\begin{equation}
Q_{Y}
=
\int_{S^{3}}
\Big[
y_{ij}\,(\Delta-3)\,\mathcal{H}^{ij}
- y_{i\ell}\,D^{\ell}v^{i}
\Big]\,
d\Omega=\int_{S^{3}}
y_{ij}\,
\Big[
(\Delta-3)\,\mathcal{H}^{ij}-D^{j}v^{i}
\Big]\,
d\Omega  \ .
\label{cargey}
\end{equation}
A convenient parametrization of a $y_{ij}$ on $S^{3}$ is the York decomposition
\begin{equation}
\begin{aligned}
y_{ij}
=&\ \frac13\,\gamma_{ij}A
+\Big(D_{i}D_{j}-\frac13\,\gamma_{ij}\Delta\Big)B
+\Big(D_{(i}V_{j)}-\frac13\,\gamma_{ij}\,D^kV_k\Big) + y_{ij}^{\text{TT}},
\label{decy}
\end{aligned}
\end{equation}
where $V_{i}$ is a vector field on $S^{3}$ and $y_{ij}^{\text{TT}}$ is transverse-traceless. Using \eqref{decy} into \eqref{cargey} and integrate by parts on $S^{3}$ we have
\begin{equation}
Q_{Y}
=
\int_{S^{3}}
\Big(
A\mathcal{J}_{A}
+B\mathcal{J}_{B}
+V_{i}\,\mathcal{J}_{V}^{\,i}+y_{ij}^{\text{TT}}\mathcal{J}^{ij}_{y^{\text{TT}}}
\Big)\,
d\Omega \ ,
\label{}
\end{equation}
with currents
\begin{equation}
\begin{aligned}
&\mathcal{J}_{A}
:=
\frac13\,\gamma_{ij}\bigg[(\Delta-3)\,\mathcal{H}^{ij}-D^{j}v^{i}\bigg],\\
&\mathcal{J}_{B}
:=
-\Big(D_{i}D_{j}-\frac13\,\gamma_{ij}\Delta\Big)\bigg[(\Delta-3)\,\mathcal{H}^{ij}-D^{j}v^{i}\bigg], \\
&\mathcal{J}_{V}^{\,i}
:=
-\Bigg(D_{j}\bigg[(\Delta-3)\,\mathcal{H}^{(ij)}-D^{(j}v^{i)}\bigg]-\frac13\,D^{i}\bigg(\gamma_{k\ell}\bigg[(\Delta-3)\,\mathcal{H}^{k\ell}-D^{\ell}v^{k}\bigg]\bigg)\Bigg), \\
&J^{ij}_{y^{\text{TT}}}:=(\Delta-3)\,\mathcal{H}^{(ij)}
\end{aligned}
\end{equation}
Under relations \eqref{trace=0}, \eqref{gaugefix0} and \eqref{comm3}
\begin{equation}
    \mathcal{J}_A\approx 0, \quad \mathcal{J}_B\approx 0,  \quad \mathcal{J}^i_V \approx -D_{j}\big((\Delta-3)\,\mathcal{H}^{(ij)}\big)+\Delta v^i +v^i.
\end{equation}

In the end, the Curtright charge in $D=5$ is 
\begin{equation}
\begin{aligned}
Q_{\Lambda, \lambda}[S^{3}_u]=Q_{\Phi}+Q_V+Q_{y^\text{TT}},
    \end{aligned}
\end{equation}  
where 
\begin{equation}
\begin{aligned}
   & Q_{\Phi}:=\oint_{S^3}\Phi\epsilon_{ijk}D^k\mathcal{H}^{[ij]}d\Omega,\\
   &Q_V:=-\oint_{S^3} V_i\bigg(D_{j}\big((\Delta-3)\,\mathcal{H}^{(ij)}\big)-\Delta v^i -v^i\bigg)d\Omega, \\
   &Q_{y^\text{TT}}:=\oint_{S^3}y_{ij}^\text{TT}(\Delta-3)\,\mathcal{H}^{(ij)}d\Omega.
\end{aligned}    
\end{equation}
This splitting of the charge is due to our de Donder plus traceless gauge fixing, to equations of motion \eqref{eqfix} and global geometric/topological conditions such as $H^1_{\text{dR}}(S^3)=H^2_{\text{dR}}(S^3)=0$ which allow Hodge and Hodge-like decompositions without harmonic components. We have so a charge whose parameter is the scalar arbitrary function $\Phi \in C^{\infty}(S^3)$, a charge whose parameter is a vector field $V^i  \in \mathfrak{Diff}(S^3)$ on $S^3$ and a charge whose parameter is a rank-2 transverse-traceless tensor $y_{ij}^\text{TT} \in \mathfrak{{TT}}(S^3)$ on $S^3$. Moreover, the charge $Q_{\Phi}$ can be rewritten using exterior calculus as
\begin{equation}
    Q_{\Phi}=2\oint_{S^3}\Phi \wedge d\mathcal{H}
\end{equation}
where we used the relation\footnote{This relation can be shown starting from 
$de^i + \omega^i{}_k \wedge e^k \;=\; 0 \,$ and using that $\mathcal{H} \;:=\; \frac12\,\mathcal{H}_{ij}\,e^i \wedge e^j \, $ since in $Q_{\Phi}$ only the antisymmetric part of $\mathcal{H}^{ij}$ is relevant. By explicit valuation of $d\mathcal{H}$ and using that on $D=3$ we have $e^k \wedge e^i \wedge e^j \;=\; \varepsilon^{kij}\,d\Omega$, we find $d\mathcal{H}
\;=\;
\frac12\,\varepsilon^{kij}\,D_k\mathcal{H}_{ij}\,d\Omega.$} $2d\mathcal{H}
\;=\;
\varepsilon_{ijk}\,D^k \mathcal{H}^{[ij]}\,d\Omega.$

\begin{comment}

% Calcolo di d\mathcal{H}:
\[
d\mathcal{H}
= \frac12\,d(\mathcal{H}_{ij}) \wedge e^i \wedge e^j
+\frac12\,\mathcal{H}_{ij}\,de^i \wedge e^j
-\frac12\,\mathcal{H}_{ij}\,e^i \wedge de^j \, .
\]
% Usando (1), cioè de^i = -\omega^i{}_k \wedge e^k, si ottiene
\[
d\mathcal{H}
= \frac12\,d(\mathcal{H}_{ij}) \wedge e^i \wedge e^j
-\frac12\,\mathcal{H}_{ij}\,\omega^i{}_k \wedge e^k \wedge e^j
+\frac12\,\mathcal{H}_{ij}\,e^i \wedge \omega^j{}_k \wedge e^k \, .
\]
% Riordinando e riconoscendo la derivata covariante di \mathcal{H}_{ij},
\begin{equation}
D\mathcal{H}_{ij}
\;:=\;
d\mathcal{H}_{ij} + \omega_i{}^{k}\,\mathcal{H}_{kj} + \omega_j{}^{k}\,\mathcal{H}_{ik}
\;=\;
(D_k\mathcal{H}_{ij})\,e^k \, ,
\tag{3}
\end{equation}
% si ottiene
\begin{equation}
d\mathcal{H}
\;=\;
\frac12\,(D_k\mathcal{H}_{ij})\,e^k \wedge e^i \wedge e^j \, .
\tag{4}
\end{equation}
\end{comment}

Let us now compute the charge algebra which is  defined by the relations 
\begin{equation}
\begin{aligned}
   & \{Q_{\Phi_1}, Q_{\Phi_2}\}= \delta_{\Phi_1} Q_{\Phi_2}, \quad \{Q_{V}, Q_{W}\} = \delta_{V} Q_{W}, \quad \{Q_{V}, Q_{\Phi}\} = \delta_{V} Q_{\Phi}, \\
   &\{Q_{y^\text{TT}}, Q_{\Phi}\}= \delta_{y^\text{TT}} Q_{\Phi}, \quad \{Q_{V}, Q_{y^\text{TT}}\} = \delta_{V} Q_{y^\text{TT}}, \quad \{Q_{y_1^\text{TT}}, Q_{y_2^\text{TT}}\} = \delta_{y_1^\text{TT}} Q_{y_2^\text{TT}}.
\end{aligned}
\end{equation}

\begin{enumerate}
    \item The computation of $\delta_{\Phi_1} Q_{\Phi_2}$ reads
    \begin{equation}
        \delta_{\Phi_1} Q_{\Phi_2}=2\oint_{S^3}\Phi_2 \wedge d(\delta_{\Phi_1}\mathcal{H})=0,
    \end{equation}
    since the gauge variation of $H_{ru\tilde{i}\tilde{j}}$ contains only $\Lambda_{ru}$ so it does not contain $\Phi$. For similar reasons also $\delta_{y^\text{TT}} Q_{\Phi}=0$ and $\delta_{y_1^\text{TT}} Q_{y_2^\text{TT}}=0$ since the gauge variation does not contains these pieces of the parameters.
    \item For the computation of $\delta_{V} Q_{W}$ we have
    \begin{equation}
    \begin{aligned}
        \delta_{V} Q_{W}&=-\oint_{S^3}W_i\delta_{V}(\mathcal{J}^id\Omega)=\oint_{S^3}W_i(\mathcal{L}_{V}(\mathcal{J}^id\Omega))=-\oint_{S^3}(\mathcal{L}_{V}W_i)\mathcal{J}^id\Omega\\
        &=Q_{[V,W]}.
    \end{aligned}
    \end{equation}
    where $\mathcal{L}_V$ is the Lie derivative with respect to $V$. By a similar computation we get also
    \begin{equation}
        \delta_{V} Q_{y^\text{TT}}=Q_{\mathcal{L}_V(y^{\text{TT}})},
    \end{equation}
    however, to ensure that $\mathcal{L}_V(y^{\text{TT}})$ is itself transverse-traceless we need $V$ to be a Killing vector, i.e. an isometry of $S^3$. Therefore, to close the algebra we need $V^i \in \mathfrak{o}(4)$. 
    \item Let us compute $\delta_V Q_\Phi$:
\begin{equation}
    \delta_V Q_\Phi = \oint_{S^3} \Phi \wedge d(\delta_V \mathcal{H}) =  \oint_{S^3} \Phi \wedge d(-\mathcal{L}_V \mathcal{H})= \oint_{S^3} \mathcal{L}_V(\Phi) \wedge d(\mathcal{H})=Q_{\mathcal{L}_V(\Phi)}
\end{equation}
where we used that the variation along a vector field of a tensor is its Lie derivative with respect to the vector field, that Lie derivative and exterior derivative commute and integrating by parts. Hence,
\begin{equation}
    \{Q_{V}, Q_{\Phi}\}=Q_{\mathcal{L}_V(\Phi)}
\end{equation}
and we have a semi-direct sum algebra.
\end{enumerate}
The final algebra is 
\begin{equation}
    \mathfrak{CBMS}(S^3):=\mathfrak{o}(4) \loplus (C^{\infty}(S^3)\oplus \mathfrak{TT}(S^3)),
    \label{alg}
\end{equation}
that resemble $\mathfrak{BMS}$ algebra where supertranslation-like transformations are parameterized by $\Phi$ while rotation transformations are parameterized by $V_i$. The Curtrigth-BMS algebra \eqref{alg} forms an abelian extension of a $\mathfrak{MBS}$-like algebra; however by Lie algebra cohomology this extension can be only trivial, in the sense that the extension splits as a semidirect sum. The $\mathfrak{TT}(S^3)$ term can be interpreted as a higher spin-like supertranslation since give rise to an abelian ideal that is a module for $\mathfrak{o}(4)$.\\
Despite this, within the boundary conditions and asymptotic gauge fixing adopted here, the leading charge analysis yields a single independent scalar supertranslation-like parameter. This should be viewed as a feature of the present framework, rather than as a definitive exclusion of a second scalar mode suggested in other approaches, such as Hamiltonian treatments. The missing second scalar could be suppressed by the specific Curtright gauge-fixing conditions used to define $\mathcal{H}^{ij}$. A smarter choice of gauge fixing or gauge-for-gauge fixing, as well as an analysis containing subleading or logarithmic expansion terms could be the right path to unveil the supertranslations structure. Accordingly, our analysis does not rule out the existence of an additional scalar parameter in a different asymptotic setup. Rather, it indicates that such a mode is not captured by the specific boundary conditions, gauge choices and leading-order truncation adopted in the present work. Moreover, in order to obtain extended or generalized Curtrigth versions of $\mathfrak{BMS}$ algebra we probably need to relax boundary constraints, much like of what happens in $D=4$. In the following table we make a comparison between the graviton and the Curtright field in $D=5$.

\renewcommand{\arraystretch}{1.25}

\begin{longtable}{|p{3.2cm}|p{4.4cm}|p{6.2cm}|}
\caption{Comparison between the graviton and the Curtright field in \( D=5 \).}
\endfirsthead
\hline
\textbf{Aspect} & \textbf{Graviton $h_{(\mu \nu)}$} & \textbf{Curtright field $\phi_{[\mu \nu]\rho}$} \\
\hline
Off shell representation & Symmetric tensor of \( \mathrm{GL}(1,4) \) with Young diagram (1,1)
&  Mixed symmetry of \( \mathrm{GL}(1,4) \) with Young diagram (2,1)\\
\hline
On shell (Lorentz) representation & Symmetric traceless tensor of \( \mathrm{SO}(3) \) with Young diagram (1,1)
&  Mixed symmetry traceless of \( \mathrm{SO}(3) \) with Young diagram (2,1)\\
\hline
Index symmetry &  Symmetric rank--2 tensor
\[
h_{\mu \nu}=h_{\nu \mu}
\]
&  Mixed-(2,1)-symmetry rank--3 tensor
\[
\phi_{[\mu \nu]\rho}=-\phi_{[\nu \mu]\rho}, \qquad
\phi_{[\mu \nu \rho]}=0
\]\\
\hline

Gauge parameter & One gauge parameter
\[
\xi_{\mu}
\]
with
\[
\delta_{\xi} h_{\mu \nu}=\partial_{\mu}\xi_{\nu}+\partial_{\nu}\xi_{\mu}
\]
& Two independent gauge parameters
\[
\lambda_{\mu \nu}\equiv\lambda_{(\mu \nu)}, \qquad \Lambda_{\mu \nu}\equiv\Lambda_{[\mu \nu]}
\]
with
\[
\delta_{\lambda,\Lambda} \phi_{[\mu \nu]\rho}
=
\partial_{[\mu}\lambda_{\nu] \rho}
+\partial_{[\mu}\Lambda_{\nu] \rho}
+2\partial_{\rho}\Lambda_{\nu \mu}
\]
\\
\hline
Gauge-for-gauge & No non-trivial gauge-for-gauge & First-stage gaue-for-gauge
\[
\Lambda'_{[\mu \nu]}= \Lambda_{[\mu \nu]}+\partial_{\mu}\Theta_{\nu}-\partial_{\nu}\Theta_{\mu},
\]
\[
\lambda'_{(\mu \nu)}= \lambda_{(\mu \nu)}+3\partial_{\mu}\Theta_{\nu}+3\partial_{\nu}\Theta_{\mu}
\]
that leave \( \delta\phi_{[\mu \nu]\rho} \) unchanged \\
\hline    
Duality relation & Standard metric-like formulation of linearized gravity& Hodge dual/dual representation description of the graviton \\
\hline
Off shell and on shell degrees of freedom & \begin{itemize}
    \item Off shell: \[
\frac{D(D+1)}{2}\overset{\text{$D=5$}}{=}15
\]
    \item On shell: \[
\frac{D(D-3)}{2}\overset{\text{$D=5$}}{=}5
\]
\end{itemize}
 & \begin{itemize}
   \item Off shell: \[
\frac{D(D^2-1)}{3}\overset{\text{$D=5$}}{=}40
\]
    \item On shell: \[
\frac{D(D-2)(D-4)}{3}\overset{\text{$D=5$}}{=}5
\]
\end{itemize}
\\
\hline
Duality relation & Standard metric-like formulation of linearized gravity& Hodge dual/dual representation description of the graviton \\
\hline
Basic asymptotic symmetry algebra & $$\mathfrak{o}(4) \loplus (C^{\infty}(S^3)\oplus C^{\infty}(S^3))$$ & $$\mathfrak{o}(4) \loplus (C^{\infty}(S^3)\oplus \mathfrak{TT}(S^3))$$\\
\hline
\end{longtable}

\section{Conclusions and outlook}

This work studied the asymptotic charge structure of the Curtright field, a prototypical gauge system with mixed index symmetry. Mixed-symmetry fields bring in genuinely new ingredients: there are multiple gauge parameters with different symmetry structure and there is complex gauge-for-gauge redundancy. The analysis developed a concrete and workable framework in which all these features can be controlled simultaneously, leading to finite charges and a closed asymptotic symmetry algebra.

A first major step was the selection of gauge fixing and the imposition of radiation-compatible fall-off conditions. These fall-offs were not chosen arbitrarily: they were motivated by the requirement that the radiative degrees of freedom have a well-defined flux through null infinity. In this setting, the equations of motion were simplified by imposing a de Donder-like gauge condition, supplemented by an on-shell condition that removes trace components. Conceptually, these conditions play the same role as the familiar gauge choices in linearized gravity: they isolate the propagating content, reduce the dynamics to a tractable wave-like form and turn the problem of asymptotic symmetries into a problem of characterizing the residual gauge transformations compatible with both the gauge choice and the boundary behavior.\\
A technically important part of the work was that to capture all potentially charge-relevant modes, the analysis allowed for more general asymptotic expansions, including half-integer powers and possible logarithmic branches. This polyhomogeneous approach is a natural generalization of what appears in higher-dimensional analyses of other gauge theories. With the residual transformations under control, the work constructed the associated surface charges using a covariant Noether-based strategy; charges are "electric-like" in the sense that they are built from radial--null components of the field strength contracted with asymptotic gauge data on the celestial sphere. This parallels the structure familiar from other gauge theories like $p$-forms one, while also reflecting the extra algebraic projections and redundancies intrinsic to mixed-symmetry fields.

A special focus of the work was the five-dimensional case, where the Curtright field is on-shell dual to the graviton. Using appropriate decompositions, the residual parameters can be organized into three sectors: one controlled by an arbitrary function on the three-sphere, which plays the role of a supertranslation-like parameter, a second controlled by a vector field on the three-sphere and a third controlled by TT modes. After an additional asymptotic gauge fixing that enforces transversality and removes trace pieces in the leading boundary tensor, the charge correspondingly splits into three contributions associated with these three sectors. This decomposition is valuable not only because it yields explicit expressions for the charges but also because it makes the symmetry algebra transparent forming an abelian extension of a $\mathfrak{BMS}$-like algebra. The scalar and TT sectors form abelian ideals, while the vector sector reproduces the standard commutator structure of vector fields on the three-sphere generating the $\mathfrak{o}(4)$ part of the algebra. The mixed bracket is the natural action of sphere diffeomorphisms on tensors. Altogether, the algebra closes as a semidirect sum between the Killing vectors of the celestial sphere and an abelian ideal given by the direct sum of scalar and TT sectors interpreted as supertranslations and higher spin supertranslations. This is the mixed-symmetry analogue of the familiar BMS-type structure, we call Curtright extension of $\mathfrak{BMS}$ algebra: $\mathfrak{CBMS}$.\\
One conceptual tension that emerged concerns the scalar sector in five dimensions. Within the adopted boundary conditions and asymptotic gauge fixing, only one independent supertranslation-like function survives in the charge. This contrasts with some expectations in the literature, where two independent scalar functions have been argued to arise at leading order in related Hamiltonian analyses. In the present framework, the "missing" scalar mode may be removed by the specific asymptotic gauge choice used to define the boundary tensor, or it may reside in subleading or logarithmic branches that were consistently truncated when isolating the minimal data needed for finite leading charges.

Several natural directions for future work follow from these results. A first direction is to revisit the boundary conditions and asymptotic gauge choices in five dimensions with the explicit goal of clarifying the status of the second scalar mode. This would involve relaxing some of the asymptotic constraints imposed on the leading boundary tensor, exploring alternative gauge-for-gauge fixings and systematically retaining a larger set of subleading terms in the polyhomogeneous expansions. Moreover, the study of possible Curtright extension of extended or generalized BMS-structure is a an interesting possibility: like in $D=4$, these extensions could appear relaxing fall-of conditions on fields.\\
A second direction is to extend the construction of asymptotic charges to broader families of mixed-symmetry gauge fields. The Curtright field is the simplest nontrivial example but the methods developed here suggest a general blueprint: choose a gauge that reduces the dynamics to a wave-like form, impose irreducibility conditions that isolate the correct Young-projected degrees of freedom and classify residual parameters on the celestial sphere while carefully accounting for gauge-for-gauge redundancies. In that context, one expects surface charges to be built from the appropriate radial--null curvature components contracted with the leading sphere data of the residual parameters, generalizing the "electric-like" pattern found here and in $p$-forms contexts. Carrying this out for other Young projected tensors would provide a systematic answer to how asymptotic symmetry algebras depend on the representation content of the field and it would establish a general theory of large gauge charges for tensors with mixed symmetries, as already anticipated and discussed in the main text.\\
A third direction concerns duality at null infinity. While the Curtright field and the graviton are dual on-shell in five dimensions, their gauge structures differ off-shell and their boundary descriptions need not coincide automatically. A detailed dictionary between metric boundary data and the dual mixed-symmetry boundary tensor would clarify which aspects of the infrared structure are invariant under duality and which are reshuffled. This includes not only the mapping of charges and symmetry parameters, but also the mapping of fluxes, memory observables, and the interpretation of radiative data. In particular, it would be valuable to understand whether the semidirect-sum algebra obtained here is exactly the dual image of the gravitational asymptotic algebra for the same class of fall-offs, or whether additional boundary counterterms and corner contributions are needed to match the two descriptions fully.\\
A fourth direction is to investigate related soft theorem and memory effects. Establishing the existence and the precise form of soft theorems and memory effects associated with these mixed-symmetry gauge systems. As in other gauge theories, large gauge transformations at null infinity are expected to yield non-trivial Ward identities for the scattering operator, which in turn should control universal soft limits of amplitudes with an additional soft particle. Such effects could be relevant both conceptually and technically: conceptually, they generalize the IR triangle relating symmetries, soft limits, and memory beyond symmetric tensor gauge fields; technically, they provide new constraints on consistent couplings, boundary conditions, and the definition of conserved quantities for mixed-symmetry fields. Moreover, since hook-type fields appear in higher-spin theories, duality-covariant formulations and in certain stringy regimes, understanding their IR sector may uncover universal features shared across these frameworks.

\appendix

\section{Equations}\label{AppB}

\subsection{Equations from residual gauge conditions}\label{B1}

% ------------------------------------------------------------
% Covariant residual gauge conditions in arbitrary D
% (retarded Bondi-like coordinates with given metric + Christoffels)
%
% Coordinates: (u,r,x^i),  i=1,...,D-2, with γ_{ij}(x) the unit S^{D-2} metric.
%
% Metric and inverse (as given):
%   g_{uu}=-1, g_{ur}=g_{ru}=-1, g_{rr}=0, g_{ij}=r^2 γ_{ij},
%   g^{uu}=0,  g^{ur}=g^{ru}=-1, g^{rr}=1, g^{ij}=r^{-2}γ^{ij}.
%
% Non-zero Christoffels (as given / implied):
%   Γ^i_{jr}=Γ^i_{rj}=(1/r) δ^i_j,
%   Γ^u_{ij}= r γ_{ij},    Γ^r_{ij}= -r γ_{ij},
%   Γ^k_{ij}=Γ^k_{ij}(γ)  (Levi-Civita of γ on S^{D-2}).
% All other Γ are taken to vanish.
%
% Angular covariant derivative D_i is the Levi-Civita derivative of γ_{ij}.
% Laplacian on S^{D-2}:  Δ := γ^{ij} D_i D_j.
% -----------------------------------------------------------

Assume the $r$ polyhomogeneous expansions
\begin{equation}
\lambda_{\mu\nu}
=\sum_{l\in\frac12\mathbb Z} r^{-l}\Bigl(\lambda_{\mu\nu}^{(l)}(u,x)
+\bar\lambda_{\mu\nu}^{(l)}(u,x)\ln r\Bigr),
\qquad
\Lambda_{\mu\nu}
=\sum_{l\in\frac12\mathbb Z} r^{-l}\Bigl(\Lambda_{\mu\nu}^{(l)}(u,x)
+\bar\Lambda_{\mu\nu}^{(l)}(u,x)\ln r\Bigr).
\label{exp_lam_Lam_covD}
\end{equation}
The residual gauge conditions are
\begin{equation}
G_{\beta}:=\nabla_{\beta}\lambda^{\mu}{}_{\mu}
-\nabla^{\alpha}\lambda_{\alpha\beta}
+\nabla^{\alpha}\Lambda_{\alpha\beta}=0,
\qquad
S_{\alpha}:=\Box\bigl(\nabla^{\beta}\Lambda_{\beta\alpha}\bigr)=0,
\label{resgauge_covD}
\end{equation}
where $\Box:=\nabla^\mu\nabla_\mu$.
\paragraph{Notation.} We adopt the following notations where $n:=D-2$
\begin{equation}
(\partial_r \lambda_{\mu \nu})^{(l)}:=-(l-1)\lambda_{\mu \nu}^{(l-1)}+\bar \lambda_{\mu \nu}^{(l-1)},
\qquad
\overline{(\partial_r \lambda_{\mu \nu})}^{(l)}:=-(l-1)\bar \lambda_{\mu \nu}^{(l-1)},
\label{convdevrl}
\end{equation}
\begin{equation}
(\partial_r \Lambda_{\mu \nu})^{(l)}:=-(l-1)\Lambda_{\mu \nu}^{(l-1)}+\bar \Lambda_{\mu \nu}^{(l-1)},
\qquad
\overline{(\partial_r \Lambda_{\mu \nu})}^{(l)}:=-(l-1)\bar \Lambda_{\mu \nu}^{(l-1)},
\label{convdevrL}
\end{equation}
\begin{equation}
(\lambda^\mu{}_\mu)^{(l)}:=-2\lambda_{ur}^{(l)}+\lambda_{rr}^{(l)}+\gamma^{ij}\lambda_{ij}^{(l-2)},
\qquad
\overline{(\lambda^\mu{}_\mu)}^{(l)}:=-2\bar\lambda_{ur}^{(l)}+\bar\lambda_{rr}^{(l)}+\gamma^{ij}\bar\lambda_{ij}^{(l-2)}.
\label{convtrace}
\end{equation}
\begin{equation}
    (\partial_r\lambda^\mu{}_\mu)^{(l)}:=-(l-1)(\lambda^\mu{}_\mu)^{(l-1)}+\overline{(\lambda^\mu{}_\mu)}^{(l-1)}, \qquad \overline{(\partial_r\lambda^\mu{}_\mu)}^{(l)}:=-(l-1)\overline{(\lambda^\mu{}_\mu)}^{(l-1)}
\end{equation}
\begin{align}
(\nabla^\beta\Lambda_{\beta u})^{(l)}:={}&
\partial_u\Lambda_{ur}^{(l)}
+(l-1-n)\Lambda_{ur}^{(l-1)}
-\bar\Lambda_{ur}^{(l-1)},\\
\overline{(\nabla^\beta\Lambda_{\beta u})}^{(l)}:={}&
\partial_u\bar\Lambda_{ur}^{(l)}
+(l-1-n)\bar\Lambda_{ur}^{(l-1)} .
\end{align}
\begin{align}
(\nabla^\beta\Lambda_{\beta r})^{(l)}:={}&
(l-1-n)\Lambda_{ur}^{(l-1)}
-\bar\Lambda_{ur}^{(l-1)},\\
\overline{(\nabla^\beta\Lambda_{\beta r})}^{(l)}:={}&
(l-1-n)\bar\Lambda_{ur}^{(l-1)} .
\end{align}
\begin{align}
(\nabla^\beta\Lambda_{\beta i})^{(l)}:={}&
 D^j\Lambda_{ji}^{(l-2)},\\
\overline{(\nabla^\beta\Lambda_{\beta i})}^{(l)}:={}& D^j\bar\Lambda_{ji}^{(l-2)} ,
\end{align}
Moreover, we use $D_i$ to indicate the covariant derivative on $S^{D-2}$, so $\Delta:= D_iD^i$ is the Laplace-Beltrami operator on $S^{D-2}$.

\paragraph{Equation \(G_u=0\).}
\begin{align}
&
+\partial_u(\lambda^\mu{}_\mu)^{(l)}
+\partial_u\lambda_{ur}^{(l)}
+\partial_u\Lambda_{ur}^{(l)}
+(n-l+1)\lambda_{uu}^{(l-1)}
+(l-1-n)\lambda_{ur}^{(l-1)}+
\nonumber\\
&
+(l-1-n)\Lambda_{ur}^{(l-1)}
+\bar\lambda_{uu}^{(l-1)}
-\bar\lambda_{ur}^{(l-1)}
-\bar\Lambda_{ur}^{(l-1)}=0 .
\end{align}
\begin{align}
&
+\partial_u\overline{(\lambda^\mu{}_\mu)}^{(l)}
+\partial_u\bar\lambda_{ur}^{(l)}
+\partial_u\bar\Lambda_{ur}^{(l)}
+(n-l+1)\bar\lambda_{uu}^{(l-1)}
+(l-1-n)\bar\lambda_{ur}^{(l-1)}+
\nonumber\\
&
+(l-1-n)\bar\Lambda_{ur}^{(l-1)}=0 .
\end{align}

\paragraph{Equation \(G_r=0\).}
\begin{align}
&
-(l-1)(\lambda^\mu{}_\mu)^{(l-1)}
+\overline{(\lambda^\mu{}_\mu)}^{(l-1)}
+\partial_u\lambda_{rr}^{(l)}
+(n-l+1)\lambda_{ur}^{(l-1)}
+(l-1-n)\lambda_{rr}^{(l-1)}+
\nonumber\\
&
+(l-1-n)\Lambda_{ur}^{(l-1)}
+\gamma^{ij}\lambda_{ij}^{(l-3)}+
+\bar\lambda_{ur}^{(l-1)}
-\bar\lambda_{rr}^{(l-1)}
-\bar\Lambda_{ur}^{(l-1)}=0.
\end{align}
\begin{align}
&
-(l-1)\overline{(\lambda^\mu{}_\mu)}^{(l-1)}
+\partial_u\bar\lambda_{rr}^{(l)}
+(n-l+1)\bar\lambda_{ur}^{(l-1)}
+(l-1-n)\bar\lambda_{rr}^{(l-1)}+
\nonumber\\
&
+(l-1-n)\bar\Lambda_{ur}^{(l-1)}
+\gamma^{ij}\bar\lambda_{ij}^{(l-3)}=0 .
\end{align}

\paragraph{Equations \(G_i=0\).}
\begin{align}
&
D_i(\lambda^\mu{}_\mu)^{(l)}
- D^j\lambda_{ji}^{(l-2)}
+ D^j\Lambda_{ji}^{(l-2)}=0 .
\end{align}
\begin{align}
D_i\overline{(\lambda^\mu{}_\mu)}^{(l)}
- D^j\bar\lambda_{ji}^{(l-2)}
+ D^j\bar\Lambda_{ji}^{(l-2)}=0 .
\end{align}

\paragraph{Equation $S_u=0$.}
\begin{align}
&+\bigl(2(l-1)-n\bigr)\partial_u(\nabla^\beta\Lambda_{\beta u})^{(l-1)}
-2\partial_u\overline{(\nabla^\beta\Lambda_{\beta u})}^{(l-1)}+\\
&
+(l-2)\bigl((l-1)-n\bigr)(\nabla^\beta\Lambda_{\beta u})^{(l-2)}+\\
&
+\bigl(n-2l+3\bigr)\overline{(\nabla^\beta\Lambda_{\beta u})}^{(l-2)}
+\Delta(\nabla^\beta\Lambda_{\beta u})^{(l-2)} =0,
\end{align}
\begin{align}
&+\bigl(2(l-1)-n\bigr)\partial_u\overline{(\nabla^\beta\Lambda_{\beta u})}^{(l-1)}
+(l-2)\bigl((l-1)-n\bigr)\overline{(\nabla^\beta\Lambda_{\beta u})}^{(l-2)}+
\nonumber\\
&
+\Delta\,\overline{(\nabla^\beta\Lambda_{\beta u})}^{(l-2)}=0 .
\end{align}

\paragraph{Equation $S_r=0$.}
\begin{align}
&+\bigl(2(l-1)-n\bigr)\partial_u(\nabla^\beta\Lambda_{\beta r})^{(l-1)}
-2\partial_u\overline{(\nabla^\beta\Lambda_{\beta r})}^{(l-1)}+
\nonumber\\
&
+(l-2)\bigl((l-1)-n\bigr)(\nabla^\beta\Lambda_{\beta r})^{(l-2)}
+\bigl(n-2l+3\bigr)\overline{(\nabla^\beta\Lambda_{\beta r})}^{(l-2)}
+\Delta(\nabla^\beta\Lambda_{\beta r})^{(l-2)}+
\nonumber\\
&
-2 D^i(\nabla^\beta\Lambda_{\beta i})^{(l-3)}
+n\Bigl[(\nabla^\beta\Lambda_{\beta u})^{(l-2)}
-(\nabla^\beta\Lambda_{\beta r})^{(l-2)}\Bigr]=0 ,
\end{align}
\begin{align}
&+\bigl(2(l-1)-n\bigr)\partial_u\overline{(\nabla^\beta\Lambda_{\beta r})}^{(l-1)}
+(l-2)\bigl((l-1)-n\bigr)\overline{(\nabla^\beta\Lambda_{\beta r})}^{(l-2)}+
\nonumber\\
&
+\Delta\,\overline{(\nabla^\beta\Lambda_{\beta r})}^{(l-2)}
-2 D^i\overline{(\nabla^\beta\Lambda_{\beta i})}^{(l-3)}
+n\Bigl[\overline{(\nabla^\beta\Lambda_{\beta u})}^{(l-2)}
-\overline{(\nabla^\beta\Lambda_{\beta r})}^{(l-2)}\Bigr]=0.
\end{align}

\paragraph{Equation $S_i=0$.}
\begin{align}
&+\bigl(2(l-1)-(n-2)\bigr)\partial_u(\nabla^\beta\Lambda_{\beta i})^{(l-1)}
-2\partial_u\overline{(\nabla^\beta\Lambda_{\beta i})}^{(l-1)}+
\nonumber\\
&
+(l-2)\bigl((l-1)-(n-2)\bigr)(\nabla^\beta\Lambda_{\beta i})^{(l-2)}
+\bigl((n-2)-2l+3\bigr)\overline{(\nabla^\beta\Lambda_{\beta i})}^{(l-2)}+
\nonumber\\
&
+ D^jD_j(\nabla^\beta\Lambda_{\beta i})^{(l-2)}
-(n-1)(\nabla^\beta\Lambda_{\beta i})^{(l-2)}
-2 D_i(\nabla^\beta\Lambda_{\beta u})^{(l-1)}+\\
&
-2 D_i(\nabla^\beta\Lambda_{\beta r})^{(l-1)}=0 ,
\end{align}
\begin{align}
&+\bigl(2(l-1)-(n-2)\bigr)\partial_u\overline{(\nabla^\beta\Lambda_{\beta i})}^{(l-1)}
+(l-2)\bigl((l-1)-(n-2)\bigr)\overline{(\nabla^\beta\Lambda_{\beta i})}^{(l-2)}+
\nonumber\\
&
+ D^jD_j\,\overline{(\nabla^\beta\Lambda_{\beta i})}^{(l-2)}
-(n-1)\overline{(\nabla^\beta\Lambda_{\beta i})}^{(l-2)}+
\nonumber\\
&
-2 D_i\overline{(\nabla^\beta\Lambda_{\beta u})}^{(l-1)}
+2 D_i\overline{(\nabla^\beta\Lambda_{\beta r})}^{(l-1)}=0 .
\end{align}

\subsection{Equations from fall-off conditions}\label{B2}
The fall-off conditions in our gauge and gauge-for-gauge fixing reads as
\subsubsection*{Case I}
\begin{equation}
\begin{aligned}
&\delta\!\bigl(\phi_{[u r]u}\bigr)
= \partial_u \lambda_{(r u)} - \partial_r \lambda_{(u u)} + 3\,\partial_u \Lambda_{[r u]}
\sim \mathcal{O}\!\left(r^{-\frac{D-2}{2}}\right) \, , \\
&\delta\!\bigl(\phi_{[u r]r}\bigr)
= \partial_u \lambda_{(r r)} - \partial_r \lambda_{(u r)} - 3\,\partial_r \Lambda_{[u r]}
\sim \mathcal{O}\!\left(r^{-\frac{D-2}{2}}\right) \, .
\end{aligned}
\end{equation}

\subsubsection*{Case II}
\begin{equation}
\begin{aligned}
&\delta\!\bigl(\phi_{[u i]u}\bigr)
= - \partial_i \lambda_{(u u)} 
\sim \mathcal{O}\!\left(r^{-\frac{D-4}{2}}\right) \, , \\
&\delta\!\bigl(\phi_{[u i]r}\bigr)
=  - \partial_i \lambda_{(u r)} - \partial_i \Lambda_{[u r]}
\sim \mathcal{O}\!\left(r^{-\frac{D-4}{2}}\right) \, , \\
&\delta\!\bigl(\phi_{[r i]u}\bigr)
=- \partial_i \lambda_{(r u)} - \partial_i \Lambda_{[r u]}
\sim \mathcal{O}\!\left(r^{-\frac{D-4}{2}}\right) \, , \\
&\delta\!\bigl(\phi_{[r i]r}\bigr)
= - \partial_i \lambda_{(r r)} 
\sim \mathcal{O}\!\left(r^{-\frac{D-4}{2}}\right) \, .
\end{aligned}
\end{equation}

\subsubsection*{Case III}
\begin{equation}
\begin{aligned}
&\delta\!\bigl(\phi_{[u i]j}\bigr)
= + \partial_u \lambda_{(i j)}
  + \partial_u \Lambda_{[i j]}
\sim \mathcal{O}\!\left(r^{-\frac{D-6}{2}}\right) \, , \\
&\delta\!\bigl(\phi_{[r i]j}\bigr)
= + \partial_r \lambda_{(i j)}
  + \partial_r \Lambda_{[i j]}
\sim \mathcal{O}\!\left(r^{-\frac{D-6}{2}}\right) \, , \\
&\delta\!\bigl(\phi_{[i j]u}\bigr)
=  + 2 \partial_u \Lambda_{[j i]}
\sim \mathcal{O}\!\left(r^{-\frac{D-6}{2}}\right) \, , \\
&\delta\!\bigl(\phi_{[i j]r}\bigr)
=
  + 2 \partial_r \Lambda_{[j i]}
\sim \mathcal{O}\!\left(r^{-\frac{D-6}{2}}\right) \, .
\end{aligned}
\end{equation}

\subsubsection*{Case IV}
\begin{equation}
\begin{aligned}
&\delta\!\bigl(\phi_{[i j]k}\bigr)
= \nabla_i \lambda_{(j k)} - \nabla_j \lambda_{(i k)}
  + \nabla_i \Lambda_{[j k]} - \nabla_j \Lambda_{[i k]}
  + 2 \nabla_k \Lambda_{[j i]}
\sim \mathcal{O}\!\left(r^{-\frac{D-8}{2}}\right) \, .
\end{aligned}
\end{equation}

Assuming the polyhomogeneous $r$ expansions for the gauge parameters
\begin{equation}
\lambda_{\mu\nu}
=\sum_{l\in\frac12\mathbb Z} r^{-l}
\bigl(\lambda_{\mu\nu}^{(l)}+\bar\lambda_{\mu\nu}^{(l)}\ln r\bigr),
\qquad
\Lambda_{\mu\nu}
=\sum_{l\in\frac12\mathbb Z} r^{-l}
\bigl(\Lambda_{\mu\nu}^{(l)}+\bar\Lambda_{\mu\nu}^{(l)}\ln r\bigr).
\end{equation}
and defining
\[
p_1=\frac{D-2}{2},\qquad
p_2=\frac{D-4}{2},\qquad
p_3=\frac{D-6}{2},\qquad
p_4=\frac{D-8}{2},
\]
the recursive equations are

%%%%%%%%%%%%%%%%%%%%%%%%%%%%%%%%%%%%%%%%%%%%%%%%%%%%%%%%%%%%
\subsubsection*{Case I}
\begin{equation}
\partial_u \bar\lambda_{(ru)}^{(l)}
+3\,\partial_u \bar\Lambda_{[ru]}^{(l)}
+(l-1)\,\bar\lambda_{(uu)}^{(l-1)}
=0,
\qquad \forall\, l<p_1 .
\end{equation}

\begin{equation}
\partial_u \lambda_{(ru)}^{(l)}
+3\,\partial_u \Lambda_{[ru]}^{(l)}
+(l-1)\,\lambda_{(uu)}^{(l-1)}
-\bar\lambda_{(uu)}^{(l)}
=0,
\qquad \forall\, l<p_1 .
\end{equation}

\begin{equation}
\partial_u \bar\lambda_{(rr)}^{(l)}
+(l-1)\,\bar\lambda_{(ur)}^{(l-1)}
+3(l-1)\,\bar\Lambda_{[ur]}^{(l-1)}
=0,
\qquad \forall\, l<p_1 .
\end{equation}

\begin{equation}
\partial_u \lambda_{(rr)}^{(l)}
+l\,\lambda_{(ur)}^{(l)}
+3l\,\Lambda_{[ur]}^{(l)}
-\bar\lambda_{(ur)}^{(l)}
-3\,\bar\Lambda_{[ur]}^{(l)}
=0,
\qquad \forall\, l<p_1 .
\end{equation}

%%%%%%%%%%%%%%%%%%%%%%%%%%%%%%%%%%%%%%%%%%%%%%%%%%%%%%%%%%%%
\subsubsection*{Case II}

\begin{equation}
\partial_i \lambda_{(uu)}^{(l)}=0,
\qquad
\partial_i \bar\lambda_{(uu)}^{(l)}=0,
\qquad \forall\, l<p_2 .
\end{equation}

\begin{equation}
\partial_i\!\left(\lambda_{(ur)}^{(l)}+\Lambda_{[ur]}^{(l)}\right)=0,
\qquad
\partial_i\!\left(\bar\lambda_{(ur)}^{(l)}+\bar\Lambda_{[ur]}^{(l)}\right)=0,
\qquad \forall\, l<p_2 .
\end{equation}

\begin{equation}
\partial_i\!\left(\lambda_{(ru)}^{(l)}+\Lambda_{[ru]}^{(l)}\right)=0,
\qquad
\partial_i\!\left(\bar\lambda_{(ru)}^{(l)}+\bar\Lambda_{[ru]}^{(l)}\right)=0,
\qquad \forall\, l<p_2 .
\end{equation}

\begin{equation}
\partial_i \lambda_{(rr)}^{(l)}=0,
\qquad
\partial_i \bar\lambda_{(rr)}^{(l)}=0,
\qquad \forall\, l<p_2 .
\end{equation}

%%%%%%%%%%%%%%%%%%%%%%%%%%%%%%%%%%%%%%%%%%%%%%%%%%%%%%%%%%%%
\subsubsection*{Case III}

\begin{equation}
\partial_u\!\left(\lambda_{(ij)}^{(l)}+\Lambda_{[ij]}^{(l)}\right)=0,
\qquad
\partial_u\!\left(\bar\lambda_{(ij)}^{(l)}+\bar\Lambda_{[ij]}^{(l)}\right)=0,
\qquad \forall\, l<p_3 .
\end{equation}

\begin{equation}
-l\,\left(\bar\lambda_{(ij)}^{(l)}+\bar\Lambda_{[ij]}^{(l)}\right)=0,
\qquad \forall\, l+1<p_3 .
\end{equation}

\begin{equation}
-l\,\left(\lambda_{(ij)}^{(l)}+\Lambda_{[ij]}^{(l)}\right)
+\left(\bar\lambda_{(ij)}^{(l)}+\bar\Lambda_{[ij]}^{(l)}\right)
=0,
\qquad \forall\, l+1<p_3 .
\end{equation}

\begin{equation}
\partial_u \Lambda_{[ji]}^{(l)}=0,
\qquad
\partial_u \bar\Lambda_{[ji]}^{(l)}=0,
\qquad \forall\, l<p_3 .
\end{equation}

\begin{equation}
-2l\,\bar\Lambda_{[ji]}^{(l)}=0,
\qquad \forall\, l+1<p_3 .
\end{equation}

\begin{equation}
-2l\,\Lambda_{[ji]}^{(l)}+2\,\bar\Lambda_{[ji]}^{(l)}=0,
\qquad \forall\, l+1<p_3 .
\end{equation}

%%%%%%%%%%%%%%%%%%%%%%%%%%%%%%%%%%%%%%%%%%%%%%%%%%%%%%%%%%%%
\subsubsection*{Case IV}

\begin{equation}
\nabla_i \bar\lambda_{(jk)}^{(l)}
-\nabla_j \bar\lambda_{(ik)}^{(l)}
+\nabla_i \bar\Lambda_{[jk]}^{(l)}
-\nabla_j \bar\Lambda_{[ik]}^{(l)}
+2\,\nabla_k \bar\Lambda_{[ji]}^{(l)}
=0,
\qquad \forall\, l<p_4 .
\end{equation}

\begin{equation}
\nabla_i \lambda_{(jk)}^{(l)}
-\nabla_j \lambda_{(ik)}^{(l)}
+\nabla_i \Lambda_{[jk]}^{(l)}
-\nabla_j \Lambda_{[ik]}^{(l)}
+2\,\nabla_k \Lambda_{[ji]}^{(l)}
=0,
\qquad \forall\, l<p_4 .
\end{equation}
\subsection{Equations of motion}\label{B3}

Making explicit the equations of motion for the independent field components, we have
\begin{equation}
\partial_r^2\phi_{[ur]u}
-2\partial_u\partial_r\phi_{[ur]u}
+\frac{D-2}{r}(\partial_r-\partial_u)\phi_{[ur]u}
+\frac{1}{r^2}\mathcal{D}^2\phi_{[ur]u}
-\frac{2}{r^3}\mathcal{D}^i\phi_{[u i] u}
-\frac{D-2}{r^2}\phi_{[ur]u}=0 ,
\label{eq:box_uru}
\end{equation}

\begin{equation}
\partial_r^2 \phi_{[ur]r}-2\,\partial_u\partial_r \phi_{[ur]r}
+\frac{1}{r^2}\,\mathcal{D}^2 \phi_{[ur]r}
-\frac{1}{r^3}\,\mathcal{D}^i\!\Big(\phi_{[ui]r}+\phi_{[ur]i}\Big)
+\frac{2}{r^4}\,\gamma^{ij}\phi_{[ui]j}=0\, .
\end{equation}
\begin{equation}
\partial_r^2 \phi_{[ur]i}-2\,\partial_u\partial_r \phi_{[ur]i}
-\frac{2}{r}\,\partial_r \phi_{[ur]i}
+\frac{2}{r}\,\partial_u \phi_{[ur]i}
+\frac{1}{r^2}\,\mathcal{D}^2 \phi_{[ur]i}
-\frac{1}{r^3}\,\mathcal{D}^k \phi_{[uk]i}
-\frac{D-4}{r^2}\,\phi_{[ur]i}=0\, .
\end{equation}

\begin{equation}
\begin{aligned}
&\partial_r^2 \phi_{[ui]u}
-2\,\partial_u\partial_r \phi_{[ui]u}
+\frac{D-4}{r}\,\partial_r \phi_{[ui]u}
-\frac{D-4}{r}\,\partial_u \phi_{[ui]u}
+\frac{1}{r^2}\,\mathcal{D}^2 \phi_{[ui]u}
-\frac{1}{r^3}\,\mathcal{D}^k \phi_{[uk]u}\\
&-\frac{D-3}{r^2}\,\phi_{[ui]u}=0\, .
\end{aligned}
\end{equation}

\begin{equation}
\begin{aligned}
&\partial_r^2 \phi_{[ui]r}
-2\,\partial_u\partial_r \phi_{[ui]r}
+\frac{D-2}{r}\,\partial_r \phi_{[ui]r}
-\frac{D-4}{r}\,\partial_u \phi_{[ui]r}
+\frac{1}{r^2}\,\mathcal{D}^2 \phi_{[ui]r}
-\frac{1}{r^3}\,\mathcal{D}^k \phi_{[uk]r}\\
&-\frac{D-2}{r^2}\,\phi_{[ui]r}
+\frac{1}{r^3}\,\mathcal{D}_i \phi_{[ur]r}=0 \, .
\end{aligned}
\end{equation}

\begin{equation}
\begin{aligned}
&\partial_r^2 \phi_{[ui]j}
-2\,\partial_u\partial_r \phi_{[ui]j}
+\frac{D-4}{r}\,\partial_r \phi_{[ui]j}
-\frac{D-6}{r}\,\partial_u \phi_{[ui]j}
+\frac{1}{r^2}\,\mathcal{D}^2 \phi_{[ui]j}
-\frac{D-3}{r^2}\,\phi_{[ui]j}
\\
&
-\frac{1}{r^3}\Big(
\mathcal{D}_j \phi_{[ui]r}
+\mathcal{D}_i \phi_{[uj]r}
-\gamma_{ij}\,\mathcal{D}^k \phi_{[uk]r}
\Big)
+\frac{2}{r^3}\,\mathcal{D}_{[i}\phi_{[u|j]u}=0 \, .
\end{aligned}
\end{equation}

\begin{comment}
\begin{equation}
\begin{aligned}
&\partial_r^2 \phi_{[ri]u}
-2\,\partial_u\partial_r \phi_{[ri]u}
+\frac{D-2}{r}\,\partial_r \phi_{[ri]u}
-\frac{D-2}{r}\,\partial_u \phi_{[ri]u}
+\frac{1}{r^2}\,\mathcal{D}^2 \phi_{[ri]u}
-\frac{D-2}{r^2}\,\phi_{[ri]u}\\
&
-\frac{1}{r^3}\Big(
\mathcal{D}_i \phi_{[ru]u}
+\mathcal{D}^k \phi_{[rk]u}
\Big)
+\frac{2}{r^3}\,\mathcal{D}_i \phi_{[ur]u}=0\, .
\end{aligned}
\end{equation}

\end{comment}

\begin{equation}
\begin{aligned}
&\partial_r^2 \phi_{[ri]r}
-2\,\partial_u\partial_r \phi_{[ri]r}
+\frac{D}{r}\,\partial_r \phi_{[ri]r}
-\frac{D-2}{r}\,\partial_u \phi_{[ri]r}
+\frac{1}{r^2}\,\mathcal{D}^2 \phi_{[ri]r}
-\frac{D}{r^2}\,\phi_{[ri]r}
\\
&
-\frac{1}{r^3}\,\mathcal{D}^k \phi_{[rk]r}
+\frac{1}{r^3}\,\mathcal{D}_i \phi_{[ur]r}
-\frac{2}{r^3}\,\mathcal{D}_i \phi_{[rr]u}=0 \, .
\end{aligned}
\end{equation}

\begin{equation}
\begin{aligned}
&\partial_r^2 \phi_{[ri]j}
-2\,\partial_u\partial_r \phi_{[ri]j}
+\frac{D-2}{r}\,\partial_r \phi_{[ri]j}
-\frac{D-4}{r}\,\partial_u \phi_{[ri]j}
+\frac{1}{r^2}\,\mathcal{D}^2 \phi_{[ri]j}
-\frac{D-2}{r^2}\,\phi_{[ri]j}
\\
&
-\frac{1}{r^3}\Big(
\mathcal{D}_j \phi_{[ri]r}
+\mathcal{D}_i \phi_{[rj]r}
-\gamma_{ij}\,\mathcal{D}^k \phi_{[rk]r}
\Big)
+\frac{1}{r^3}\,\mathcal{D}_j \phi_{[ui]r}
-\frac{1}{r^3}\,\mathcal{D}_i \phi_{[uj]r}=0 \, .
\end{aligned}
\end{equation}

\begin{comment}
\begin{equation}
\begin{aligned}
&\partial_r^2 \phi_{[ij]u}
-2\,\partial_u\partial_r \phi_{[ij]u}
+\frac{D-4}{r}\,\partial_r \phi_{[ij]u}
-\frac{D-6}{r}\,\partial_u \phi_{[ij]u}
+\frac{1}{r^2}\,\mathcal{D}^2 \phi_{[ij]u}
-\frac{2(D-3)}{r^2}\,\phi_{[ij]u}
\\
&
-\frac{1}{r^3}\Big(
\mathcal{D}_i \phi_{[ju]r}
-\mathcal{D}_j \phi_{[iu]r}
\Big)
+\frac{2}{r^3}\Big(
\mathcal{D}_i \phi_{[ju]u}
-\mathcal{D}_j \phi_{[iu]u}
\Big)=0 \, .
\end{aligned}
\end{equation}

\begin{equation}
\begin{aligned}
&\partial_r^2 \phi_{[ij]r}
-2\,\partial_u\partial_r \phi_{[ij]r}
+\frac{D-2}{r}\,\partial_r \phi_{[ij]r}
-\frac{D-4}{r}\,\partial_u \phi_{[ij]r}
+\frac{1}{r^2}\,\mathcal{D}^2 \phi_{[ij]r}
-\frac{2(D-2)}{r^2}\,\phi_{[ij]r}
\\
&
-\frac{1}{r^3}\Big(
\mathcal{D}_i \phi_{[jr]r}
-\mathcal{D}_j \phi_{[ir]r}
\Big)
+\frac{1}{r^3}\Big(
\mathcal{D}_i \phi_{[ju]r}
-\mathcal{D}_j \phi_{[iu]r}
\Big)=0 \, .
\end{aligned}
\end{equation}

\end{comment}

\begin{equation}
\begin{aligned}
&\partial_r^2 \phi_{[ij]k}
-2\,\partial_u\partial_r \phi_{[ij]k}
+\frac{D-4}{r}\,\partial_r \phi_{[ij]k}
-\frac{D-6}{r}\,\partial_u \phi_{[ij]k}
+\frac{1}{r^2}\,\mathcal{D}^2 \phi_{[ij]k}
-\frac{2(D-3)}{r^2}\,\phi_{[ij]k}
\\
&
-\frac{1}{r^3}\Big(
\mathcal{D}_k \phi_{[ij]r}
+\mathcal{D}_i \phi_{[jk]r}
-\mathcal{D}_j \phi_{[ik]r}
\Big)
+\frac{1}{r^3}\Big(
\mathcal{D}_k \phi_{[ij]u}
+\mathcal{D}_i \phi_{[jk]u}
-\mathcal{D}_j \phi_{[ik]u}
\Big)=0 \, .
\end{aligned}
\end{equation}

\section{Hodge-like decomposition for $(p,q)$-mixed symmetry tensors}
In this section we are going to generalize the Hodge decomposition to a Young projected $(p,q)$-biform, i.e. a mixed symmetry tensors assoaiced with Young diagram $\lambda=(p,q)$. We refer to \cite{Manzoni:2025zmi} for the basic definitions. The context is the generalization of the de Rham complex (for differential forms) to multi-forms and, after Young projection, to differential tensors of mixed symmetry.

Let $M$ be a smooth $D$-dimensional manifold (with pseudo-riemannian structure and metric $g$ when Hodge operators are used).
For integers $p_1,\dots,p_N \in \{0,\dots,D\}$, the $N$-multi-form space is
\begin{equation}\label{eq:Nmultiform-space}
\Omega^{p_1\otimes\cdots\otimes p_N}(M)
:= \Omega^{p_1}(M)\otimes\cdots\otimes\Omega^{p_N}(M)\, .
\end{equation}
An element $T\in\Omega^{p_1\otimes\cdots\otimes p_N}(M)$ has local components
$T_{[\mu^{(1)}_1\cdots\mu^{(1)}_{p_1}]\cdots[\mu^{(N)}_1\cdots\mu^{(N)}_{p_N}]}$
antisymmetric within each block of indices. In general, $\Omega^{p_1\otimes\cdots\otimes p_N}(M)$ carries a reducible representation of $\mathrm{GL}(D)$.
Given a Young diagram $\lambda$ compatible with the degrees $(p_1,\dots,p_N)$, one extracts the corresponding irreducible mixed-symmetry sector via the Young projector.

For each slot $i\in\{1,\dots,N\}$, the \emph{$i$-th differential} is the ordinary de Rham differential acting only on the $i$-th factor:
\begin{equation}\label{eq:ith-differential}
d^{(i)}:\Omega^{p_1\otimes\cdots\otimes p_i\otimes\cdots\otimes p_N}(M)
\longrightarrow
\Omega^{p_1\otimes\cdots\otimes (p_i+1)\otimes\cdots\otimes p_N}(M)\, .
\end{equation}
They satisfy the standard relations
\begin{equation}\label{eq:commuting-differentials}
d^{(i)}\circ d^{(i)}=0\quad\forall i\, ,\qquad
d^{(i)}\circ d^{(j)}=d^{(j)}\circ d^{(i)}\quad\forall i\neq j\, .
\end{equation}

Given a mixed-symmetry gauge field $B$ (understood as Young projected), one defines the $i$-th (partial) field strength
\begin{equation}\label{eq:partial-field-strength}
H^{(i)}:=\Pi_{\lambda^{(i)}}\bigl(d^{(i)}B\bigr)\, ,
\end{equation}
whenever the Young projection yields a well-defined tableau $\lambda^{(i)}$ (obtained by adding one box in the $i$-th block).
These partial field strengths are generally not fully gauge invariant but are useful intermediate objects whose appropriate combination with traces can give rise to gauge invariant objects.

The unambiguous gauge-invariant field strength is obtained by composing all slot differentials.
Define the de Rham-like differential
\begin{equation}\label{eq:deRham-like-differential}
\delta^{(N)}:= d^{(N)}\circ d^{(N-1)}\circ\cdots\circ d^{(2)}\circ d^{(1)}\, \qquad \delta^{(N)}\circ\delta^{(N)}=0\,.
\end{equation}
The field strength of $B$ is then
\begin{equation}\label{eq:full-field-strength}
H:=\Pi_{\lambda^{+}}\bigl(\delta^{(N)}B\bigr)\in
\Omega^{(p_1+1)\otimes\cdots\otimes(p_N+1)}(M)\, ,
\end{equation}
where $\lambda^{+}$ denotes the tableau obtained by increasing each block degree by one (when meaningful).
For $N=1$ one recovers the usual de Rham differential and the ordinary curvature $H=dB$.

If $(M,g)$ is oriented and equipped with a metric, one defines the $i$-th Hodge morphism acting only on the $i$-th slot:
\begin{equation}\label{eq:ith-hodge}
\star^{(i)}:\Omega^{p_1\otimes\cdots\otimes p_i\otimes\cdots\otimes p_N}(M)
\longrightarrow
\Omega^{p_1\otimes\cdots\otimes (D-p_i)\otimes\cdots\otimes p_N}(M)\, .
\end{equation}
The full multi-form Hodge operator is the composition
\begin{equation}\label{eq:full-hodge}
\star := \star^{(N)}\circ\star^{(N-1)}\circ\cdots\circ\star^{(1)}\, .
\end{equation}

An $N$-de Rham-like complex is the cochain complex built from multi-form space with differential $\delta^{(N)}$ and we can define the associated cohomology that turn out to be isomorphic to de Rham cohomology.

In classical Hodge theory on an oriented Riemannian manifold $(M^{D},g)$, the codifferential $d^{\dagger}$ is defined as the formal adjoint of the de Rham differential $d$ with respect to the $L^{2}$ inner product induced by $g$ and the Riemannian volume form. The same logic extends to spaces of $N$-multi-forms.

A natural pointwise inner product is obtained by contracting with $g$ in each antisymmetric block:
\begin{equation}
\inner{A}{B}_{x}
:=\frac{1}{p_{1}!\cdots p_{N}!}\,
A_{\mu^{(1)}_{1}\cdots\mu^{(1)}_{p_{1}}\cdots\mu^{(N)}_{1}\cdots\mu^{(N)}_{p_{N}}}\,
B^{\mu^{(1)}_{1}\cdots\mu^{(1)}_{p_{1}}\cdots\mu^{(N)}_{1}\cdots\mu^{(N)}_{p_{N}}}\, ,
\label{prod}
\end{equation}
and the $L^{2}$ pairing, that defines an $L^2$-norm, is
\begin{equation}
\Ltwo{A}{B}:=\int_{M}\inner{A}{B}\,\vol_{g}\, .
\label{prodt}
\end{equation}
One can think of $\inner{A}{B}_x\,\vol_{g}$ as the natural analogue of $\alpha\wedge *\beta$:
the product is obtained by wedging in each slot and then extracting the scalar density against $\vol_{g}$ exacly as in the Hodge theory.

We define the codifferential associated with $\delta^{(N)}$ as the adjoint $\big(\delta^{(N)}\big)^{\dagger}$ such that 
\begin{equation}
    \Ltwo{\delta^{(N)}A}{B}=\Ltwo{A}{\big(\delta^{(N)}\big)^{\dagger}B}\, .
\end{equation}
Since $\delta^{(N)}$ is a composition we have
\[
\big(\delta^{(N)}\big)^{\dagger}
=\bigl(d^{(1)}\bigr)^{\dagger}\circ\bigl(d^{(2)}\bigr)^{\dagger}\circ\cdots\circ\bigl(d^{(N)}\bigr)^{\dagger}\, .
\]
For each slot $i$, the adjoint of $d^{(i)}$ can be written as
\[
\bigl(d^{(i)}\bigr)^{\dagger}B
=
(-1)^{D(p_{i}+1)+1}s\,\star^{(i)}\, d^{(i)}\, \star^{(i)}B\, ,
\]
when $B$ has $i$-th degree $p_{i}$.
This is the direct analogue of the classical identity $d^{\dagger}=\pm *d*$ for ordinary forms,
applied only to the $i$-th factor. Because the slotwise operators on distinct slots commute, one can collect all stars into the total star and obtain
\[
\big(\delta^{(N)}\big)^{\dagger}B
=
(-1)^{\sum_{i=1}^{N}\left(D(p_{i}+1)+1\right)}s^N\,
\star\,\delta^{(N)}\,\star\,B\, .
\]

We now restrict ourself to $(p,q)$-biforms, and relative mixed symmetry tensors, on a compact riemannian manifold without boundary. We may define the the laplacian 
\begin{equation}
\Delta_{\delta^{(2)}}:=\delta^{(2)}\big(\delta^{(2)}\big)^{\dagger}+\big(\delta^{(2)}\big)^{\dagger}\delta^{(2)}
\label{lap+}
\end{equation}
that is of fourth order; however a better choice in order to have an elliptical operator is to define the laplacian as 
\begin{equation}
\widetilde{\Delta} \;:=\; \Delta^{2}+\Delta_{\delta^{(2)}}\ .
\end{equation}
where 
\begin{equation}
\Delta:= d^{(1)}(d^{(1)})^{\dagger}+(d^{(1)})^{\dagger}d^{(1)}+d^{(2)}(d^{(2)})^{\dagger}+(d^{(2)})^{\dagger}d^{(2)}=\Delta_1 + \Delta_2\ .
\label{lap1}
\end{equation}
Indeed even if $\Delta_{\delta^{(2)}}$ is not elliptic, the sum $\widetilde{\Delta}$ is, let us see why in the next proposition.

\begin{Proposition}\label{prop:ellipticity-Delta-tilde}
Let
\begin{equation*}\label{eq:Delta-tilde}
\widetilde{\Delta}\;:=\;\Delta^{2}+\Delta_{\delta^{(2)}}
\end{equation*}
the fourth-order operator where $\Delta$ and $\Delta_{\delta^{(2)}}$ are defined in \eqref{lap1} and \eqref{lap+}. Then $\widetilde{\Delta}$ is elliptic on $\Omega^{p\otimes q}(M)$ and so is on a Young projected subspace.
\end{Proposition}

\begin{proof}
Let $\xi\in T_x^*M$.
The principal symbols of the first-order operators are
\[
\sigma_1\!\bigl(d^{(i)}\bigr)(x,\xi)=\varepsilon_i(\xi):=\xi\wedge_{(i)},\qquad
\sigma_1\!\bigl((d^{(i)})^\dagger\bigr)(x,\xi)=\iota_i(\xi):=\iota^{(i)}_{\xi^\sharp}.
\]
It follows that the principal symbol of each $\Delta^{(i)}$ is
\begin{equation*}\label{eq:principal-symbol-slot-laplacian}
\sigma_2\!\bigl(\Delta^{(i)}\bigr)(x,\xi)=|\xi|^2\,\mathrm{Id},
\quad\Rightarrow\quad
\sigma_2(\Delta)(x,\xi)=2|\xi|^2\,\mathrm{Id},
\end{equation*}
and therefore
\begin{equation*}\label{eq:principal-symbol-Delta-square}
\sigma_4\!\bigl(\Delta^2\bigr)(x,\xi)=\bigl(2|\xi|^2\bigr)^2\,\mathrm{Id}=4|\xi|^4\,\mathrm{Id}.
\end{equation*}
Moreover,
\[
\sigma_2\!\bigl(\delta^{(2)}\bigr)(x,\xi)=\varepsilon_2(\xi)\varepsilon_1(\xi),
\qquad
\sigma_2\!\bigl((\delta^{(2)})^\dagger\bigr)(x,\xi)=\iota_1(\xi)\iota_2(\xi),
\]
so $\sigma_4(\Delta_{\delta^{(2)}})(x,\xi)$ is positive semidefinite but generally not invertible since
\[
\sigma_4\!\bigl(\delta^{(2)}(\delta^{(2)})^\dagger\bigr)(x,\xi)
=\varepsilon_2\varepsilon_1\iota_1\iota_2
=|\xi|^4\,\Pi^{(1)}_{+}(\xi)\,\Pi^{(2)}_{+}(\xi),
\]
\[
\sigma_4\!\bigl((\delta^{(2)})^\dagger\delta^{(2)}\bigr)(x,\xi)
=\iota_1\iota_2\varepsilon_2\varepsilon_1
=|\xi|^4\,\Pi^{(1)}_{-}(\xi)\,\Pi^{(2)}_{-}(\xi).
\]
where
\[
\Pi^{(i)}_{+}(\xi):=\frac{1}{|\xi|^2}\,\varepsilon_i(\xi)\,\iota_i(\xi),\qquad
\Pi^{(i)}_{-}(\xi):=\frac{1}{|\xi|^2}\,\iota_i(\xi)\,\varepsilon_i(\xi).
\]
Therefore
\begin{equation}\label{eq:symbol-Delta-delta2}
\sigma_4\!\bigl(\Delta_{\delta^{(2)}}\bigr)(x,\xi)
=|\xi|^4\Bigl(\Pi^{(1)}_{+}(\xi)\Pi^{(2)}_{+}(\xi)+\Pi^{(1)}_{-}(\xi)\Pi^{(2)}_{-}(\xi)\Bigr),
\end{equation}
which has kernel on $\Omega^{p_1 \otimes p_2}(M)$.
In particular, $\Delta_{\delta^{(2)}}$ is not elliptic by itself.

For $\xi\neq 0$, the principal symbol of $\widetilde{\Delta}$ is the sum
\[
\sigma_4(\widetilde{\Delta})(x,\xi)
=
4|\xi|^4\,\mathrm{Id}+\sigma_4(\Delta_{\delta^{(2)}})(x,\xi);
\]
moreover,
\begin{itemize}
    \item since $\Delta$ is a second-order elliptic operator, its symbol $\sigma_{\Delta}(x, \xi)$ is strictly positive definite. Consequently, the symbol of $\Delta^2$ is also strictly positive definite:
    \begin{equation}
    \langle \sigma_{\Delta^2}(x, \xi) v, v \rangle > 0, \quad \forall v \neq 0, \, \forall \xi \neq 0.
    \end{equation}
    
    \item the operator $\Delta_{\delta^{(2)}}$ is non-elliptic, meaning its symbol $\sigma_{\Delta_{\delta^{(2)}}}(x, \xi)$ may have a non-trivial kernel in certain directions. However, its algebraic structure (it is defined by the Hodge type structure) ensures it is semi-definite positive:
    \begin{equation}
    \langle \sigma_{\Delta_{\delta^{(2)}}}(x, \xi) v, v \rangle \geq 0, \quad \forall v \neq 0, \, \forall \xi \neq 0.
    \end{equation}
\end{itemize}
Therefore,
\begin{equation}
\langle \sigma_{\widetilde{\Delta}}(x, \xi) v, v \rangle = \langle \sigma_{\Delta^2}(x, \xi) v, v \rangle + \langle \sigma_{\Delta_{\delta^{(2)}}}(x, \xi) v, v \rangle > 0, \quad \forall v \neq 0, \, \forall \xi \neq 0;
\end{equation}
since $\sigma_{\widetilde{\Delta}}(x, \xi)$ is strictly positive definite, all its eigenvalues are strictly positive, ensuring it is invertible. This proves that $\widetilde{\Delta}$ is an elliptic operator.
\end{proof}

Let $T\in \Omega^{p,q}(M)$, so
\begin{equation}
\begin{aligned}
&\Ltwo{\widetilde{\Delta}T}{T}=\Ltwo{{\Delta^2}T}{T}+\Ltwo{{\Delta_{\delta^{(2)}}}T}{T};
\end{aligned}
\end{equation}
using the definitions of codifferentials and that $\Delta$ is selfadjoint we get
\begin{equation}
\begin{aligned}
    &\big\langle \Delta_{\delta^{(2)}} T , T \big\rangle
=\big\langle \delta^{(2)}(\delta^{(2)})^{\dagger}T , T \big\rangle
+\big\langle (\delta^{(2)})^{\dagger}\delta^{(2)}T , T \big\rangle
=\|(\delta^{(2)})^{\dagger}T\|^{2}+\|\delta^{(2)}T\|^{2}\\
&\langle \Delta^{2}T, T\rangle = \langle \Delta T, \Delta T\rangle = \|\Delta T\|^{2}.
\end{aligned}
\end{equation}
Therefore,
\[
T\in\ker(\widetilde{\Delta})
\quad\Leftrightarrow\quad
\Delta T=0,\ \ \delta^{(2)}T=0,\ \ (\delta^{(2)})^{\dagger}T=0,
\]
equivalently
\begin{equation}
\begin{aligned}
\ker(\widetilde{\Delta})
&=\ker(\Delta)\cap \ker\big(\delta^{(2)}\big)\cap \ker\big((\delta^{(2)})^{\dagger}\big)=\\
&= \ker\big(d^{(1)}\big)\cap \ker\big((d^{(1)})^{\dagger}\big)\cap \ker\big(d^{(2)}\big)\cap \ker\big((d^{(2)})^{\dagger}\big)\cap \ker\big(\delta^{(2)}\big)\cap \ker\big((\delta^{(2)})^{\dagger}\big),
\end{aligned}
\end{equation}
where we used that $\text{im}(d^{(1)}) \perp \text{im}(d^{(2)})$ with respect to the product \eqref{prod} and similar for codifferentials. Since if $T \in \ker(\Delta)$ then $T \in \ker\big(\delta^{(2)}\big)\cap \ker\big((\delta^{(2)})^{\dagger}\big)$, this is because $\delta^{(2)}$ and its codifferential are composition of $d^{(1)},d^{(2)}$ and their codifferentials, it follows that $\ker(\Delta) \subset \ker({\Delta}_{\delta^{(2)}})$ and so $\ker(\widetilde{\Delta}) = \ker(\Delta) \cap \ker({\Delta}_{\delta^{(2)}})= \ker({\Delta})$.\\
By ellipticity of $\widetilde{\Delta}$ and compactness of $M$ the kernel is finite-dimensional and $\text{im}(\widetilde{\Delta})=(\text{ker}(\widetilde{\Delta}))^{\perp}$. We define the space of 2-harmonic $(p,q)$-biforms as
\begin{equation}
    \mathcal{H}^{p \otimes q}(M):=\{T \in \Omega^{p \otimes q}(M) \ | \ \widetilde{\Delta}T=0\}.
\end{equation}
\begin{Theorem}
    If $M$ is compact then 
    $$ \mathcal{H}^{p \otimes q}(M) \cong \mathcal{H}^{p} \otimes \mathcal{H}^{q}$$ where $\mathcal{H}^k$ is the standard Hodge space of harmonic $k$-forms.
\end{Theorem}
\begin{proof}
    Since $M$ is compact $ \mathcal{H}^{p \otimes q}(M)$ is finite-dimensional. Moreover, since  $\ker(\widetilde{\Delta}) = \ker(\Delta) \cap \ker({\Delta}_{\delta^{(2)}})= \ker({\Delta})$ we have
    \[\mathcal{H}^{p \otimes q}(M):=\{T \in \Omega^{p \otimes q}(M) \ | \ {\Delta}T=0\}=\text{ker}(\Delta).\] Since $\Delta=\Delta_1+\Delta_2$ it follows that $\text{ker}(\Delta)=\text{ker}(\Delta_1)\cap\text{ker}(\Delta_2)$ and 
    \begin{equation*}
    \begin{aligned}
        \text{ker}(\Delta_1)&=\{T \in \Omega^{p \otimes q}(M) \ | \ T \in \mathcal{H}^p \otimes \Omega^q(M)\},\\
        \text{ker}(\Delta_2)&=\{T \in \Omega^{p \otimes q}(M) \ | \ T \in \Omega^p(M) \otimes \mathcal{H}^q\};\\
     \end{aligned}    
    \end{equation*}
     so \[\text{ker}(\Delta) = (\mathcal{H}^p \otimes \Omega^q(M)) \cap (\Omega^p(M) \otimes \mathcal{H}^q)\cong \mathcal{H}^p \otimes \mathcal{H}^q\]
\end{proof}

In the end, the space $\Omega^{p\otimes q}(M)$ can be decomposed as
\begin{equation}
    \Omega^{p \otimes q}(M) \cong  \text{im}(\widetilde{\Delta}) \oplus \mathcal{H}^{p \otimes q}(M);
\end{equation}
explicitly if $T\in \Omega^{p \otimes q}$ we have
\begin{equation}\label{eq:hodge_like_biform_9terms}
\begin{aligned}
T \;=\;& T_{\text{harm}}
+ d^{(1)}A + \bigl(d^{(1)}\bigr)^{\dagger}B
+ d^{(2)}C + \bigl(d^{(2)}\bigr)^{\dagger}D+\\
&\quad
+ d^{(1)}d^{(2)}\phi
+ d^{(1)}\bigl(d^{(2)}\bigr)^{\dagger}\Psi
+ \bigl(d^{(1)}\bigr)^{\dagger}d^{(2)}\Xi
+ \bigl(d^{(1)}\bigr)^{\dagger}\bigl(d^{(2)}\bigr)^{\dagger}\Theta \, ,
\end{aligned}
\end{equation}
where the harmonic part \(T_{\text{harm}}\) satisfies
\begin{equation}\label{eq:harmonic_part_biform}
T_{\text{harm}}\in\mathcal H^{p\otimes q}(M)=\ker(\widetilde\Delta)=\ker(\Delta)
\cong \mathcal H^{p}(M)\otimes\mathcal H^{q}(M)\, ,
\end{equation}
and the potentials have the following bidegrees:
\begin{equation}\label{eq:degrees_potentials_biform}
\begin{gathered}
A\in\Omega^{(p-1)\otimes q}(M),\qquad
B\in\Omega^{(p+1)\otimes q}(M),\\
C\in\Omega^{p\otimes (q-1)}(M),\qquad
D\in\Omega^{p\otimes (q+1)}(M),\\
\phi\in\Omega^{(p-1)\otimes (q-1)}(M),\qquad
\Psi\in\Omega^{(p-1)\otimes (q+1)}(M),\\
\Xi\in\Omega^{(p+1)\otimes (q-1)}(M),\qquad
\Theta\in\Omega^{(p+1)\otimes (q+1)}(M).
\end{gathered}
\end{equation}
This is the general decomposition for a $(p,q)$-biform; however for a $(p,q)$ mixed symmetry tensor the potentials has to be in appropriate representations, meaning that the Young diagram associated to them has to be admissible. In particular this means that the decomposition for a $(p,p)$ mixed symmetry tensor reduces to 
\begin{equation}\label{eq:hodge_like_mixed_9terms}
\begin{aligned}
T \;=\;& T_{\text{harm}} + \bigl(d^{(1)}\bigr)^{\dagger}B
+ d^{(2)}C + \delta^{(2)}\phi
+ \bigl(d^{(1)}\bigr)^{\dagger}d^{(2)}\Xi
+ \bigl(\delta^{(2)}\bigr)^{\dagger}\Theta \, ;
\end{aligned}
\end{equation} 
furthermore, some potential could be vanishing once we require to have the right indices symmetries to appear in the decomposition.

\section{Useful commutation relations and asymptotic gauge fixing}\label{AppD}
% Goal: show that S_{ij} := \nabla^{k}\nabla^{\ell} Y_{k i \ell j} is symmetric on (S^3,g).
% Assumptions on Y_{k i \ell j}: purely algebraic Riemann symmetries:
%  (i)  Y_{k i \ell j} = - Y_{i k \ell j},
%  (ii) Y_{k i \ell j} = - Y_{k i j \ell},
%  (iii)Y_{k i \ell j} =  Y_{\ell j k i},
%  (iv) first Bianchi: Y_{k i \ell j} + Y_{i \ell k j} + Y_{\ell k i j} = 0.
% Connection: Levi--Civita on S^3. Curvature: constant sectional curvature K.

\subsection{Vanishing of the antisymmetric part of $S_{ij}$}\label{s}

Using the pair-exchange symmetry \((k i)\leftrightarrow(\ell j)\),
\[
Y_{k j \ell i} = Y_{\ell i k j},
\]
we compute
\[
S_{ji}
= \nabla^{k}\nabla^{\ell} Y_{k j \ell i}
= \nabla^{k}\nabla^{\ell} Y_{\ell i k j}
= \nabla^{\ell}\nabla^{k} Y_{k i \ell j},
\]
where in the last step we only relabeled the dummy indices \(k \leftrightarrow \ell\).
Hence
\[
S_{ij} - S_{ji}
= \bigl(\nabla^{k}\nabla^{\ell} - \nabla^{\ell}\nabla^{k}\bigr) Y_{k i \ell j}
= [\nabla^{k},\nabla^{\ell}]\, Y_{k i \ell j}.
\]

For any \((0,4)\)-tensor \(T_{k i \ell j}\),
\[
[\nabla_{a},\nabla_{b}]\,T_{k i \ell j}
= - R^{m}{}_{k a b}\,T_{m i \ell j}
  - R^{m}{}_{i a b}\,T_{k m \ell j}
  - R^{m}{}_{\ell a b}\,T_{k i m j}
  - R^{m}{}_{j a b}\,T_{k i \ell m}.
\]
Applying this to \(T=Y\) and contracting with \(g^{ka} g^{\ell b}\) yields
\[
S_{ij} - S_{ji}
= g^{ka} g^{\ell b} [\nabla_{a},\nabla_{b}]\,Y_{k i \ell j}
= T_1 + T_2 + T_3 + T_4,
\]
where
\[
\begin{aligned}
T_1 &\coloneqq - g^{ka} g^{\ell b} R^{m}{}_{k a b}\,Y_{m i \ell j},\\
T_2 &\coloneqq - g^{ka} g^{\ell b} R^{m}{}_{i a b}\,Y_{k m \ell j},\\
T_3 &\coloneqq - g^{ka} g^{\ell b} R^{m}{}_{\ell a b}\,Y_{k i m j},\\
T_4 &\coloneqq - g^{ka} g^{\ell b} R^{m}{}_{j a b}\,Y_{k i \ell m}.
\end{aligned}
\]

On \(S^3\) (a space form) the Riemann tensor is
\[
R_{a b c d} = K\,(g_{a c} g_{b d} - g_{a d} g_{b c}).
\]
Raising the first index,
\[
R^{m}{}_{k a b}
= g^{m c} R_{c k a b}
= K\,(\delta^{m}_{a} g_{k b} - \delta^{m}_{b} g_{k a}).
\]

Define the ``Ricci contraction'' of \(Y\) by
\[
y_{ij} \coloneqq Y^{k}{}_{i k j} = g^{k\ell} Y_{k i \ell j}.
\]

\paragraph{Computation of \(T_1\).}
\[
\begin{aligned}
g^{ka} g^{\ell b} R^{m}{}_{k a b}
&= K\Bigl(g^{ka} g^{\ell b}\delta^{m}_{a} g_{k b} - g^{ka} g^{\ell b}\delta^{m}_{b} g_{k a}\Bigr) \\
&= K\Bigl(g^{km} g^{\ell b} g_{k b} - g^{\ell m} g^{ka} g_{k a}\Bigr) \\
&= K\Bigl(g^{km}\delta^{\ell}_{k} - g^{\ell m}\delta^{k}_{k}\Bigr)
= K\,(g^{\ell m} - 3 g^{\ell m})
= -2K\, g^{\ell m}.
\end{aligned}
\]
Therefore
\[
T_1
= -(-2K g^{\ell m}) Y_{m i \ell j}
= 2K\, g^{\ell m} Y_{m i \ell j}
= 2K\, y_{ij}.
\]

\paragraph{Computation of \(T_2\).}
Using \(R^{m}{}_{i a b} = K(\delta^{m}_{a} g_{i b} - \delta^{m}_{b} g_{i a})\), we get
\[
\begin{aligned}
g^{ka} g^{\ell b} R^{m}{}_{i a b} Y_{k m \ell j}
&= K\Bigl(g^{ka} g^{\ell b}\delta^{m}_{a} g_{i b} Y_{k m \ell j}
        - g^{ka} g^{\ell b}\delta^{m}_{b} g_{i a} Y_{k m \ell j}\Bigr) \\
&= K\Bigl(g^{km} \delta^{\ell}_{i} Y_{k m \ell j}
        - g^{\ell m} \delta^{k}_{i} Y_{k m \ell j}\Bigr) \\
&= K\Bigl(g^{km} Y_{k m i j} - g^{\ell m} Y_{i m \ell j}\Bigr).
\end{aligned}
\]
The first term vanishes since \(Y_{k m i j}\) is antisymmetric in \((k,m)\) while \(g^{km}\) is symmetric:
\[
g^{km} Y_{k m i j} = 0.
\]
For the second term, using \(Y_{i m \ell j} = -Y_{m i \ell j}\),
\[
- g^{\ell m} Y_{i m \ell j}
= - g^{\ell m}(-Y_{m i \ell j})
= g^{\ell m} Y_{m i \ell j}
= y_{ij}.
\]
Hence
\[
g^{ka} g^{\ell b} R^{m}{}_{i a b} Y_{k m \ell j} = K\, y_{ij},
\qquad
T_2 = -K\, y_{ij}.
\]

\paragraph{Computation of \(T_3\).}
Using \(R^{m}{}_{\ell a b} = K(\delta^{m}_{a} g_{\ell b} - \delta^{m}_{b} g_{\ell a})\), first contract with \(g^{\ell b}\):
\[
\begin{aligned}
g^{\ell b} R^{m}{}_{\ell a b}
&= K\Bigl(g^{\ell b}\delta^{m}_{a} g_{\ell b} - g^{\ell b}\delta^{m}_{b} g_{\ell a}\Bigr) \\
&= K\Bigl(\delta^{m}_{a}\, g^{\ell b} g_{\ell b} - g^{\ell m} g_{\ell a}\Bigr)
= K\Bigl(3\delta^{m}_{a} - \delta^{m}_{a}\Bigr)
= 2K\,\delta^{m}_{a}.
\end{aligned}
\]
Therefore
\[
g^{ka} g^{\ell b} R^{m}{}_{\ell a b}
= g^{ka}(2K\delta^{m}_{a})
= 2K\, g^{k m},
\]
and thus
\[
T_3
= -(2K g^{k m}) Y_{k i m j}
= -2K\, y_{ij}.
\]

\paragraph{Computation of \(T_4\).}
Similarly,
\[
g^{ka} g^{\ell b} R^{m}{}_{j a b} Y_{k i \ell m}
= K\Bigl(g^{km} Y_{k i j m} - g^{\ell m} Y_{j i \ell m}\Bigr).
\]
The second term vanishes since \(Y_{j i \ell m}\) is antisymmetric in \((\ell,m)\):
\[
g^{\ell m} Y_{j i \ell m} = 0.
\]
For the first term, use pair exchange \(Y_{k i j m} = Y_{j m k i}\) and then antisymmetry in the first pair:
\[
\begin{aligned}
g^{km} Y_{k i j m}
&= g^{km} Y_{j m k i}
= - g^{km} Y_{m j k i}
= - g^{mk} Y_{k j m i}
= - g^{km} Y_{k j m i}
= - y_{ji}.
\end{aligned}
\]
Hence
\[
g^{ka} g^{\ell b} R^{m}{}_{j a b} Y_{k i \ell m} = -K\, y_{ji},
\qquad
T_4 = -(-K y_{ji}) = K\, y_{ji}.
\]

Summing up,
\[
S_{ij} - S_{ji}
= (2K y_{ij}) + (-K y_{ij}) + (-2K y_{ij}) + (K y_{ji})
= K\,(y_{ji} - y_{ij}).
\]

Using the first Bianchi identity,
\[
Y_{k i \ell j} + Y_{i \ell k j} + Y_{\ell k i j} = 0,
\]
and contracting with \(g^{k\ell}\), we obtain
\[
g^{k\ell} Y_{k i \ell j} + g^{k\ell} Y_{i \ell k j} + g^{k\ell} Y_{\ell k i j} = 0.
\]
The first term is \(y_{ij}\).
The third term vanishes since \(Y_{\ell k i j}\) is antisymmetric in \((\ell,k)\):
\[
g^{k\ell} Y_{\ell k i j} = 0.
\]
For the second term, use \(Y_{i \ell k j} = -Y_{\ell i k j}\) and then pair exchange \(Y_{\ell i k j} = Y_{k j \ell i}\):
\[
g^{k\ell} Y_{i \ell k j}
= - g^{k\ell} Y_{\ell i k j}
= - g^{k\ell} Y_{k j \ell i}
= - y_{ji}.
\]
Therefore
\[
y_{ij} - y_{ji} = 0 \quad\Longrightarrow\quad y_{ij} = y_{ji}.
\]
Plugging this into \(S_{ij}-S_{ji} = K(y_{ji}-y_{ij})\) gives
\[
S_{ij} - S_{ji} = 0 \quad\Longrightarrow\quad S_{ij} = S_{ji}.
\]

\subsection{Asymptotic gauge fixing}\label{agfa}
The gauge condition equation reads
\begin{equation}
    0=(D^iH_{ruij})'=D^iH_{ruij}+\delta (D^iH_{ruij})
\end{equation}
where 
\begin{equation}
    \delta(H_{ruij})=-2D_i\partial_j\Lambda_{ru}+2r\gamma_{ij}(\partial_u\Lambda_{ru}+\partial_r\Lambda_{ru}),
\end{equation}
and we have used the gauge-for-gauge \eqref{gfg}.
Therefore, expanding order by order, thee equation for order $\frac{D-4}{2}$ is
\begin{equation}
    D^iH_{ruij}^{\big(\frac{D-4}{2}\big)}-2\Delta\partial_j\Lambda_{ru}^{\big(\frac{D-4}{2}\big)}+2D_j\partial_u\Lambda_{ru}^{\big(\frac{D-2}{2}\big)}-(D-4)D_j\Lambda_{ru}^{\big(\frac{D-4}{2}\big)}-2D_j\bar{\Lambda}_{ru}^{\big(\frac{D-4}{2}\big)}=0;
\end{equation}
that turns out to be compatible and there are enough arbitrary functions to ensure solvability.

\subsection{Other useful commutators}\label{comms}

For tensor $\mathcal{H}^{ij}$ we have
\begin{equation}
[D_{a},D_{b}]\,\mathcal{H}^{ij}
=
R^{i}_{mab}\mathcal{H}^{mj} + R^{j}_{mab}\mathcal{H}^{im} \ .
\label{C1}
\end{equation}
On a 3D space form of sectional curvature $K$ ,
\begin{equation}
R_{abcd}=K\big(\gamma_{ac}\gamma_{bd}-\gamma_{ad}\gamma_{bc}\big) \ , \qquad
R_{ab}=2K\,\gamma_{ab} \ .
\label{C2}
\end{equation}
A direct contraction of \eqref{C1} by $\gamma^{b\ell} $ and by using \eqref{C2} we get
\begin{equation}
[D_{j},D^{\ell}]\,\mathcal{H}^{ij}
=
K\mathcal{H}^{\ell i}+ 2K\mathcal{H}^{i \ell} - K\gamma^{i \ell}\mathcal{H};
\label{C3}
\end{equation}
using $v^{i}:=D_{j}H^{ij}$ we get,
\begin{equation}
D_{j}D^{\ell}\mathcal{H}^{ij}
=
D^{\ell}D_{j}\mathcal{H}^{ij} + [D_{j},D^{\ell}]\mathcal{H}^{ij}
=
D^{\ell}v^{i}+K\mathcal{H}^{\ell i}+ 2K\mathcal{H}^{i \ell} - K\gamma^{i \ell}\mathcal{H}.
\label{comm2}
\end{equation}

Contracting \eqref{C1} on the pair $(i,j)$ we get
\begin{equation}
[D_{i},D_{j}]\,\mathcal{H}^{ij}
=
R_{mj}\mathcal{H}^{mj}-R_{mi}\mathcal{H}^{im}
=0;
\end{equation}
since $D_{i}\mathcal{H}^{ij}=0$ we have $D_{j}D_{i}\mathcal{H}^{ij}=0$, hence
\begin{equation}
0
=
D_{j}D_{i}\mathcal{H}^{ij}
=
D_{i}D_{j}\mathcal{H}^{ij}+[D_{j},D_{i}]\mathcal{H}^{ij}
=
D_{i}D_{j}\mathcal{H}^{ij}=D_{i}v^{i}.
\label{comm3}
\end{equation}

\section*{Conflict of Interest, Funding and Data Availability}
The author declares no conflict of interest. This research received no external funding. No new data were created or analyzed in this study.

\section*{Acknowledgements}
The author wants to thank the life partner M.L.L.; without her, I would not have had the will to complete this work. Thank you for being there every day: I love you.

% End of appendix snippet.

\bibliographystyle{JHEP} 
\bibliography{bibl}

@article{Lee:1990nz,
    author = "Lee, J. and Wald, Robert M.",
    title = "{Local symmetries and constraints}",
    doi = "10.1063/1.528801",
    journal = "J. Math. Phys.",
    volume = "31",
    pages = "725--743",
    year = "1990"
}

@article{Campoleoni:2010zq,
  author = {Campoleoni, Andrea and Fredenhagen, Stefan and Pfenninger, Stephan and Theisen, Stefan},
  title = {Asymptotic symmetries of three-dimensional gravity coupled to higher-spin fields},
  journal = {JHEP},
  volume = {11},
  pages = {007},
  year = {2010},
  doi = {10.1007/JHEP11(2010)007},
  eprint = {1008.4744},
  archivePrefix = {arXiv},
  primaryClass = {hep-th}
}

@article{Henneaux:2010xg,
    author = "Henneaux, Marc and Rey, Soo-Jong",
    title = "{Nonlinear $W_{\infty}$ as Asymptotic Symmetry of Three-Dimensional Higher Spin Anti-de Sitter Gravity}",
    eprint = "1008.4579",
    archivePrefix = "arXiv",
    primaryClass = "hep-th",
    doi = "10.1007/JHEP12(2010)007",
    journal = "JHEP",
    volume = "12",
    pages = "007",
    year = "2010"
}

@article{Campoleoni:2011hg,
    author = "Campoleoni, Andrea and Fredenhagen, Stefan and Pfenninger, Stefan",
    title = "{Asymptotic W-symmetries in three-dimensional higher-spin gauge theories}",
    eprint = "1107.0290",
    archivePrefix = "arXiv",
    primaryClass = "hep-th",
    reportNumber = "AEI-2011-041",
    doi = "10.1007/JHEP09(2011)113",
    journal = "JHEP",
    volume = "09",
    pages = "113",
    year = "2011"
}

@article{Gonzalez:2013oaa,
    author = "Gonzalez, Hernan A. and Matulich, Javier and Pino, Miguel and Troncoso, Ricardo",
    title = "{Asymptotically flat spacetimes in three-dimensional higher spin gravity}",
    eprint = "1307.5651",
    archivePrefix = "arXiv",
    primaryClass = "hep-th",
    reportNumber = "CECS-PHY-13-06",
    doi = "10.1007/JHEP09(2013)016",
    journal = "JHEP",
    volume = "09",
    pages = "016",
    year = "2013"
}

@article{Joung:2017hsi,
    author = "Joung, Euihun and Kim, Jaewon and Kim, Jihun and Rey, Soo-Jong",
    title = "{Asymptotic Symmetries of Colored Gravity in Three Dimensions}",
    eprint = "1712.07744",
    archivePrefix = "arXiv",
    primaryClass = "hep-th",
    doi = "10.1007/JHEP03(2018)104",
    journal = "JHEP",
    volume = "03",
    pages = "104",
    year = "2018"
}

@article{Monjo:2025u13,
  author = {Monjo, Robert},
  title = {Weak mixing angle under $U(1,3)$ colored gravity},
  journal = {JHEP},
  volume = {06},
  pages = {207},
  year = {2025},
  doi = {10.1007/JHEP06(2025)207},
  eprint = {2502.11236},
  archivePrefix = {arXiv},
  primaryClass = {hep-ph}
}

@article{Barnich:2014kra,
    author = "Barnich, Glenn and Oblak, Blagoje",
    title = "{Notes on the BMS group in three dimensions: I. Induced representations}",
    eprint = "1403.5803",
    archivePrefix = "arXiv",
    primaryClass = "hep-th",
    doi = "10.1007/JHEP06(2014)129",
    journal = "JHEP",
    volume = "06",
    pages = "129",
    year = "2014"
}

@article{Bunster:2013oaa,
    author = {Bunster, Claudio and Henneaux, Marc and H{\"o}rtner, Sergio},
    title = "{Twisted Self-Duality for Linearized Gravity in D dimensions}",
    eprint = "1306.1092",
    archivePrefix = "arXiv",
    primaryClass = "hep-th",
    doi = "10.1103/PhysRevD.88.064032",
    journal = "Phys. Rev. D",
    volume = "88",
    number = "6",
    pages = "064032",
    year = "2013"
}

@article{West:2014xma,
  author = {West, Peter},
  title = {Dual gravity and $E_{11}$},
  year = {2014},
  eprint = {1411.0920},
  archivePrefix = {arXiv},
  primaryClass = {hep-th}
}

@article{Hohm:2013pua,
    author = "Hohm, Olaf and Samtleben, Henning",
    title = "{U-duality covariant gravity}",
    eprint = "1307.0509",
    archivePrefix = "arXiv",
    primaryClass = "hep-th",
    reportNumber = "LMU-ASC-45-13",
    doi = "10.1007/JHEP09(2013)080",
    journal = "JHEP",
    volume = "09",
    pages = "080",
    year = "2013"
}

@article{Hohm:2014fxa,
  author = {Hohm, Olaf and Samtleben, Henning},
  title = {Exceptional field theory. III. $E_{8(8)}$},
  journal = {Phys. Rev. D},
  volume = {90},
  pages = {066002},
  year = {2014},
  eprint = {1406.3348},
  archivePrefix = {arXiv},
  primaryClass = {hep-th}
}

@article{Hohm:2013uia,
    author = "Hohm, Olaf and Samtleben, Henning",
    title = "{Exceptional field theory. II. E$_{7(7)}$}",
    eprint = "1312.4542",
    archivePrefix = "arXiv",
    primaryClass = "hep-th",
    reportNumber = "MIT-CTP-4522",
    doi = "10.1103/PhysRevD.89.066017",
    journal = "Phys. Rev. D",
    volume = "89",
    pages = "066017",
    year = "2014"
}

@article{Hohm:2013vpa,
    author = "Hohm, Olaf and Samtleben, Henning",
    title = "{Exceptional Field Theory I: $E_{6(6)}$ covariant Form of M-Theory and Type IIB}",
    eprint = "1312.0614",
    archivePrefix = "arXiv",
    primaryClass = "hep-th",
    reportNumber = "MIT-CTP-4519",
    doi = "10.1103/PhysRevD.89.066016",
    journal = "Phys. Rev. D",
    volume = "89",
    number = "6",
    pages = "066016",
    year = "2014"
}

@article{Ciambelli:2022vot,
    author = "Ciambelli, Luca",
    title = "{From Asymptotic Symmetries to the Corner Proposal}",
    eprint = "2212.13644",
    archivePrefix = "arXiv",
    primaryClass = "hep-th",
    doi = "10.22323/1.435.0002",
    journal = "PoS",
    volume = "Modave2022",
    pages = "002",
    year = "2023"
}

@article{Afshar:2018apx,
    author = "Afshar, Hamid and Esmaeili, Erfan and Sheikh-Jabbari, M. M.",
    title = "{Asymptotic Symmetries in $p$-Form Theories}",
    eprint = "1801.07752",
    archivePrefix = "arXiv",
    primaryClass = "hep-th",
    reportNumber = "IPM-P-2018-002",
    doi = "10.1007/JHEP05(2018)042",
    journal = "JHEP",
    volume = "05",
    pages = "042",
    year = "2018"
}

@article{Campoleoni:2017qot,
    author = "Campoleoni, Andrea and Francia, Dario and Heissenberg, Carlo",
    title = "{Asymptotic Charges at Null Infinity in Any Dimension}",
    eprint = "1712.09591",
    archivePrefix = "arXiv",
    primaryClass = "hep-th",
    doi = "10.3390/universe4030047",
    journal = "Universe",
    volume = "4",
    number = "3",
    pages = "47",
    year = "2018"
}

@article{Labastida:1987kw,
    author = "Labastida, J. M. F.",
    title = "{Massless Particles in Arbitrary Representations of the Lorentz Group}",
    reportNumber = "IASSNS-HEP-87-43",
    doi = "10.1016/0550-3213(89)90490-2",
    journal = "Nucl. Phys. B",
    volume = "322",
    pages = "185--209",
    year = "1989"
}

@article{Ruzziconi:2019pzd,
    author = "Ruzziconi, Romain",
    title = "{Asymptotic Symmetries in the Gauge Fixing Approach and the BMS Group}",
    eprint = "1910.08367",
    archivePrefix = "arXiv",
    primaryClass = "hep-th",
    doi = "10.22323/1.384.0003",
    journal = "PoS",
    volume = "Modave2019",
    pages = "003",
    year = "2020"
}

@article{Curtright:1980yk,
    author = "Curtright, Thomas",
    title = "{GENERALIZED GAUGE FIELDS}",
    reportNumber = "EFI-80-04",
    doi = "10.1016/0370-2693(85)91235-3",
    journal = "Phys. Lett. B",
    volume = "165",
    pages = "304--308",
    year = "1985"
}

@article{Curtright:1980un,
    author = "Curtright, Thomas L.",
    editor = "Durand, Loyal and Pondrom, Lee G.",
    title = "{HIGH SPIN FIELDS}",
    reportNumber = "UFTP-80-16",
    doi = "10.1063/1.2948665",
    journal = "AIP Conf. Proc.",
    volume = "68",
    pages = "985--988",
    year = "1980"
}

@article{Curtright:1980yj,
    author = "Curtright, Thomas L. and Freund, Peter G. O.",
    title = "{MASSIVE DUAL FIELDS}",
    reportNumber = "EFI 80/05-CHICAGO",
    doi = "10.1016/0550-3213(80)90174-1",
    journal = "Nucl. Phys. B",
    volume = "172",
    pages = "413--424",
    year = "1980"
}

@article{Curtright:2019yur,
    author = "Curtright, Thomas L.",
    title = "{Massive Dual Spinless Fields Revisited}",
    eprint = "1907.11530",
    archivePrefix = "arXiv",
    primaryClass = "hep-th",
    doi = "10.1016/j.nuclphysb.2019.114784",
    journal = "Nucl. Phys. B",
    volume = "948",
    pages = "114784",
    year = "2019"
}

@article{sagnotti2012notes,
      title={Notes on Strings and Higher Spins}, 
      author={A. Sagnotti},
      year={2012},
      eprint={1112.4285},
      archivePrefix={arXiv},
      primaryClass={hep-th}
}

@article{Campoleoni:2020ejn,
    author = "Campoleoni, Andrea and Francia, Dario and Heissenberg, Carlo",
    title = "{On asymptotic symmetries in higher dimensions for any spin}",
    eprint = "2011.04420",
    archivePrefix = "arXiv",
    primaryClass = "hep-th",
    reportNumber = "NORDITA 2020-103",
    doi = "10.1007/JHEP12(2020)129",
    journal = "JHEP",
    volume = "12",
    pages = "129",
    year = "2020"
}

@phdthesis{Esmaeili:2020eua,
    author = "Esmaeili, Erfan",
    title = "{$p$-form gauge fields: charges and memories}",
    eprint = "2010.13922",
    archivePrefix = "arXiv",
    primaryClass = "hep-th",
    school = "IPM, Tehran",
    month = "9",
    year = "2020"
}

@article{Bekaert:2006py,
    author = "Bekaert, Xavier and Boulanger, Nicolas",
    title = "{The unitary representations of the Poincar\textbackslash{}'e group in any spacetime dimension}",
    eprint = "hep-th/0611263",
    archivePrefix = "arXiv",
    doi = "10.21468/SciPostPhysLectNotes.30",
    journal = "SciPost Phys. Lect. Notes",
    volume = "30",
    pages = "1",
    year = "2021"
}

@phdthesis{Heissenberg:2019fbn,
    author = "Heissenberg, Carlo",
    title = "{Topics in Asymptotic Symmetries and Infrared Effects}",
    eprint = "1911.12203",
    archivePrefix = "arXiv",
    primaryClass = "hep-th",
    school = "Pisa, Scuola Normale Superiore",
    year = "2019"
}

@article{Francia:2018jtb,
    author = "Francia, Dario and Heissenberg, Carlo",
    title = "{Two-Form Asymptotic Symmetries and Scalar Soft Theorems}",
    eprint = "1810.05634",
    archivePrefix = "arXiv",
    primaryClass = "hep-th",
    doi = "10.1103/PhysRevD.98.105003",
    journal = "Phys. Rev. D",
    volume = "98",
    number = "10",
    pages = "105003",
    year = "2018"
}

@article{Campiglia:2015qka,
    author = "Campiglia, Miguel and Laddha, Alok",
    title = "{Asymptotic symmetries of QED and Weinberg\textquoteright{}s soft photon theorem}",
    eprint = "1505.05346",
    archivePrefix = "arXiv",
    primaryClass = "hep-th",
    doi = "10.1007/JHEP07(2015)115",
    journal = "JHEP",
    volume = "07",
    pages = "115",
    year = "2015"
}

@article{Bondi:1962px,
    author = "Bondi, H. and van der Burg, M. G. J. and Metzner, A. W. K.",
    title = "{Gravitational waves in general relativity. 7. Waves from axisymmetric isolated systems}",
    doi = "10.1098/rspa.1962.0161",
    journal = "Proc. Roy. Soc. Lond. A",
    volume = "269",
    pages = "21--52",
    year = "1962"
}

@article{Sachs:1962wk,
    author = "Sachs, R. K.",
    title = "{Gravitational waves in general relativity. 8. Waves in asymptotically flat space-times}",
    doi = "10.1098/rspa.1962.0206",
    journal = "Proc. Roy. Soc. Lond. A",
    volume = "270",
    pages = "103--126",
    year = "1962"
}

@article{PhysRev.128.2851,
  title = {Asymptotic Symmetries in Gravitational Theory},
  author = {Sachs, R.},
  journal = {Phys. Rev.},
  volume = {128},
  issue = {6},
  pages = {2851--2864},
  numpages = {0},
  year = {1962},
  month = {12},
  publisher = {American Physical Society},
  doi = {10.1103/PhysRev.128.2851},
  url = {https://link.aps.org/doi/10.1103/PhysRev.128.2851}
}

@article{Afshar_2019,
   title={String memory effect},
   volume={2019},
   ISSN={1029-8479},
   url={http://dx.doi.org/10.1007/JHEP02(2019)053},
   DOI={10.1007/jhep02(2019)053},
   number={2},
   journal={Journal of High Energy Physics},
   publisher={Springer Science and Business Media LLC},
   author={Afshar, Hamid and Esmaeili, Erfan and Sheikh-Jabbari, M. M.},
   year={2019},
   month=2 }

@article{Jokela:2019apz,
    author = "Jokela, Niko and Kajantie, K. and Sarkkinen, Miika",
    title = "{Memory effect in Yang-Mills theory}",
    eprint = "1903.10231",
    archivePrefix = "arXiv",
    primaryClass = "hep-th",
    reportNumber = "HIP-2019-09/TH",
    doi = "10.1103/PhysRevD.99.116003",
    journal = "Phys. Rev. D",
    volume = "99",
    number = "11",
    pages = "116003",
    year = "2019"
}

@article{Pate_2017,
   title={Color Memory: A Yang-Mills Analog of Gravitational Wave Memory},
   volume={119},
   ISSN={1079-7114},
   url={http://dx.doi.org/10.1103/PhysRevLett.119.261602},
   DOI={10.1103/physrevlett.119.261602},
   number={26},
   journal={Physical Review Letters},
   publisher={American Physical Society (APS)},
   author={Pate, Monica and Raclariu, Ana-Maria and Strominger, Andrew},
   year={2017},
   month=dec }

@article{Tolish:2016ggo,
    author = "Tolish, Alexander and Wald, Robert M.",
    title = "{Cosmological memory effect}",
    eprint = "1606.04894",
    archivePrefix = "arXiv",
    primaryClass = "gr-qc",
    doi = "10.1103/PhysRevD.94.044009",
    journal = "Phys. Rev. D",
    volume = "94",
    number = "4",
    pages = "044009",
    year = "2016"
}

@article{Kehagias_2016,
   title={BMS in cosmology},
   volume={2016},
   ISSN={1475-7516},
   url={http://dx.doi.org/10.1088/1475-7516/2016/05/059},
   DOI={10.1088/1475-7516/2016/05/059},
   number={05},
   journal={Journal of Cosmology and Astroparticle Physics},
   publisher={IOP Publishing},
   author={Kehagias, A. and Riotto, A.},
   year={2016},
   month=5, pages={059–059} }

@article{Barnich:2010ojg,
    author = "Barnich, Glenn and Troessaert, Cedric",
    editor = "Anagnostopoulos, Konstantinos N. and Bahns, Dorothea and Grosse, Harald and Irges, Nikos and Zoupanos, George",
    title = "{Supertranslations call for superrotations}",
    eprint = "1102.4632",
    archivePrefix = "arXiv",
    primaryClass = "gr-qc",
    reportNumber = "ULB-TH-11-02",
    doi = "10.22323/1.127.0010",
    journal = "PoS",
    volume = "CNCFG2010",
    pages = "010",
    year = "2010"
}

@article{Campiglia_2020,
	doi = {10.1103/physrevd.101.104039},
	url = {https://doi.org/10.1103%2Fphysrevd.101.104039},
	year = 2020,
	month = {5},
	publisher = {American Physical Society ({APS})},
	volume = {101},
	number = {10},
	author = {Miguel Campiglia and Javier Peraza},
	title = {Generalized {BMS} charge algebra},
	journal = {Physical Review D}
}

@article{Freidel_2021weyl,
	doi = {10.1007/jhep07(2021)170},
	url = {https://doi.org/10.1007%2Fjhep07%282021%29170},
	year = 2021,
	month = {7},
	publisher = {Springer Science and Business Media {LLC}
},
	volume = {2021},
	number = {7},
	author = {Laurent Freidel and Roberto Oliveri and Daniele Pranzetti and Simone Speziale},
	title = {The Weyl {BMS} group and Einstein's equations},
	journal = {Journal of High Energy Physics}
}

@article{Pasterski:2015zua,
    author = "Pasterski, Sabrina",
    title = "{Asymptotic Symmetries and Electromagnetic Memory}",
    eprint = "1505.00716",
    archivePrefix = "arXiv",
    primaryClass = "hep-th",
    doi = "10.1007/JHEP09(2017)154",
    journal = "JHEP",
    volume = "09",
    pages = "154",
    year = "2017"
}

@article{Strominger_2014,
   title={Asymptotic symmetries of Yang-Mills theory},
   volume={2014},
   ISSN={1029-8479},
   url={http://dx.doi.org/10.1007/JHEP07(2014)151},
   DOI={10.1007/jhep07(2014)151},
   number={7},
   journal={Journal of High Energy Physics},
   publisher={Springer Science and Business Media LLC},
   author={Strominger, Andrew},
   year={2014},
   month=jul }

@article{Bonga:2020fhx,
    author = "Bonga, B\'eatrice and Prabhu, Kartik",
    title = "{BMS-like symmetries in cosmology}",
    eprint = "2009.01243",
    archivePrefix = "arXiv",
    primaryClass = "gr-qc",
    doi = "10.1103/PhysRevD.102.104043",
    journal = "Phys. Rev. D",
    volume = "102",
    number = "10",
    pages = "104043",
    year = "2020"
}

@article{Ferreira:2016hee,
    author = "Ferreira, Ricardo Z. and Sandora, McCullen and Sloth, Martin S.",
    title = "{Asymptotic Symmetries in de Sitter and Inflationary Spacetimes}",
    eprint = "1609.06318",
    archivePrefix = "arXiv",
    primaryClass = "hep-th",
    doi = "10.1088/1475-7516/2017/04/033",
    journal = "JCAP",
    volume = "04",
    pages = "033",
    year = "2017"
}

@article{Campoleoni:2019ptc,
    author = "Campoleoni, Andrea and Francia, Dario and Heissenberg, Carlo",
    title = "{Electromagnetic and color memory in even dimensions}",
    eprint = "1907.05187",
    archivePrefix = "arXiv",
    primaryClass = "hep-th",
    doi = "10.1103/PhysRevD.100.085015",
    journal = "Phys. Rev. D",
    volume = "100",
    number = "8",
    pages = "085015",
    year = "2019"
}

@article{Campoleoni:2017mbt,
    author = "Campoleoni, Andrea and Francia, Dario and Heissenberg, Carlo",
    title = "{On higher-spin supertranslations and superrotations}",
    eprint = "1703.01351",
    archivePrefix = "arXiv",
    primaryClass = "hep-th",
    doi = "10.1007/JHEP05(2017)120",
    journal = "JHEP",
    volume = "05",
    pages = "120",
    year = "2017"
}

@article{NicolasBoulanger_2003,
doi = {10.1088/1126-6708/2003/06/060},
url = {https://dx.doi.org/10.1088/1126-6708/2003/06/060},
year = {2003},
month = {7},
publisher = {},
volume = {2003},
number = {06},
pages = {060},
author = {Nicolas Boulanger and  Sandrine Cnockaert and  Marc Henneaux},
title = {A note on spin-s duality},
journal = {Journal of High Energy Physics},
}

@article{Compere:2018aar,
    author = "Comp\`ere, Geoffrey and Fiorucci, Adrien",
    title = "{Advanced Lectures on General Relativity}",
    eprint = "1801.07064",
    archivePrefix = "arXiv",
    primaryClass = "hep-th",
    month = "1",
    year = "2018"
}

@book{hamermesh1989group,
  title={Group Theory and Its Application to Physical Problems},
  author={Hamermesh, M.},
  isbn={9780486661810},
  lccn={89023257},
  series={Addison Wesley Series in Physics},
  url={https://books.google.it/books?id=c0o9\_wlCzgcC},
  year={1989},
  publisher={Dover Publications}
}

@article{strominger2018lectures,
      title={Lectures on the Infrared Structure of Gravity and Gauge Theory}, 
      author={Andrew Strominger},
      year={2018},
      eprint={1703.05448},
      archivePrefix={arXiv},
      primaryClass={hep-th}
}

@book{green_schwarz_witten_2012, place={Cambridge}, series={Cambridge Monographs on Mathematical Physics}, title={Superstring Theory: 25th Anniversary Edition}, volume={1}, DOI={10.1017/CBO9781139248563}, publisher={Cambridge University Press}, author={Green, Michael B. and Schwarz, John H. and Witten, Edward}, year={2012}, collection={Cambridge Monographs on Mathematical Physics}}

@book{Green:2012pqa,
    author = "Green, Michael B. and Schwarz, John H. and Witten, Edward",
    title = "{Superstring Theory Vol. 2}: {25th Anniversary Edition}",
    doi = "10.1017/CBO9781139248570",
    publisher = "Cambridge University Press",
    series = "Cambridge Monographs on Mathematical Physics",
    month = "11",
    year = "2012"
}

@book{Polchinski:1998rq,
    author = "Polchinski, J.",
    title = "{String theory. Vol. 1: An introduction to the bosonic string}",
    doi = "10.1017/CBO9780511816079",
    publisher = "Cambridge University Press",
    series = "Cambridge Monographs on Mathematical Physics",
    month = "12",
    year = "2007"
}

@book{Polchinski:1998rr,
    author = "Polchinski, J.",
    title = "{String theory. Vol. 2: Superstring theory and beyond}",
    doi = "10.1017/CBO9780511618123",
    publisher = "Cambridge University Press",
    series = "Cambridge Monographs on Mathematical Physics",
    month = "12",
    year = "2007"
}

@article{vukovic2018,
      title={Higher spin theory}, 
      author={Ivan Vuković},
      year={2018},
      eprint={1809.02179},
      archivePrefix={arXiv},
      primaryClass={hep-th},
      url={https://arxiv.org/abs/1809.02179}, 
}

@article{Hamada_2017,
   title={Soft pion theorem, asymptotic symmetry and new memory effect},
   volume={2017},
   ISSN={1029-8479},
   url={http://dx.doi.org/10.1007/JHEP11(2017)203},
   DOI={10.1007/jhep11(2017)203},
   number={11},
   journal={Journal of High Energy Physics},
   publisher={Springer Science and Business Media LLC},
   author={Hamada, Yuta and Sugishita, Sotaro},
   year={2017},
   month=nov }

@article{witfre,
  title = {Systematics of higher-spin gauge fields},
  author = {de Wit, Bernard and Freedman, Daniel Z.},
  journal = {Phys. Rev. D},
  volume = {21},
  issue = {2},
  pages = {358--367},
  numpages = {0},
  year = {1980}
}

@article{Vasiliev_2004,
   title={Higher spin gauge theories in various dimensions},
   volume={52},
   ISSN={1521-3978},
   url={http://dx.doi.org/10.1002/prop.200410167},
   DOI={10.1002/prop.200410167},
   number={6–7},
   journal={Fortschritte der Physik},
   publisher={Wiley},
   author={Vasiliev, M.A.},
   year={2004},
   month=may, pages={702–717} 
}

@article{Ferrero:2024eva,
    author = "Ferrero, Pietro and Francia, Dario and Heissenberg, Carlo and Romoli, Matteo",
    title = "{Double-copy supertranslations}",
    eprint = "2402.11595",
    archivePrefix = "arXiv",
    primaryClass = "hep-th",
    doi = "10.1103/PhysRevD.110.026009",
    journal = "Phys. Rev. D",
    volume = "110",
    number = "2",
    pages = "026009",
    year = "2024"
}

@article{Romoli:2024hlc,
    author = "Romoli, Matteo",
    title = "{$ \mathcal{O} $(r$^{N}$) two-form asymptotic symmetries and renormalized charges}",
    eprint = "2409.08131",
    archivePrefix = "arXiv",
    primaryClass = "hep-th",
    doi = "10.1007/JHEP12(2024)085",
    journal = "JHEP",
    volume = "12",
    pages = "085",
    year = "2024"
}

@article{Labastida:1986gy,
    author = "Labastida, J. M. F. and Morris, T. R.",
    title = "{MASSLESS MIXED SYMMETRY BOSONIC FREE FIELDS}",
    reportNumber = "PRINT-86-1128 (IAS,PRINCETON)",
    doi = "10.1016/0370-2693(86)90143-7",
    journal = "Phys. Lett. B",
    volume = "180",
    pages = "101--106",
    year = "1986"
}

@article{Hull:2000zn,
    author = "Hull, C. M.",
    title = "{Strongly coupled gravity and duality}",
    eprint = "hep-th/0004195",
    archivePrefix = "arXiv",
    reportNumber = "QMW-00-03",
    doi = "10.1016/S0550-3213(00)00323-0",
    journal = "Nucl. Phys. B",
    volume = "583",
    pages = "237--259",
    year = "2000"
}

@article{Hull:2001iu,
    author = "Hull, C. M.",
    title = "{Duality in gravity and higher spin gauge fields}",
    eprint = "hep-th/0107149",
    archivePrefix = "arXiv",
    reportNumber = "QMUL-PH-01-01",
    doi = "10.1088/1126-6708/2001/09/027",
    journal = "JHEP",
    volume = "09",
    pages = "027",
    year = "2001"
}

@article{Medeiros_2003,
   title={Exotic Tensor Gauge Theory and Duality},
   volume={235},
   ISSN={1432-0916},
   url={http://dx.doi.org/10.1007/s00220-003-0810-z},
   DOI={10.1007/s00220-003-0810-z},
   number={2},
   journal={Communications in Mathematical Physics},
   publisher={Springer Science and Business Media LLC},
   author={Medeiros, P.F. de and Hull, C.M.},
   year={2003},
   month=apr, pages={255–273} }

@article{Manzoni:2024agc,
    author = "Manzoni, Federico",
    title = "{Axialgravisolitons at infinite corner}",
    eprint = "2404.04951",
    archivePrefix = "arXiv",
    primaryClass = "gr-qc",
    doi = "10.1088/1361-6382/ad61b5",
    journal = "Class. Quant. Grav.",
    volume = "41",
    number = "17",
    pages = "177001",
    year = "2024"
}

@article{manz2,
    author = "Francia, Dario and Manzoni, Federico",
    title = "{Asymptotic charges of $p-$forms and their dualities in any $D$}",
    eprint = "2411.04926",
    archivePrefix = "arXiv",
    primaryClass = "hep-th",
    month = "11",
    year = "2024"
}

@article{Romoli:2025map,
    author = "Romoli, Matteo and Manzoni, Federico",
    title = "{Higher-order p-form Asymptotic Symmetries in $D=p+2$}",
    eprint = "2503.22572",
    archivePrefix = "arXiv",
    primaryClass = "hep-th",
    doi = "10.1007/s10773-026-06276-7",
    journal = "Int. J. Theor. Phys.",
    volume = "65",
    number = "3",
    pages = "56",
    year = "2026"
}

@article{Manzoni:2024tow,
    author = "Manzoni, Federico",
    title = "{Duality, asymptotic charges and higher form symmetries in p-form gauge theories}",
    eprint = "2411.05602",
    archivePrefix = "arXiv",
    primaryClass = "math-ph",
    doi = "10.1140/epjc/s10052-026-15359-y",
    journal = "Eur. Phys. J. C",
    volume = "86",
    number = "2",
    pages = "155",
    year = "2026"
}

@article{deAguiarAlves:2025vfu,
    author = "de Aguiar Alves, N{\'\i}ckolas and Landulfo, Andr{\'e} G. S.",
    title = "{Sound as a gauge theory and its infrared triangle}",
    eprint = "2512.15796",
    archivePrefix = "arXiv",
    primaryClass = "hep-th",
    month = "12",
    year = "2025"
}

@article{Manzoni:2025zmi,
    author = "Manzoni, Federico",
    title = "{Duality, asymptotic charges and algebraic topology in mixed symmetry tensor gauge theories and applications}",
    eprint = "2501.05104",
    archivePrefix = "arXiv",
    primaryClass = "math-ph",
    month = "1",
    year = "2025"
}

@phdthesis{Manzoni:2025gxw,
    author = "Manzoni, Federico",
    title = "{Duality and asymptotic symmetries in gravisolitons and gauge theories}",
    URL = "https://hdl.handle.net/11590/505556",
    school = "Universit{\`a} degli Studi di Roma Tre",
    year = "2025"
}

@article{HollandsIshibashi2005BondiEnergyHigherD,
  author       = {Hollands, Stefan and Ishibashi, Akihiro},
  title        = {Asymptotic flatness and Bondi energy in higher dimensional gravity},
  journal      = {Journal of Mathematical Physics},
  volume       = {46},
  number       = {2},
  pages        = {022503},
  year         = {2005},
  doi          = {10.1063/1.1829152},
  eprint       = {gr-qc/0304054},
  archivePrefix= {arXiv},
  primaryClass = {gr-qc}
}

@article{TanabeTanahashiShiromizu2010NullInfinity5D,
  author       = {Tanabe, Kentaro and Tanahashi, Norihiro and Shiromizu, Tetsuya},
  title        = {On asymptotic structure at null infinity in five dimensions},
  journal      = {Journal of Mathematical Physics},
  volume       = {51},
  number       = {6},
  pages        = {062502},
  year         = {2010},
  doi          = {10.1063/1.3429580},
  eprint       = {0909.0426},
  archivePrefix= {arXiv},
  primaryClass = {gr-qc}
}

@article{TanabeKinoshitaShiromizu2011NullInfinityAnyD,
  author       = {Tanabe, Kentaro and Kinoshita, Shunichiro and Shiromizu, Tetsuya},
  title        = {Asymptotic flatness at null infinity in arbitrary dimensions},
  journal      = {Physical Review D},
  volume       = {84},
  pages        = {044055},
  year         = {2011},
  doi          = {10.1103/PhysRevD.84.044055},
  eprint       = {1104.0303},
  archivePrefix= {arXiv},
  primaryClass = {gr-qc}
}

@article{KapecLysovPasterskiStrominger2017HigherDSupertranslations,
  author       = {Kapec, Daniel and Lysov, Vyacheslav and Pasterski, Sabrina and Strominger, Andrew},
  title        = {Higher-dimensional supertranslations and {W}einberg's soft graviton theorem},
  journal      = {Annals of Mathematical Sciences and Applications},
  volume       = {2},
  number       = {1},
  pages        = {69--94},
  year         = {2017},
  doi          = {10.4310/AMSA.2017.v2.n1.a2},
  eprint       = {1502.07644},
  archivePrefix= {arXiv},
  primaryClass = {hep-th}
}

@article{Fuentealba:2022BMS5PRL,
  author        = {Fuentealba, Oscar and Henneaux, Marc and Matulich, Javier and Troessaert, C{\'e}dric},
  title         = {Bondi-Metzner-Sachs Group in Five Spacetime Dimensions},
  journal       = {Physical Review Letters},
  volume        = {128},
  number        = {5},
  pages         = {051103},
  year          = {2022},
  doi           = {10.1103/PhysRevLett.128.051103},
  eprint        = {2111.09664},
  archivePrefix = {arXiv},
  primaryClass  = {hep-th}
}

@article{Fuentealba:2022HamiltonianBMS5,
  author        = {Fuentealba, Oscar and Henneaux, Marc and Matulich, Javier and Troessaert, C{\'e}dric},
  title         = {Asymptotic structure of the gravitational field in five spacetime dimensions: Hamiltonian analysis},
  journal       = {Journal of High Energy Physics},
  number        = {07},
  pages         = {149},
  year          = {2022},
  doi           = {10.1007/JHEP07(2022)149},
  eprint        = {2206.04972},
  archivePrefix = {arXiv},
  primaryClass  = {hep-th}
}

@article{Barnich:2010BMSsuperrotations,
  author        = {Barnich, Glenn and Troessaert, C{\'e}dric},
  title         = {Symmetries of asymptotically flat 4 dimensional spacetimes at null infinity revisited},
  journal       = {Physical Review Letters},
  volume        = {105},
  pages         = {111103},
  year          = {2010},
  doi           = {10.1103/PhysRevLett.105.111103},
  eprint        = {0909.2617},
  archivePrefix = {arXiv},
  primaryClass  = {gr-qc}
}

@article{Barnich:2011ChargeAlgebra,
  author        = {Barnich, Glenn and Troessaert, C{\'e}dric},
  title         = {BMS charge algebra},
  journal       = {Journal of High Energy Physics},
  number        = {12},
  pages         = {105},
  year          = {2011},
  doi           = {10.1007/JHEP12(2011)105},
  eprint        = {1106.0213},
  archivePrefix = {arXiv},
  primaryClass  = {hep-th}
}

@article{Campiglia:2014SubleadingSoft,
  author        = {Campiglia, Miguel and Laddha, Alok},
  title         = {Asymptotic symmetries and subleading soft graviton theorem},
  journal       = {Physical Review D},
  volume        = {90},
  pages         = {124028},
  year          = {2014},
  doi           = {10.1103/PhysRevD.90.124028},
  eprint        = {1408.2228},
  archivePrefix = {arXiv},
  primaryClass  = {hep-th}
}

@article{Campiglia:2015NewSymmetriesSmatrix,
  author        = {Campiglia, Miguel and Laddha, Alok},
  title         = {New symmetries for the Gravitational S-matrix},
  journal       = {Journal of High Energy Physics},
  number        = {04},
  pages         = {076},
  year          = {2015},
  doi           = {10.1007/JHEP04(2015)076},
  eprint        = {1502.02318},
  archivePrefix = {arXiv},
  primaryClass  = {hep-th}
}

@article{Kapec:2014VirasoroSmatrix,
  author        = {Kapec, Daniel and Lysov, Vyacheslav and Pasterski, Sabrina and Strominger, Andrew},
  title         = {Semiclassical Virasoro Symmetry of the Quantum Gravity S-matrix},
  journal       = {Journal of High Energy Physics},
  number        = {08},
  pages         = {058},
  year          = {2014},
  doi           = {10.1007/JHEP08(2014)058},
  eprint        = {1406.3312},
  archivePrefix = {arXiv},
  primaryClass  = {hep-th}
}

@article{Freidel:2021BMSW,
  author        = {Freidel, Laurent and Oliveri, Roberto and Pranzetti, Daniele and Speziale, Simone},
  title         = {The Weyl BMS group and Einstein's equations},
  journal       = {Journal of High Energy Physics},
  number        = {07},
  pages         = {170},
  year          = {2021},
  doi           = {10.1007/JHEP07(2021)170},
  eprint        = {2104.05793},
  archivePrefix = {arXiv},
  primaryClass  = {hep-th}
}

@article{CURTRIGHT1985304,
title = {Generalized gauge fields},
journal = {Physics Letters B},
volume = {165},
number = {4},
pages = {304-308},
year = {1985},
issn = {0370-2693},
doi = {https://doi.org/10.1016/0370-2693(85)91235-3},
url = {https://www.sciencedirect.com/science/article/pii/0370269385912353},
author = {Thomas Curtright}
}

@article{Hull:2024qpy,
    author = {Hull, Chris and Lindstr{\"o}m, Ulf and Vel{\'a}squez Cotini Hutt, Maxwell L.},
    title = "{Gauge-invariant charges of the dual graviton}",
    eprint = "2412.10503",
    archivePrefix = "arXiv",
    primaryClass = "hep-th",
    doi = "10.1007/JHEP02(2025)198",
    journal = "JHEP",
    volume = "02",
    pages = "198",
    year = "2025"
}

@article{Hull:2023iny,
    author = "Hull, C. M.",
    title = "{Magnetic charges for the graviton}",
    eprint = "2310.18441",
    archivePrefix = "arXiv",
    primaryClass = "hep-th",
    reportNumber = "Imperial-TP-2023-CH-03",
    doi = "10.1007/JHEP05(2024)257",
    journal = "JHEP",
    volume = "05",
    pages = "257",
    year = "2024"
}

@ARTICLE{2024CQGra..41s5012H,
       author = {{Hull}, Chris and {Hutt}, Maxwell L. and {Lindstr{\"o}m}, Ulf},
        title = "{Gauge-invariant magnetic charges in linearised gravity}",
      journal = {Classical and Quantum Gravity},
         year = 2024,
        month = oct,
       volume = {41},
          doi = {10.1088/1361-6382/ad718a},
archivePrefix = {arXiv},
       eprint = {2405.08876},
 primaryClass = {hep-th},
}

@article{Hull:2024xgo,
    author = {Hull, Chris and Hutt, Maxwell L. and Lindstr{\"o}m, Ulf},
    title = "{Charges and topology in linearised gravity}",
    eprint = "2401.17361",
    archivePrefix = "arXiv",
    primaryClass = "hep-th",
    reportNumber = "Imperial-TP-2024-CH-01",
    doi = "10.1007/JHEP07(2024)097",
    journal = "JHEP",
    volume = "07",
    pages = "097",
    year = "2024"
}

@article{Hull:2024bcl,
    author = {Hull, Chris and Hutt, Maxwell L. and Lindstr{\"o}m, Ulf},
    title = "{Generalised symmetries in linear gravity}",
    eprint = "2409.00178",
    archivePrefix = "arXiv",
    primaryClass = "hep-th",
    reportNumber = "Imperial-TP-2024-CH-5, UUITP-24/24",
    doi = "10.1007/JHEP04(2025)046",
    journal = "JHEP",
    volume = "04",
    pages = "046",
    year = "2025"
}

@article{Hull:2024ism,
    author = {Hull, Chris and Hutt, Maxwell L. and Lindstr{\"o}m, Ulf},
    title = "{Gauging generalised symmetries in linear gravity}",
    eprint = "2410.08720",
    archivePrefix = "arXiv",
    primaryClass = "hep-th",
    reportNumber = "Imperial-TP-2024-CH-6, UUITP-27/24",
    doi = "10.1007/JHEP01(2025)145",
    journal = "JHEP",
    volume = "01",
    pages = "145",
    year = "2025"
}

\end{document}